%% file: main.tex
\documentclass[12pt]{article}
\usepackage[utf8]{inputenc} % UTF-8 encoding
\usepackage[T1]{fontenc} % Font encoding
\usepackage{amsmath, amsfonts, amssymb, amsthm, mathtools} % For math symbols and environments
\usepackage{geometry} % For customizing margins
\usepackage{multirow} % For table cells spanning multiple rows
\usepackage{xr} % For external references
\usepackage{float}
\usepackage{graphicx} % For including graphics
\graphicspath{{figures_and_tables/figures/}}
\usepackage{tikz} % For creating graphics programmatically
\usetikzlibrary{positioning,fit,shapes}
\usepackage{subcaption}

\usepackage[ruled,vlined, linesnumbered]{algorithm2e} % For writing algorithms
\usepackage{algorithmic} % For algorithmic environment

\usepackage[authoryear]{natbib} % Author-year citation style
\usepackage{bibunits}
\defaultbibliographystyle{ecta}
\defaultbibliography{references}

\numberwithin{equation}{section}

\usepackage{latexsym} % Additional math symbols
\usepackage{enumerate} % For custom lists
\usepackage{enumitem}
\usepackage{setspace} % For line spacing
\usepackage[toc,title,titletoc,header]{appendix} % For appendices
\usepackage[bottom]{footmisc} % Forces footnotes to be at the bottom
\usepackage{cases}
\usepackage[pdfborder={0 0 0}]{hyperref} % For hyperlinks
\hypersetup{
    colorlinks=true,
    linkcolor=red,
    citecolor=purple,
    filecolor=red,
    urlcolor=green
}

\usepackage{libertine} % Serif font
\usepackage[libertine]{newtxmath} % Matching math font

\newtheorem{theorem}{Theorem}
\newtheorem{lemma}[theorem]{Lemma}
\newtheorem{prop}{Proposition}
\newtheorem{cor}{Corollary}
\newtheorem{ass}{Assumption}
\newtheorem{remark}{Remark}
\newtheorem{df}{Definition}
\newtheorem{eg}{Example}

\definecolor{red}{RGB}{168,53,49}
\definecolor{purple}{RGB}{57,49,168}
\definecolor{yellow}{RGB}{240,228,66}
\definecolor{green}{RGB}{0,158,115}
\definecolor{panna}{RGB}{244,239,201}

\newcommand{\E}{\mathbb{E}} % Shortcut for the expectation operator
\newcommand{\ind}{\mathbf{1}} % Shortcut for the indicator function
\newcommand{\R}{\mathbb{R}} % Shortcut for the set of real numbers
\newcommand{\Var}{\operatorname{Var}}
\DeclarePairedDelimiter{\abs}{\lvert}{\rvert}

\begin{document}
\relax
\hypersetup{pageanchor=false}
\hypersetup{pageanchor=true}

\title{\vspace{0cm}\textbf{Complementarities in Sparse Two-Sided Interaction Models\thanks{I am deeply grateful to Federico Bugni, Ivan Canay, and Chuck Manski for their guidance throughout this project. I thank Eric Auerbach, Stéphane Bonhomme, Devis Decet, Anastasiia Evdokimova, Bruno Fava, Danil Fedchenko, Joel Horowitz, Patrick Kline, Rasmus Lentz, Giacomo Marcolin, Amilcar Velez, Zhen Xie, and the participants of several seminars and conferences for their valuable comments and suggestions. This draft replaces a previously circulated version titled ``Identification, Estimation, and Inference in Two-Sided Interaction Models''. All errors are my own.}
}}

\author{
Federico Crippa\\
Department of Economics\\
University of California, Berkeley \\
\small \href{mailto:federico.crippa@berkeley.edu}{federico.crippa@berkeley.edu}
}

\date{ August 6, 2026.
}
\maketitle

\vspace{-1cm}
\thispagestyle{empty}

\begin{spacing}{1.1}
\begin{abstract}

This paper studies complementarities in two-sided interaction models when only a sparse subset of potential matches is realized. I introduce the Tukey model, adding one complementarity parameter to two-way fixed effects. Variation within a four-cycle in the observed matching network identifies this parameter, while connectedness permits identification of agent productivities. I propose a cycle-based estimator that avoids estimating latent productivities. The estimator is consistent and asymptotically normal under sparse-network asymptotics, yielding a formal test of no complementarities. An application to World Bank manager-country assignments finds negative complementarities, suggesting that more capable managers are especially valuable in more complex environments.

\end{abstract}
\end{spacing}

\medskip

%\noindent KEYWORDS: Interactive fixed effects, AKM, complementarities.
%\noindent JEL classification codes: C13, C23.

\thispagestyle{empty}

\newpage
\hypersetup{pageanchor=true}
\setcounter{page}{1}
\begin{bibunit}

\section{Introduction}

In many economic settings, outcomes result from interactions between two distinct types of agents and depend on the latent characteristics each side brings to the match. For example, wages reflect both worker skills and firm attributes \citep{abowd1999high}; corporate performance depends on managerial ability together with company-specific features \citep{bertrand2003managing}; and the productivity of public offices reflects local conditions and the capacity of the bureaucrats in charge \citep{fenizia2022managers}. In such settings, an agent's contribution may depend not only on their own characteristics but also on those of the agent with whom they interact.

The two-way fixed effects (TWFE) model, widely used to study these settings, rules out this possibility. Its additively separable interaction function implies that the marginal contribution of each agent is independent of the characteristics of the other side of the match. The interaction function is therefore modular, and complementarities are excluded by assumption. Yet complementarities are central to many economic questions. Are high-productivity workers especially productive when matched with high-productivity firms? Do more capable managers add greater value in more challenging environments? Does the assignment of agents across matches affect aggregate output? Answering these questions requires a model that separates the contributions of the two sides from the additional value generated by their interaction.

To do so, this paper introduces a framework for modeling two-sided interactions. Any model in this framework has two components: (i) a matching network, where nodes represent agents and edges indicate which pairs are observed, and (ii) an interaction function, which maps the latent characteristics of matched agents into outcomes. The interaction function describes the model's economic structure, while the matching network summarizes the information available in the data. Separating these two components makes explicit the role of the observed matching network in identification. In settings where TWFE models are commonly applied, only a small fraction of all potential matches is observed, and each agent typically interacts with only a few agents on the other side. Which forms of complementarity can be recovered therefore depends on whether these observed matches form sufficiently informative network patterns.

Within this framework, I center the analysis on the Tukey model, named after John Tukey, who proposed an analogous functional form in the context of nonlinear ANOVA \citep{tukey1949one}. The model enriches the TWFE specification by assuming a constant cross-partial derivative of the interaction function, governed by a single scalar parameter. This interaction parameter fully characterizes the presence, direction, and strength of complementarities. Positive and negative values imply positive and negative complementarities, respectively, while a value of zero eliminates complementarities and reduces the model to TWFE. The Tukey model therefore provides a parsimonious but economically meaningful departure from additive separability.

The paper makes three main contributions. First, I characterize identification in the Tukey model. Variation within a four-cycle (a closed path with four edges and four distinct nodes) eliminates the additive components of the interaction function and identifies the interaction parameter. Once this parameter is identified, connectedness identifies the latent productivities. I then use two richer specifications to show that identification requirements become more demanding as the interaction function becomes more flexible. Allowing complementarities to vary across agents or otherwise relaxing restrictions on the interaction function requires network structures that are rarely observed in the sparse settings typical of empirical applications.

These results reveal a broader trade-off between the flexibility of the interaction function and the richness of the matching network. Weaker restrictions on the interaction function require more informative patterns of observed matches. The two sides of this trade-off differ in an important respect: the interaction function is unobserved and must be restricted by assumption, whereas the matching network is observed and its relevant properties can be checked directly. The matching network can therefore discipline the choice of interaction specification by determining which complementarity structures can be identified from the observed pattern of matches. From this perspective, the Tukey model occupies an important middle ground. It allows complementarities while retaining agent-specific latent heterogeneity and requires only a modest strengthening of the connectedness condition used for TWFE. By contrast, substantially richer forms of complementarity may require information that the observed network simply does not contain.

Second, I propose a cycle-based estimator of the Tukey interaction parameter and study its properties under large-graph asymptotics that allow node degrees to remain bounded. Because identification relies on variation within four-cycles, the estimator aggregates information across a collection of edge-disjoint four-cycles to isolate and consistently estimate the interaction parameter. Unlike procedures that estimate this parameter jointly with the latent productivities, the proposed estimator does not require estimating either productivity collection. A key requirement for consistency is a growing number of selected four-cycles. This condition can hold even when node degrees remain bounded; in employer–employee matched data, I show that four-cycles can be abundant. The estimator is asymptotically normal, with its variance depending on the outcome errors, productivity heterogeneity within four-cycles, and the orientation assigned to each cycle. Because arbitrary labels may fail to align the identifying variation, I propose a rank-based orientation rule that uses observable variables informative about the ordering of latent productivities.

Third, I develop a formal test of no complementarities. The test exploits the nesting of TWFE within the Tukey model: because TWFE corresponds to an interaction parameter equal to zero, testing this null provides a test of the modular TWFE restriction within the Tukey class. Empirical applications of TWFE often discuss the assumption of no complementarities and rely on informal diagnostics. Formal procedures for testing additive separability exist in complete two-way layouts \citep{tukey1949one}, but, to my knowledge, no formal test suited to the sparse incomplete matching networks has been developed. The proposed test focuses directly on the cycle variation that distinguishes modular from non-modular interactions in this setting and does not require estimating individual latent productivities, so its asymptotic properties can be studied under sequences that reflect common data structures, including settings in which agents have only a small number of observed links.

To illustrate the empirical value of the Tukey model, I revisit the application in \citet{limodio2021bureaucrat} on the assignment of public managers to tasks in the implementation of World Bank projects. Project success is modeled as the outcome of the interaction between the ability of the manager in charge and the characteristics of the country in which the project is implemented. Estimates from the TWFE model indicate negative sorting: high-performing managers are more likely to be assigned to low-performing countries. The Tukey model provides an additional insight. The interaction parameter is negative and statistically different from zero, indicating negative complementarities: high-performing managers have greater value added in low-productivity countries. This submodular interaction structure is consistent with more capable managers having greater value added in more complex environments and can therefore help rationalize the negative sorting in the assignment of managers to countries.

\subsection{Related Literature}

The class of models studied in this paper is related to a broad literature on two-sided interactions, much of which is reviewed in \citet{bonhomme2020econometric}. The main benchmark is the TWFE model, which has been widely applied to worker-firm interactions \citep{abowd1999high, card2013workplace, kline2024firm}, as well as to interactions between managers and firms \citep{bertrand2003managing}, teachers and students \citep{jackson2014teacher, chetty2014measuring, chetty2014measuring2}, patients and healthcare providers \citep{finkelstein2016sources}, and bureaucrats and geographic postings \citep{fenizia2022managers, limodio2021bureaucrat}. Across these applications, TWFE imposes a modular interaction function and therefore rules out complementarities. This may be restrictive in settings where the value generated by an interaction depends on the particular characteristics of both agents involved.

One approach to relaxing modularity is to introduce complementarities through grouped heterogeneity. \citet{bonhomme2019distributional} propose a model in which agents are partitioned into a finite number of groups, with all agents within a group sharing the same latent characteristics. \citet{lei2023estimating} adopt a similar grouping approach to estimate a low-rank factor model. These models allow richer interaction functions than TWFE and have generated valuable empirical insights \citep{weigel2024super, mourot2025surgeons}, but rely on grouped heterogeneity. The approach developed here instead introduces complementarities while retaining agent-specific latent heterogeneity.

Concerns about complementarities have also surfaced in many TWFE applications. A common diagnostic estimates a saturated specification and compares its $R^2$ with that of the additive model. The typically small improvement in fit has been interpreted as evidence that complementarities are quantitatively unimportant \citep{card2013workplace, song2019firming, fenizia2022managers, adhvaryu2024no}. \citet{kline2024firm}, however, show that such comparisons can be misleading and argue that complementarities are better detected using cycles in the matching network. Building on this insight, I use cycle variation to construct a formal test of the modularity restriction imposed by TWFE.

The interpretation of worker-firm specifications raises a distinct issue. Although worker-firm interactions are the canonical application of TWFE \citep{abowd1999high, card2013workplace, kline2024firm}, interpreting these models as structural wage equations has faced important critiques \citep{eeckhout2011identifying, hagedorn2017identifying, lopes2018firm, eeckhout2018sorting}. The framework studied here does not resolve these concerns. I treat the interaction function as a reduced-form object and study what can be learned about its complementarity structure from the observed matching network. For labor-market applications, Supplementary Material~\ref{sec:sm_microfoundation} provides a simple environment in which the Tukey equation arises from nonseparable production and wage setting, illustrating how the interaction parameter can be interpreted under additional economic assumptions.

This paper is also related to work on nonadditive interactions in balanced unit-time panels. \citet{tukey1949one} introduced a nonadditive interaction in a complete two-way layout, while more recent contributions use factor models or nonparametric interaction structures \citep{bai2009panel, freyberger2018non, freeman2023linear, sbaisassi2024dyadic, armstrong2025robust}. Two features distinguish the present setting. First, much of the panel literature uses flexible interactions to account for unobserved heterogeneity when estimating coefficients on observed regressors, whereas here the complementarity structure and the associated latent productivities are themselves objects of interest. Second, balanced panels observe all unit-time combinations, whereas the matching networks considered here contain only a subset of all possible matches. The observed pattern of matches therefore determines identification and inference in this setting, rather than being fixed by the design of the data.

Finally, the paper contributes to the literature on network econometrics. Like \citet{bramoulle2009identification}, \citet{graham2017econometric}, and \citet{de2018identifying}, the identification analysis imposes conditions on the observed graph. The cycle-based argument eliminates the agent-specific nuisance parameters and identifies the interaction parameter, paralleling the strategy used by \citet{jochmans2017two} for multiplicative models. For inference, I consider growing networks that may remain sparse and in which node degrees may remain bounded. Most existing inference results for sparse networks are developed under asymptotic sequences that allow node degrees to diverge \citep{jochmans2019fixed, cai2022linear}. Results that accommodate bounded-degree graphs, including \citet{verdier2020estimation} and \citet{auerbach2021local}, do not focus on parameters governing interactions. The analysis therefore extends inference on interaction parameters to network structures that resemble many datasets on two-sided interactions.

\subsection{Paper Structure and Notation}

The remainder of the paper is organized as follows. Section~\ref{sec:model} introduces the Bipartite Interaction framework and the Tukey model. Section~\ref{sec:identification} studies identification, while Section~\ref{sec:estimation_inference} develops estimation and inference for the interaction parameter and discusses the estimation of latent productivities. Section~\ref{sec:empirical_illustration} presents an empirical illustration of how the Tukey model can be used in practice and what it reveals about the interaction function. Section~\ref{sec:conclusion} concludes. The Appendix contains proofs of the main text results and material closely related to the main arguments. The Supplementary Material develops the richer interaction specifications used as identification benchmarks and collects auxiliary results.

Throughout the paper, lowercase symbols denote parameters or quantities associated with an individual agent. These objects are scalar unless otherwise stated; for example, $\alpha_i$ denotes the latent productivity of worker $i$, but may be a vector when productivity is multidimensional. Uppercase symbols denote sets; for example, $\alpha_i \in A$, where $A$ is typically a compact metric space. Bold symbols denote collections: for example, $\boldsymbol{\alpha} = (\alpha_1, \dots, \alpha_I)$ is the collection of worker productivities.

\section{BI Framework and Tukey Model} \label{sec:model}

In this section, I introduce the Bipartite Interaction (BI) framework and the Tukey model and briefly describe the richer interaction specifications used as identification benchmarks. I use the labor market as a running example, with workers and firms constituting the two sides of the interaction. The framework applies more generally to any two-sided interaction setting.

\subsection{Bipartite Interaction Framework}

Let $I \in \mathbb{N}$ and $J \in \mathbb{N}$ denote the numbers of workers and firms, respectively, and define $[I] \coloneqq \{1,\dots,I\}$ and $[J] \coloneqq \{1,\dots,J\}$. Each worker $i \in [I]$ has a fixed latent productivity $\alpha_i \in A$, and each firm $j \in [J]$ has a fixed latent productivity $\psi_j \in \Psi$, where $A$ and $\Psi$ are compact subsets of $\R^d$. Let $\boldsymbol{\alpha} \coloneqq (\alpha_1,\dots,\alpha_I)$ and $\boldsymbol{\psi} \coloneqq (\psi_1,\dots,\psi_J)$ denote the corresponding productivity collections.

The potential outcome of the interaction between worker $i$ and firm $j$, such as the wage worker $i$ would receive if employed by firm $j$, is
\begin{gather*}
    y_{ij} = \underbrace{f(\alpha_i,\psi_j)}_{\theta_{ij}} + \eta_{ij},
\end{gather*}
where $f \colon A \times \Psi \to \R$ is the interaction function and $\eta_{ij}$ is a mean-zero random term. Define $\theta_{ij} \coloneqq f(\alpha_i,\psi_j) = \E[y_{ij}]$ as the deterministic component of the outcome. To focus on $f$, I omit covariates. When repeated observations of the same pair are available, coefficients on covariates that vary within pairs can in principle be estimated using pair fixed effects, after which the analysis can be applied to covariate-adjusted outcomes; I leave a formal treatment of this extension for future work.

The cross-partial derivative of $f$ characterizes its modularity and complementarity properties. Formal definitions and a general discussion of modularity are provided by \citet{Topkis1998}. When $\alpha_i$ and $\psi_j$ are scalars and $f$ is twice differentiable, the interaction function is \emph{modular} if its cross-partial derivative is zero, \emph{supermodular} if it is nonnegative, and \emph{submodular} if it is nonpositive. I say that the interaction function exhibits \emph{complementarities} when the cross-partial derivative is nonzero, with \emph{positive complementarities} when it is positive and \emph{negative complementarities} when it is negative.

The potential outcome $y_{ij}$ is defined for every worker-firm pair and represents the outcome that would be realized if match $(i,j)$ occurred, whether or not it is observed in the data. In practice, only a subset of matches is realized, and $y_{ij}$ is observed only for those pairs. Define
\begin{gather*}
    D_{ij}
    =
    \begin{cases}
        1, & \text{if } y_{ij} \text{ is observed}, \\
        0, & \text{otherwise},
    \end{cases}
\end{gather*}
as the indicator of an observed match, treated as fixed in the analysis, and let $\mathcal{O}_{IJ} = \{(i,j) : D_{ij} = 1\}$ denote the set of observed matches.

Let $G_{IJ} = ([I],[J],\mathcal{O}_{IJ})$ be the bipartite network that links worker $i$ to firm $j$ whenever $D_{ij} = 1$. The network is bipartite because its nodes are partitioned into two disjoint sets, workers and firms, with no edges within either set.

\begin{df}
    {\normalfont (Matching Network)}
    The matching network $G_{IJ}$ is the bipartite graph with node sets $[I]$ and $[J]$ and edge set $\mathcal{O}_{IJ} = \{(i,j) : D_{ij} = 1\}$. An edge $(i,j)$ indicates that $y_{ij}$ is observed.
\end{df}

Figure~\ref{fig:bipartite_graph} shows an example with 5 workers (purple), 3 firms (green), and 6 edges, each representing an observed match.

\begin{figure}[ht]
    \centering
    \begin{tikzpicture}[
        scale=0.62,
        >=stealth,
        font=\small,
        pointI/.style={circle, fill=purple, draw=black, inner sep=2pt},
        pointJ/.style={circle, fill=green, draw=black, inner sep=2pt}
    ]
        % 1. Points
        \foreach \y/\lab in {1.5/1, 0.75/2, 0/3, -0.75/4, -1.5/5} {
            \node[pointI] (I\lab) at (0,\y) {};
        }
        \foreach \y/\lab in {0.75/1, 0/2, -1.5/3} {
            \node[pointJ] (J\lab) at (4,\y) {};
        }

        % 2. Labels
        \node[left=4pt of I1] {$i_1$};
        \node[left=4pt of I2] {$i_2$};
        \node[left=4pt of I3] {$i_3$};
        \node[left=4pt of I4] {$i_4$};
        \node[left=4pt of I5] {$i_5$};

        \node[right=4pt of J1] {$j_1$};
        \node[right=4pt of J2] {$j_2$};
        \node[right=4pt of J3] {$j_3$};

        % 3. Edges
        \draw (I1) -- (J1);
        \draw (I2) -- (J1);
        \draw (I2) -- (J2);
        \draw (I3) -- (J2);
        \draw (I4) -- (J2);
        \draw (I5) -- (J3);
    \end{tikzpicture}

    \caption{\small Matching network with workers in purple and firms in green. In this example, $I = 5$, $J = 3$, and the set of observed matches is $\mathcal{O}_{IJ} = \{(i_1,j_1), (i_2,j_1), (i_2,j_2), (i_3,j_2), (i_4,j_2), (i_5,j_3)\}$.}
    \label{fig:bipartite_graph}
\end{figure}

The fact that worker $i_2$, for example, is linked to two firms, $j_1$ and $j_2$, does not imply that the matches occur simultaneously. The matching network is constructed over a chosen time window, which may span several years, and each edge may correspond to a different period: $i_2$ may be employed by $j_1$ in one year and by $j_2$ in another. The BI framework is static and abstracts from the timing of moves, so the order of these matches is irrelevant.

The matching network is deterministic: the BI framework does not model its formation, and the structure of $G_{IJ}$ is taken as given. The only sources of randomness are the terms $\{\eta_{ij}\}_{(i,j) \in \mathcal{O}_{IJ}}$, one for each observed outcome. These terms are mutually independent and mean-zero and may have pair-specific distributions, allowing for heteroskedasticity. Because $G_{IJ}$ is fixed, the observation of match $(i,j)$ does not depend on $\eta_{ij}$, so the random term is exogenous with respect to the matching network.

For any realized match ($D_{ij} = 1$), the model considers a single outcome $y_{ij}$. If repeated observations of the same worker-firm pair are available, they can be averaged to obtain a single $y_{ij}$, leaving the analysis unchanged.

\paragraph{Objects of interest and scope.}
Given this observational structure, the primitive parameters to be recovered are the interaction function $f$ and the productivity collections $\boldsymbol{\alpha}$ and $\boldsymbol{\psi}$. The main object of interest is the complementarity structure of $f$: whether an agent's contribution depends on the latent productivity of the agent on the other side of the match and, if so, the direction and strength of this dependence.

The productivity collections are also economically relevant. They describe heterogeneity within each side of the interaction and can be used to study sorting, variance decompositions, counterfactual assignments, and other derived quantities. A comprehensive analysis of these quantities is beyond the scope of the paper: I focus on identifying the primitives $f$, $\boldsymbol{\alpha}$, and $\boldsymbol{\psi}$ under alternative restrictions on the interaction function and the matching network.

\paragraph{Interaction specifications.}
What can be learned about these primitives depends jointly on the restrictions imposed on $f$ and the structure of the observed matching network. Different restrictions on $f$ define different models within the BI framework and determine which features of the complementarity structure can be identified from a given network. The main focus of the paper is the Tukey model, introduced in the next section. Supplementary Material~\ref{sec:sm_extensions} considers two richer specifications as benchmarks for assessing how identification requirements change as the interaction function becomes more flexible.

\subsection{Tukey Model}

The cross-partial derivative of $f$ captures the complementarity structure of the interaction function. A parsimonious way to model complementarities is to assume that this derivative is constant and summarized by a single parameter. This yields the specification
\begin{gather} \label{eq:tukey}
    \theta_{ij} = \alpha_i + \psi_j + \beta_0 \alpha_i \psi_j, \tag{Tukey model}
\end{gather}
with scalar $\alpha_i \in A \subset \R$, $\psi_j \in \Psi \subset \R$, and $\beta_0 \in B \subset \R$, where $A$, $\Psi$, and $B$ are compact. I refer to this specification as the Tukey model, after the statistician John Tukey, who studied an analogous functional form in two-way ANOVA to test whether two categorical factors affect the response additively \citep{tukey1949one, ward1952non, vsimevcek2013modification}. That literature focuses on testing the null hypothesis $\beta_0 = 0$ under homoskedastic, normally distributed errors $\eta_{ij}$ and a complete matching network $G_{IJ}$, rather than on estimating the parameter itself.

In the BI framework, the interaction parameter $\beta_0$ has a direct economic meaning. It equals the constant cross-partial derivative $\partial^2 f / \partial \alpha \partial \psi$; it captures all departures from modularity and governs the complementarity pattern between the two productivities.

Although the Tukey model summarizes the entire complementarity structure with a single parameter, it accommodates supermodular ($\beta_0 \geq 0$), submodular ($\beta_0 \leq 0$), and modular ($\beta_0 = 0$) interaction functions. The sign of $\beta_0$ determines the direction of complementarities, while its magnitude determines their strength by governing the importance of the multiplicative component relative to the additive components.

In a labor-market interpretation, the Tukey model can also arise as a wage equation. Supplementary Material~\ref{sec:sm_microfoundation} provides a simple microfoundation in which, under linear utility, nonseparable task production, and full worker bargaining power, the deterministic component of the hiring wage takes the Tukey form. This interpretation is not required for the econometric results below, which apply to any two-sided setting satisfying the BI specification.

When $\beta_0 = 0$, the Tukey model reduces to the widely used TWFE specification, in its baseline form without covariates:
\begin{gather} \label{eq:TWFE}
    \theta_{ij} = \alpha_i + \psi_j. \tag{TWFE model}
\end{gather}
In the TWFE model, the cross-partial derivative $\partial^2 f / \partial \alpha \partial \psi$ is zero, so the interaction function is modular and complementarities are ruled out by assumption. A worker's marginal contribution is independent of the firm with which they are matched, and a firm's marginal contribution is likewise independent of the worker. Under modularity, total output for a fixed set of agents on each side is invariant to their assignment: any allocation is efficient.

While the assumption of no complementarities is restrictive, especially given the emphasis in economic theory on complementarities as a driver of sorting patterns such as positive assortative matching \citep{becker1973theory, shimer2000assortative}, its appropriateness depends on the empirical setting and the research question. In practice, the TWFE model is often viewed less as a literal description of interactions and more as a tractable approximation to richer structures \citep{abowd1999high, card2013workplace}. The nesting of TWFE within the Tukey model makes it possible to assess when this approximation is likely to be informative and, conversely, when ignoring complementarities may lead to misleading conclusions; see Section~\ref{sec:twfe_sorting_bias} for details.

\subsection{Richer Interaction Functions} \label{sec:extension}

The Tukey model introduces complementarities in a simple and interpretable way. To capture richer complementarity patterns, Supplementary Material~\ref{sec:sm_extensions} considers two additional specifications. The \emph{heterogeneous-slope} model allows the Tukey interaction parameter to vary across firms, accommodating heterogeneity in complementarities across one side of the match. The \emph{isotonic} model requires the interaction function to be strictly increasing in each argument but otherwise leaves its complementarity structure unrestricted.

These specifications serve as identification benchmarks. Comparing them with the Tukey model shows that the requirements on the matching network become stronger as restrictions on the interaction function are relaxed. The comparison also clarifies the role of the constant cross-partial assumption: it allows economically meaningful complementarities to be identified under substantially weaker graph conditions than richer interaction structures require. With this trade-off in view, I now turn to identification in the Tukey model.

\section{Identification} \label{sec:identification}

This section derives the identification conditions for the Tukey model. Section~\ref{sec:definition_identification} defines the notion of identification used in the analysis. Section~\ref{sec:tukey_identification} presents the main results, while Section~\ref{sec:twfe_sorting_bias} studies how the TWFE estimands approximate the productivity collections $\boldsymbol{\alpha}$ and $\boldsymbol{\psi}$ when the interaction function follows the Tukey model.

\subsection{Identification in the BI Framework} \label{sec:definition_identification}

I adopt the notion of point identification in \citet{koopmans1949identification}, which requires the model parameters to be uniquely recoverable from the distribution of observables. In the BI framework, the observables are the collection $\{y_{ij}\}_{(i,j) \in \mathcal{O}_{IJ}}$. For identification purposes, I can therefore treat $\E[y_{ij}] = \theta_{ij}$ as known for each observed match.

Formally, I study identification through the noiseless map
\begin{gather*}
    (f, \boldsymbol{\alpha}, \boldsymbol{\psi}, G_{IJ})
    \mapsto
    \boldsymbol{\theta}_{\mathcal{O}}
    \coloneqq
    \{\theta_{ij}\}_{(i,j) \in \mathcal{O}_{IJ}},
\end{gather*}
where $(f, \boldsymbol{\alpha}, \boldsymbol{\psi})$ are unknown and the graph $G_{IJ}$ is known. The parameter tuple $(f, \boldsymbol{\alpha}, \boldsymbol{\psi})$ is \emph{point-identified} if, for a given $G_{IJ}$, the map
\begin{gather*}
    (f, \boldsymbol{\alpha}, \boldsymbol{\psi})
    \mapsto
    \boldsymbol{\theta}_{\mathcal{O}}
\end{gather*}
is injective. That is, if $(f, \boldsymbol{\alpha}, \boldsymbol{\psi})$ and $(f', \boldsymbol{\alpha}', \boldsymbol{\psi}')$ generate the same $\boldsymbol{\theta}_{\mathcal{O}}$, then they must coincide.

Equivalently, a parameter is identified if it can be expressed as a function of $\boldsymbol{\theta}_{\mathcal{O}}$. I use this characterization to establish identification in the proofs.

This definition deliberately abstracts from sampling noise and treats the distribution of each $y_{ij}$ as known, even though in applications each match is typically observed only once. While such a notion does not distinguish between parameters that can or cannot be consistently estimated in the presence of error terms $\eta_{ij}$, it provides a fundamental benchmark: if a parameter is not identified in the noiseless model, then no estimator can recover it. Conversely, whenever a parameter is identified in this sense, it warrants further analysis to determine whether, and under what conditions, consistent estimation is feasible.

A location normalization is often required because the absolute levels of $\boldsymbol{\alpha}$ and $\boldsymbol{\psi}$ are not uniquely determined by the observables. In the TWFE model, for example, adding a constant to every $\alpha_i$ and subtracting it from every $\psi_j$ leaves $\boldsymbol{\theta}_{\mathcal{O}}$ unchanged. Related invariances arise in richer specifications with complementarities, where transformations of $\boldsymbol{\alpha}$ can be offset by corresponding changes in $\boldsymbol{\psi}$ and $f$. To fix the reference level and ensure point identification, I impose a normalization such as $\sum_i \alpha_i = 0$ or $\alpha_1 = 0$. Whenever a normalization is required, I state it explicitly and use the form that yields the clearest expressions.

\subsection{Identification in the Tukey Model} \label{sec:tukey_identification}

The Tukey model introduces complementarities through a single parameter, $\beta_0$, which equals the constant cross-partial derivative and fully characterizes the complementarity structure of $f$. I present the identification results in two steps: first, the identification of $\beta_0$; second, the identification of the productivity collections $\boldsymbol{\alpha}$ and $\boldsymbol{\psi}$.

\subsubsection{Identification of \texorpdfstring{$\beta_0$}{beta0}}

Identification of $\beta_0$ requires additional structure in the matching network $G_{IJ}$. I begin by recalling the definition of a cycle in a bipartite graph.

\begin{df}
    {\normalfont (Cycle in the Matching Network)}
    In the bipartite graph $G_{IJ} = ([I],[J],\mathcal{O}_{IJ})$, a cycle of length $2K$, with $K \geq 2$, is a closed alternating sequence of workers and firms,
    \begin{gather*}
        i_1, j_1, i_2, j_2, \dots, i_K, j_K, i_1,
    \end{gather*}
    such that $(i_k,j_k)$ and $(i_{k+1},j_k)$ belong to $\mathcal{O}_{IJ}$ for each $k = 1,\dots,K$, with $i_{K+1} = i_1$, and all workers $i_1,\dots,i_K$ and firms $j_1,\dots,j_K$ are distinct.
\end{df}

The graph in Figure~\ref{fig:bipartite_graph_panel_a} contains no cycles. By contrast, adding an edge between $i_1$ and $j_2$, as in Figure~\ref{fig:bipartite_graph_panel_b}, creates the closed path $i_1, j_1, i_2, j_2, i_1$, which forms a four-cycle (a cycle of length 4).

\begin{figure}[ht]
    \centering
    \captionsetup[subfigure]{justification=centering,skip=0pt}

    % ---------- (a) ----------
    \begin{subfigure}[t]{0.48\textwidth}
        \centering
        \begin{tikzpicture}[
            scale=0.62,
            >=stealth,
            font=\small,
            pointI/.style={circle, fill=purple, draw=black, inner sep=2pt},
            pointJ/.style={circle, fill=green, draw=black, inner sep=2pt}
        ]
            % 1. Points
            \foreach \y/\lab in {1.5/1, 0.75/2, 0/3, -0.75/4, -1.5/5} {
                \node[pointI] (I\lab) at (0,\y) {};
            }
            \foreach \y/\lab in {0.75/1, 0/2, -1.5/3} {
                \node[pointJ] (J\lab) at (4,\y) {};
            }

            % 2. Labels
            \node[left=4pt of I1] {$i_1$};
            \node[left=4pt of I2] {$i_2$};
            \node[left=4pt of I3] {$i_3$};
            \node[left=4pt of I4] {$i_4$};
            \node[left=4pt of I5] {$i_5$};

            \node[right=4pt of J1] {$j_1$};
            \node[right=4pt of J2] {$j_2$};
            \node[right=4pt of J3] {$j_3$};

            % 3. Edges
            \draw (I1) -- (J1);
            \draw (I2) -- (J1);
            \draw (I2) -- (J2);
            \draw (I3) -- (J2);
            \draw (I4) -- (J2);
            \draw (I5) -- (J3);
        \end{tikzpicture}
        \phantomsubcaption
        \label{fig:bipartite_graph_panel_a}
    \end{subfigure}
    \hfill
    % ---------- (b) ----------
    \begin{subfigure}[t]{0.48\textwidth}
        \centering
        \begin{tikzpicture}[
            scale=0.62,
            >=stealth,
            font=\small,
            pointI/.style={circle, fill=purple, draw=black, inner sep=2pt},
            pointJ/.style={circle, fill=green, draw=black, inner sep=2pt}
        ]
            % 1. Points
            \foreach \y/\lab in {1.5/1, 0.75/2, 0/3, -0.75/4, -1.5/5} {
                \node[pointI] (I\lab) at (0,\y) {};
            }
            \foreach \y/\lab in {0.75/1, 0/2, -1.5/3} {
                \node[pointJ] (J\lab) at (4,\y) {};
            }

            % 2. Labels
            \node[left=4pt of I1] {$i_1$};
            \node[left=4pt of I2] {$i_2$};
            \node[left=4pt of I3] {$i_3$};
            \node[left=4pt of I4] {$i_4$};
            \node[left=4pt of I5] {$i_5$};

            \node[right=4pt of J1] {$j_1$};
            \node[right=4pt of J2] {$j_2$};
            \node[right=4pt of J3] {$j_3$};

            % 3. Edges
            \draw (I1) -- (J1);
            \draw[red, very thick] (I1) -- (J2);
            \draw (I2) -- (J1);
            \draw (I2) -- (J2);
            \draw (I3) -- (J2);
            \draw (I4) -- (J2);
            \draw (I5) -- (J3);
        \end{tikzpicture}
        \phantomsubcaption
        \label{fig:bipartite_graph_panel_b}
    \end{subfigure}

    \caption{\small Two matching networks: (a) a network without cycles and (b) a network containing the four-cycle $i_1 \to j_1 \to i_2 \to j_2 \to i_1$.}
    \label{fig:bipartite_graph_cycle}
\end{figure}

The key condition for identification of $\beta_0$ is stated in Assumption~\ref{ass:informative_cycle}.

\begin{ass} \label{ass:informative_cycle}
    {\normalfont (Informative Four-Cycle)}
    The matching network $G_{IJ}$ contains a four-cycle $i_1,j_1,i_2,j_2$ such that $\alpha_{i_1} \neq \alpha_{i_2}$ and $\psi_{j_1} \neq \psi_{j_2}$.
\end{ass}

Assumption~\ref{ass:informative_cycle} requires a four-cycle with productivity heterogeneity on both sides of the interaction. Without such heterogeneity, the cycle contrasts that isolate $\beta_0$ vanish, so the cycle contains no information about the interaction parameter.

\begin{theorem} \label{thm:beta_identification}
    {\normalfont (Identification of $\beta_0$)}
    Under the \ref{eq:tukey}, Assumption~\ref{ass:informative_cycle} is sufficient for identification of $\beta_0$. If the matching network contains at most one cycle, the assumption is also necessary.
\end{theorem}

Identification of $\beta_0$ requires a cycle in the matching network. When the network contains a unique cycle, point identification requires that cycle to be of length four. The proof in Appendix~\ref{proof:beta_identification} shows that the outcomes along a cycle of length $2K$ generate a degree-$(K-1)$ polynomial in $\beta_0$, whose coefficients are known functions of $\boldsymbol{\theta}_{\mathcal{O}}$. The identification set consists of the roots of this polynomial. When $K = 2$, the polynomial is linear, and an informative four-cycle uniquely identifies $\beta_0$. When $K > 2$, a single cycle yields an identification set with at most $K-1$ elements. Multiple longer cycles can nevertheless achieve point identification if their identification sets intersect at a single value.

The role of cycles in detecting departures from modularity was previously noted by \citet{card2013workplace} and discussed by \citet{kline2024firm}. In those studies, cycles serve as a diagnostic device; here, cycle variation identifies the interaction parameter.

In the labor-market setting, Assumption~\ref{ass:informative_cycle} requires worker mobility across firms. Because each worker can be matched with only one firm at a time, the multiple links needed to form a cycle arise when workers change employers. The condition therefore requires at least two workers to move between the same pair of firms, with productivity heterogeneity among both the workers and the firms involved. This is not especially restrictive: labor markets are typically segmented into local or sectoral clusters, and when one worker moves between two firms, the likelihood of additional movers between the same firms is higher than under random matching. Evidence supports this: in the application of \cite{kline2024firm}, for instance, about 55\% of firms belong to at least one cycle.

\subsubsection{Identification of \texorpdfstring{$\boldsymbol{\alpha}$}{alpha} and \texorpdfstring{$\boldsymbol{\psi}$}{psi}}

Once $\beta_0$ is identified, the productivity collections $\boldsymbol{\alpha}$ and $\boldsymbol{\psi}$ can be identified under an additional condition on $G_{IJ}$.

\begin{ass} \label{ass:connectedness}
    {\normalfont (Connectedness)}
    The matching network $G_{IJ}$ is connected: for any two nodes in $G_{IJ}$, there exists a path (i.e., a sequence of nodes linked by edges) joining them.
\end{ass}

The graph in Figure~\ref{fig:bipartite_graph_connected_panel_a} violates Assumption~\ref{ass:connectedness}; for example, no path connects nodes $i_4$ and $i_5$. Adding an edge between $i_4$ and $j_3$, as in Figure~\ref{fig:bipartite_graph_connected_panel_b}, makes the graph connected.

\begin{figure}[ht]
    \centering
    \captionsetup[subfigure]{justification=centering,skip=0pt}

    % ---------- (a) ----------
    \begin{subfigure}[t]{0.48\textwidth}
        \centering
        \begin{tikzpicture}[
            scale=0.62,
            >=stealth,
            font=\small,
            pointI/.style={circle, fill=purple, draw=black, inner sep=2pt},
            pointJ/.style={circle, fill=green, draw=black, inner sep=2pt}
        ]
            % 1. Points
            \foreach \y/\lab in {1.5/1, 0.75/2, 0/3, -0.75/4, -1.5/5} {
                \node[pointI] (I\lab) at (0,\y) {};
            }
            \foreach \y/\lab in {0.75/1, 0/2, -1.5/3} {
                \node[pointJ] (J\lab) at (4,\y) {};
            }

            % 2. Labels
            \node[left=4pt of I1] {$i_1$};
            \node[left=4pt of I2] {$i_2$};
            \node[left=4pt of I3] {$i_3$};
            \node[left=4pt of I4] {$i_4$};
            \node[left=4pt of I5] {$i_5$};

            \node[right=4pt of J1] {$j_1$};
            \node[right=4pt of J2] {$j_2$};
            \node[right=4pt of J3] {$j_3$};

            % 3. Edges
            \draw (I1) -- (J1);
            \draw (I1) -- (J2);
            \draw (I2) -- (J1);
            \draw (I2) -- (J2);
            \draw (I3) -- (J2);
            \draw (I4) -- (J2);
            \draw (I5) -- (J3);
        \end{tikzpicture}
        \phantomsubcaption
        \label{fig:bipartite_graph_connected_panel_a}
    \end{subfigure}
    \hfill
    % ---------- (b) ----------
    \begin{subfigure}[t]{0.48\textwidth}
        \centering
        \begin{tikzpicture}[
            scale=0.62,
            >=stealth,
            font=\small,
            pointI/.style={circle, fill=purple, draw=black, inner sep=2pt},
            pointJ/.style={circle, fill=green, draw=black, inner sep=2pt}
        ]
            % 1. Points
            \foreach \y/\lab in {1.5/1, 0.75/2, 0/3, -0.75/4, -1.5/5} {
                \node[pointI] (I\lab) at (0,\y) {};
            }
            \foreach \y/\lab in {0.75/1, 0/2, -1.5/3} {
                \node[pointJ] (J\lab) at (4,\y) {};
            }

            % 2. Labels
            \node[left=4pt of I1] {$i_1$};
            \node[left=4pt of I2] {$i_2$};
            \node[left=4pt of I3] {$i_3$};
            \node[left=4pt of I4] {$i_4$};
            \node[left=4pt of I5] {$i_5$};

            \node[right=4pt of J1] {$j_1$};
            \node[right=4pt of J2] {$j_2$};
            \node[right=4pt of J3] {$j_3$};

            % 3. Edges
            \draw (I1) -- (J1);
            \draw (I1) -- (J2);
            \draw (I2) -- (J1);
            \draw (I2) -- (J2);
            \draw (I3) -- (J2);
            \draw (I4) -- (J2);
            \draw (I5) -- (J3);
            \draw[red, very thick] (I4) -- (J3);
        \end{tikzpicture}
        \phantomsubcaption
        \label{fig:bipartite_graph_connected_panel_b}
    \end{subfigure}

    \caption{\small Two matching networks: (a) a disconnected network and (b) a connected network satisfying Assumption~\ref{ass:connectedness}.}
    \label{fig:bipartite_graph_connected}
\end{figure}

The identification result for the productivity collections $\boldsymbol{\alpha}$ and $\boldsymbol{\psi}$ is stated below. The result excludes the degenerate parameter configurations in which $1 + \beta_0 \alpha_i = 0$ for some worker $i$ or $1 + \beta_0 \psi_j = 0$ for some firm $j$.

\begin{theorem} \label{thm:tukey_identification}
    {\normalfont (Identification of $\boldsymbol{\alpha}$ and $\boldsymbol{\psi}$ in the Tukey Model)}
    Suppose that the \ref{eq:tukey} holds, $\beta_0$ is identified, and the normalization $\alpha_{i_0} = 0$ is imposed. At nondegenerate parameter values, Assumption~\ref{ass:connectedness} is necessary and sufficient for identification of $\boldsymbol{\alpha}$ and $\boldsymbol{\psi}$.
\end{theorem}

Theorem~\ref{thm:tukey_identification} extends the classical identification result for the TWFE model by allowing $\beta_0$ to differ from zero, provided that it is identified. The proof first constructs, for each observed match, a function of $\boldsymbol{\theta}_{\mathcal{O}}$ and $\beta_0$ that equals the product of known functions of the corresponding worker and firm productivities. Connectedness of $G_{IJ}$ then allows all worker and firm productivities to be recovered up to the normalization.

In labor-market applications, the matching network is rarely fully connected, especially over short time horizons. Analyses therefore commonly focus on the largest connected component. In the West German labor market studied by \citet{card2013workplace}, for example, the largest component contains more than 95\% of workers and 90\% of firms. \citet{bonhomme2023much} report similar firm coverage for Austria, Italy, Sweden, Norway, and the United States, although worker coverage is often substantially lower and in some cases falls below 50\%.

Theorems~\ref{thm:beta_identification} and~\ref{thm:tukey_identification} highlight the additional requirements introduced by allowing complementarities. Relative to the TWFE model, the Tukey model requires only a modest strengthening of the network conditions: in addition to connectedness, the graph must contain at least one informative four-cycle to identify $\beta_0$. This result exploits the fact that $\beta_0$ is a global parameter governing the complementarity structure of $f$ and does not vary across workers or firms. It can therefore be identified from local variation in $G_{IJ}$ and, together with connectedness, used to recover the productivity collections $\boldsymbol{\alpha}$ and $\boldsymbol{\psi}$ throughout the graph. This raises a natural question: what is lost by imposing the simpler TWFE specification when the interaction function exhibits complementarities?

\subsection{TWFE as an Approximation} \label{sec:twfe_sorting_bias}

I address this question by studying how well the TWFE model approximates interactions governed by the Tukey model. A common justification for TWFE is that departures from additive separability may be quantitatively limited, making the model a useful approximation even when complementarities are present. Because TWFE is nested within the Tukey model, this claim can be evaluated formally in the noiseless setting. Appendix~\ref{appendix:approximation} characterizes when the approximation is accurate and shows that TWFE can yield misleading conclusions when complementarities are non-negligible.

When the interaction function follows the Tukey model but the researcher imposes TWFE, the multiplicative term $\beta_0 \alpha_i \psi_j$ is absorbed into additive worker and firm components. How it is absorbed depends on the structure of the observed matching network and the sorting pattern embedded in it. I derive the resulting bias in the TWFE estimands and show, for example, that under a supermodular interaction function and positive sorting, the sorting measured by TWFE can be zero. In such cases, TWFE incorrectly suggests no sorting.

Thus, even in the noiseless setting, TWFE estimands may have limited informational value when complementarities are important. This motivates methods that allow the interaction function $f$ to depart from modularity. The Tukey model provides a parsimonious extension of TWFE, introducing a single interaction parameter while remaining identifiable under substantially weaker graph conditions than the richer benchmark specifications studied in Supplementary Material~\ref{sec:sm_ident_supp}. Having established when this parameter can be recovered from the deterministic components of observed outcomes, I now study its estimation and inference when outcomes are observed with noise.

\section{Estimation and Inference in the Tukey Model} \label{sec:estimation_inference}

This section develops estimation and inference for the Tukey model, focusing primarily on the interaction parameter. Section~\ref{sec:asymptotic_framework} presents the asymptotic framework. Section~\ref{sec:estimation_inference_beta0} introduces an estimator of $\beta_0$, establishes its consistency and asymptotic distribution, and uses it to construct a test of no complementarities within the Tukey model. Section~\ref{sec:estimation_inference_prod} discusses the distinct challenges involved in estimating the productivity collections $\boldsymbol{\alpha}$ and $\boldsymbol{\psi}$.

\subsection{Sparse-Network Asymptotic Framework} \label{sec:asymptotic_framework}

To study the large-sample properties of the estimator introduced below, I consider an asymptotic setting in which the numbers of workers and firms both grow, with $I \to \infty$ and $J \to \infty$. Information accumulates through the addition of new nodes and edges to the bipartite graph, rather than through repeated observations along existing edges.

The analysis therefore considers deterministic sequences of expanding matching networks $G_{IJ}$, worker productivity collections $\boldsymbol{\alpha}_I$, and firm productivity collections $\boldsymbol{\psi}_J$, with the set of observed matches $\mathcal{O}_{IJ}$ varying with $(I,J)$. I suppress the subscripts $I$ and $J$ on the productivity collections when no confusion can arise. Conditional on these sequences, randomness comes from the error terms $\{\eta_{ij}\}_{(i,j) \in \mathcal{O}_{IJ}}$, which are independent across matches. When the orientation sequence is random, inference is also conducted conditionally on its realization.

This framework differs from stochastic network formation models, in which the graph itself is random. It reflects applications where the observed matching network is treated as fixed and inference concerns the outcome shocks conditional on its structure. Equivalently, the network and productivities may be viewed as random objects on which the analysis conditions.

The role of the matching network in the asymptotic analysis mirrors its role in identification: conditions on the sequence of graphs $G_{IJ}$ ensure the validity of the asymptotic results. Although these conditions concern the full sequence and cannot be verified from a single observed graph, their relevance can be assessed by comparing the observed network with the properties required asymptotically. This comparison indicates whether the asymptotic results are likely to approximate finite-sample behavior well.

The large-$I$, large-$J$ asymptotic framework is well suited to labor-market applications, where datasets often contain millions of workers and firms. At the same time, the matching networks are sparse: each worker is observed with only a handful of firms, and each firm with a modest number of workers. For example, over a one-year horizon in the U.S. labor market, a typical worker is employed by one or two firms, while a typical firm employs only a few dozen workers. To accommodate this structure, the asymptotic framework allows node degrees to remain bounded as $I$ and $J$ grow, rather than requiring any worker to be linked to many firms or any firm to many workers.

\subsection{Estimation and Inference for the Interaction Parameter} \label{sec:estimation_inference_beta0}

Recall the Tukey model:
\begin{gather*}
    y_{ij} = \alpha_i + \psi_j + \beta_0 \alpha_i \psi_j + \eta_{ij},
\end{gather*}
with $\E[\eta_{ij}] = 0$, and $\{\eta_{ij}\}_{(i,j) \in \mathcal{O}_{IJ}}$ independent.

This section introduces an estimator of $\beta_0$ and studies its properties under the large-sample framework described above. \citet{tukey1949one} and the subsequent literature on nonadditivity in ANOVA consider the same model but focus on testing the hypothesis $\beta_0 = 0$ rather than on estimating $\beta_0$.

Theorem~\ref{thm:beta_identification} shows that a single informative four-cycle suffices to identify $\beta_0$. With noisy outcomes, however, consistent estimation requires pooling information across multiple four-cycles in $G_{IJ}$. I therefore treat each four-cycle as an observational unit and impose conditions on the sequence of matching networks ensuring that the number of edge-disjoint four-cycles used by the estimator grows with $I$ and $J$.

The analysis below considers two estimators. The first is a generic cycle-based estimator, denoted $\hat{\beta}_{L,s}$, where $\boldsymbol{s}_L$ records the orientation assigned to each cycle. This estimator provides the basis for the theoretical analysis, as the consistency and asymptotic normality results apply to any admissible orientation sequence. The second is a feasible estimator, denoted $\hat{\beta}_{L,z}$, obtained from a rank-based orientation rule that uses observable ranking variables informative about latent productivity rankings. After deriving the generic results, I provide conditions under which $\hat{\beta}_{L,z}$ inherits the asymptotic properties of $\hat{\beta}_{L,s}$ and can be used for inference and testing.

\subsubsection{Cycle-Based Estimator and Orientations}

Index the four-cycles used by the estimator by $\ell = 1, \dots, L$. For clarity of exposition, I restrict attention to edge-disjoint cycles, so that no two cycles share an edge, although they may share nodes. This restriction is not essential, and information from overlapping cycles can also be aggregated, but focusing on edge-disjoint cycles keeps the notation tractable.

Each cycle $\ell$ contains two distinct workers and two distinct firms. Fix an arbitrary deterministic reference ordering of these nodes, denoting the workers by $(i_{\ell 1},i_{\ell 2})$ and the firms by $(j_{\ell 1},j_{\ell 2})$. This ordering has no substantive meaning and may, for example, be based on numerical identifiers.

Under this reference ordering, define the two cycle statistics
\begin{align*}
    \hat{\Delta}_{1,\ell}
    &\coloneqq
    y_{i_{\ell 1}j_{\ell 1}} - y_{i_{\ell 2}j_{\ell 1}} - y_{i_{\ell 1}j_{\ell 2}} + y_{i_{\ell 2}j_{\ell 2}}, \\
    \hat{\Delta}_{2,\ell}
    &\coloneqq
    y_{i_{\ell 1}j_{\ell 1}} y_{i_{\ell 2}j_{\ell 2}} - y_{i_{\ell 2}j_{\ell 1}} y_{i_{\ell 1}j_{\ell 2}}.
\end{align*}
The first statistic is a difference-in-differences contrast of the four outcomes in the cycle. The second is the difference between the two cross-products formed from these outcomes.

Swapping the two workers or the two firms changes the sign of both statistics, while swapping both pairs leaves them unchanged. All ordered representations of a cycle can therefore be summarized by an orientation sign $s_\ell \in \{-1,1\}$, which determines the oriented cycle statistics $s_\ell \hat{\Delta}_{1,\ell}$ and $s_\ell \hat{\Delta}_{2,\ell}$.

Let $\boldsymbol{s}_L \coloneqq (s_1,\dots,s_L)$ denote the orientation sequence. For a given $\boldsymbol{s}_L$, define the generic cycle-based estimator
\begin{gather} \label{eq:beta_hat}
    \hat{\beta}_{L,s}
    \coloneqq
    -\frac{
        \tfrac{1}{L} \sum_{\ell=1}^L
        s_\ell \hat{\Delta}_{1,\ell}
    }{
        \tfrac{1}{L} \sum_{\ell=1}^L
        s_\ell \hat{\Delta}_{2,\ell}
    }.
\end{gather}
The estimator does not require estimating either productivity collection. Consequently, $\beta_0$ can be estimated consistently even when individual productivities cannot, such as when each worker is linked to only a few firms; see Remark~\ref{remark:consistency}.

The estimator depends on the orientation sequence. With $L$ cycles, there are $2^L$ possible orientation sequences, and different choices can produce different estimates from the same data. This motivates the rank-based orientation rule introduced in Section~\ref{sec:rbl}.

To see why the estimator recovers $\beta_0$, define
\begin{gather*}
    A_\ell
    \coloneqq
    (\alpha_{i_{\ell 1}} - \alpha_{i_{\ell 2}})
    (\psi_{j_{\ell 1}} - \psi_{j_{\ell 2}}).
\end{gather*}
The sign of $A_\ell$ depends on the arbitrary reference ordering, while its absolute value $|A_\ell|$ measures the combined productivity variation within cycle $\ell$ and is invariant to that ordering.

The cycle statistics admit the decompositions
\begin{align*}
    \hat{\Delta}_{1,\ell}
    &=
    \Delta_{1,\ell} + \epsilon_{1,\ell}, \\
    \hat{\Delta}_{2,\ell}
    &=
    \Delta_{2,\ell} + \epsilon_{2,\ell},
\end{align*}
where
\begin{gather*}
    \Delta_{1,\ell} = \beta_0 A_\ell,
    \qquad
    \Delta_{2,\ell} = -A_\ell,
\end{gather*}
and
\begin{align}
    \epsilon_{1,\ell}
    = &
    \eta_{i_{\ell 1}j_{\ell 1}}
    - \eta_{i_{\ell 2}j_{\ell 1}}
    - \eta_{i_{\ell 1}j_{\ell 2}}
    + \eta_{i_{\ell 2}j_{\ell 2}},
    \label{eq:epsilon1} \\
    \epsilon_{2,\ell}
    = & 
    \theta_{i_{\ell 1}j_{\ell 1}} \eta_{i_{\ell 2}j_{\ell 2}}
    + \theta_{i_{\ell 2}j_{\ell 2}} \eta_{i_{\ell 1}j_{\ell 1}}
    - \theta_{i_{\ell 2}j_{\ell 1}} \eta_{i_{\ell 1}j_{\ell 2}}
    - \theta_{i_{\ell 1}j_{\ell 2}} \eta_{i_{\ell 2}j_{\ell 1}} \notag \\
    &
    + \eta_{i_{\ell 1}j_{\ell 1}} \eta_{i_{\ell 2}j_{\ell 2}}
    - \eta_{i_{\ell 2}j_{\ell 1}} \eta_{i_{\ell 1}j_{\ell 2}}.
    \label{eq:epsilon2}
\end{align}
Consequently,
\begin{gather*}
    s_\ell \hat{\Delta}_{1,\ell}
    =
    \beta_0 s_\ell A_\ell + s_\ell \epsilon_{1,\ell},
    \qquad
    s_\ell \hat{\Delta}_{2,\ell}
    =
    -s_\ell A_\ell + s_\ell \epsilon_{2,\ell}.
\end{gather*}
Under the conditions introduced below, the sample averages of the oriented error components $s_\ell \epsilon_{1,\ell}$ and $s_\ell \epsilon_{2,\ell}$ converge to zero. The population components share the same oriented productivity-gap term, which cancels in their ratio, leaving $\beta_0$.

\subsubsection{Rank-Based Orientation and the Feasible Estimator} \label{sec:rbl}

The generic estimator $\hat{\beta}_{L,s}$ leaves the orientation sequence unspecified. I now define a feasible rule based on observable ranking variables. Let $\boldsymbol{z}^{\alpha} = (z^\alpha_1,\dots,z^\alpha_I)$ and $\boldsymbol{z}^{\psi} = (z^\psi_1,\dots,z^\psi_J)$ denote observable characteristics for workers and firms, respectively. These variables are treated as nonrandom and are used to orient each cycle according to the rankings they induce. Through the orientation signs they induce, the ranking variables play a role analogous to instruments: treating them as independent of the outcome errors ensures exogeneity, while their informativeness about the latent productivity rankings provides relevance. For expositional simplicity, assume that there are no within-cycle ties. When ties occur, they are resolved using a pre-specified rule independent of the outcome errors.

\begin{df} \label{def:rbl}
    {\normalfont (Rank-Based Orientation)}
    For each cycle $\ell$, define
    \begin{gather*}
        s_{\ell,z}
        \coloneqq
        \operatorname{sign}
        \left[
            \left(z^\alpha_{i_{\ell 1}} - z^\alpha_{i_{\ell 2}}\right)
            \left(z^\psi_{j_{\ell 1}} - z^\psi_{j_{\ell 2}}\right)
        \right].
    \end{gather*}
    The rank-based orientation sequence is $\boldsymbol{s}_{L,z} \coloneqq (s_{1,z},\dots,s_{L,z})$.
\end{df}

Under this rule, a cycle is assigned sign $s_{\ell,z} = 1$ when the ranking variables on the two sides induce the same ordering relative to the reference ordering: either
\begin{gather*}
    z^\alpha_{i_{\ell 1}} > z^\alpha_{i_{\ell 2}}
    \quad \text{and} \quad
    z^\psi_{j_{\ell 1}} > z^\psi_{j_{\ell 2}},
\end{gather*}
or both inequalities are reversed. The cycle is assigned sign $s_{\ell,z} = -1$ when the induced orderings are opposite.

The oriented statistics $s_{\ell,z} \hat{\Delta}_{1,\ell}$ and $s_{\ell,z} \hat{\Delta}_{2,\ell}$ are invariant to the arbitrary reference ordering. Reversing the worker or firm ordering changes the signs of both $s_{\ell,z}$ and the corresponding cycle statistics, leaving their products unchanged.

The corresponding rank-based estimator is
\begin{gather} \label{eq:rbl}
    \hat{\beta}_{L,z}
    \coloneqq
    -\frac{
        \tfrac{1}{L} \sum_{\ell=1}^L
        s_{\ell,z} \hat{\Delta}_{1,\ell}
    }{
        \tfrac{1}{L} \sum_{\ell=1}^L
        s_{\ell,z} \hat{\Delta}_{2,\ell}
    }.
    \tag{Rank-based estimator}
\end{gather}
Thus, $\hat{\beta}_{L,z}$ is the generic estimator in Equation~\ref{eq:beta_hat} evaluated at the orientation sequence induced by the ranking variables.

The ranking variables should be observable characteristics that are plausibly informative about latent productivity rankings and are not constructed from the outcomes entering the cycle statistics. For example, years of schooling may help rank workers by productivity, while firm size may help rank firms. Outcomes are not suitable ranking variables because they contain the errors $\eta_{ij}$, which can make the orientation sequence dependent on the cycle errors. Assumption~\ref{ass:labeling} formalizes the required restriction, and Appendix~\ref{appendix:outcome_based_labeling} shows that outcome-based orientation may lead to biased estimates and invalid inference.

The following example illustrates how the orientation rule works in practice. The next section develops the asymptotic theory for the generic estimator $\hat{\beta}_{L,s}$. I then return to the rank-based estimator and provide conditions under which the orientation sequence induced by the ranking variables satisfies the required regularity conditions.

\begin{eg} \label{ex:labelings}
    {\normalfont (Rank-Based Orientation in Practice)}
    Consider two cycles. The first involves workers Alice and Bob and firms Canon and Dell; the second involves workers Elizabeth and Fred and firms General Motors and Honda. Fix Alice--Bob and Canon--Dell as the reference ordering for the first cycle, and Elizabeth--Fred and General Motors--Honda for the second.

    In the first cycle, Alice earns 120 at Canon and 100 at Dell, while Bob earns 100 at Canon and 90 at Dell. The reference cycle statistics are
    \begin{gather*}
        \hat{\Delta}_{1,1} = 120 - 100 - 100 + 90 = 10,
        \qquad
        \hat{\Delta}_{2,1} = 120 \times 90 - 100 \times 100 = 800.
    \end{gather*}
    For the second cycle, the statistics are $\hat{\Delta}_{1,2} = -20$ and $\hat{\Delta}_{2,2} = -900$.

    Use years of schooling as the worker ranking variable and the number of employees as the firm ranking variable. Their values are:
    \begin{center}
        \begin{minipage}{0.48\textwidth}
            \centering
            \begin{tabular}{cc}
                \hline \hline
                Worker & $z^\alpha_i$ \\
                \hline
                Alice & 16 \\
                Bob & 14 \\
                Elizabeth & 18 \\
                Fred & 12 \\
                \hline \hline
            \end{tabular}
        \end{minipage}
        \begin{minipage}{0.48\textwidth}
            \centering
            \begin{tabular}{cc}
                \hline \hline
                Firm & $z^\psi_j$ \\
                \hline
                Canon & 170{,}000 \\
                Dell & 110{,}000 \\
                General Motors & 160{,}000 \\
                Honda & 190{,}000 \\
                \hline \hline
            \end{tabular}
        \end{minipage}
    \end{center}

    The induced orientations and oriented cycle statistics are
    \begin{center}
        \begin{tabular}{ccccrrrrr}
            \hline \hline
            $i_{\ell 1}$ & $i_{\ell 2}$ &
            $j_{\ell 1}$ & $j_{\ell 2}$ &
            $s_{\ell,z}$ &
            $\hat{\Delta}_{1,\ell}$ &
            $\hat{\Delta}_{2,\ell}$ &
            $s_{\ell,z} \hat{\Delta}_{1,\ell}$ &
            $s_{\ell,z} \hat{\Delta}_{2,\ell}$ \\
            \hline
            Alice & Bob & Canon & Dell
            & $1$ & $10$ & $800$ & $10$ & $800$ \\
            Elizabeth & Fred & General Motors & Honda
            & $-1$ & $-20$ & $-900$ & $20$ & $900$ \\
            \hline \hline
        \end{tabular}
    \end{center}
    Each row of the final table represents one cycle and constitutes an observation used by the estimator. The estimator $\hat{\beta}_{L,z}$ is the negative ratio of the averages of the last two columns.
\end{eg}

\subsubsection{Consistency}

To establish consistency of $\hat{\beta}_{L,s}$, I impose the following conditions.

\begin{ass} \label{ass:consistency}
    \begin{enumerate}[label=3.\arabic*]
        \item[]
        \item \label{ass:error_eta}
        {\normalfont (Error Regularity)}
        The error terms $\eta_{ij}$ are independent across observed matches and satisfy $\E[\eta_{ij}] = 0$, $\Var(\eta_{ij}) \geq C_\eta > 0$, and $\E[\abs{\eta_{ij}}^{\,2+\delta}] \leq M < \infty$ for some $\delta > 0$ and uniform constants $C_\eta$ and $M$.

        \item \label{ass:cycles}
        {\normalfont (Cycle Growth)}
        The sequence of matching networks $G_{IJ}$ admits a collection of $L$ edge-disjoint four-cycles such that $L \to \infty$ as $I \to \infty$ and $J \to \infty$.

        \item \label{ass:seq_mu}
        {\normalfont (Cycle Heterogeneity)}
        The sequences of latent productivities and matching networks satisfy
        \begin{gather*}
            \mu_L
            \coloneqq
            \frac{1}{L} \sum_{\ell=1}^L |A_\ell|
            > C_\mu > 0
        \end{gather*}
        for all sufficiently large $L$.

        \item \label{ass:labeling}
        {\normalfont (Orientation Regularity)}
        The orientation sequence $\boldsymbol{s}_L = (s_1,\dots,s_L)$ is independent of the outcome errors. Its effective orientation signal,
        \begin{gather*}
            \kappa_{L,s}
            \coloneqq
            \frac{1}{L} \sum_{\ell=1}^L s_\ell A_\ell,
        \end{gather*}
        satisfies $|\kappa_{L,s}| > C_\kappa > 0$ for all sufficiently large $L$, almost surely when $\boldsymbol{s}_L$ is random.
    \end{enumerate}
\end{ass}

Assumption~\ref{ass:error_eta} imposes standard regularity conditions on the error terms. It does not require the $\eta_{ij}$ to be identically distributed and therefore allows for heteroskedasticity: each error may have a distinct distribution, provided it has mean zero, positive variance, and a finite $(2+\delta)$ moment. The variance lower bound and moment condition are not required for consistency but are used to derive the asymptotic distribution of the estimator.

Assumption~\ref{ass:cycles} requires the number of four-cycles to grow with the size of the matching network. It places no requirement that node degrees diverge and is compatible with sequences in which they remain bounded. For example, every worker and firm may have degree no greater than two. Because $\hat{\beta}_{L,s}$ treats four-cycles as the units of observation, the assumption ensures that the number of such units diverges.

Assumption~\ref{ass:seq_mu} requires the average absolute productivity-gap product to remain bounded away from zero. The condition is invariant to the arbitrary reference ordering and is closely related to Assumption~\ref{ass:informative_cycle}, as it ensures persistent productivity heterogeneity across cycles. Equivalently, it rules out sequences in which the identifying variation within cycles vanishes asymptotically.

Assumption~\ref{ass:labeling} restricts the orientation sequence. The scalar $\kappa_{L,s}$ combines the magnitude of the productivity variation across cycles with the extent to which the selected orientations align that variation in a common direction. The assumption rules out orientation sequences under which positive and negative contributions cancel asymptotically. When the orientation sequence is random, it also requires independence from the outcome errors, since dependence on the errors entering the cycle statistics would invalidate the averaging argument. The assumption does not prescribe a unique orientation rule; it states the conditions that any admissible sequence must satisfy. The rank-based orientation in Definition~\ref{def:rbl} induces such a sequence when the ranking variables satisfy the conditions stated below.

Strong consistency then follows from the strong law of large numbers, as stated in the following theorem.

\begin{theorem} \label{thm:consistency}
    {\normalfont (Strong Consistency)}
    Under Assumption~\ref{ass:consistency}, as $I \to \infty$ and $J \to \infty$,
    \begin{gather*}
        \hat{\beta}_{L,s} \xrightarrow{\mathrm{a.s.}} \beta_0.
    \end{gather*}
\end{theorem}

The graph requirement in Assumption~\ref{ass:cycles} is relatively weak. Under the bipartite Erdős--Rényi model with balanced growth of $I$ and $J$, Supplementary Material~\ref{sec:sm_er} shows that the total number of four-cycles diverges when $\sqrt{IJ} p_{IJ} \to \infty$. By contrast, connectivity requires the stronger condition $\frac{\sqrt{IJ} p_{IJ}}{\log(\sqrt{IJ})} > 1$. Thus, the graph condition required for identification of $\boldsymbol{\alpha}$ and $\boldsymbol{\psi}$ guarantees a diverging total number of four-cycles. 

The effective sample size for estimating $\beta_0$ is $L$, the number of cycles used by the estimator, rather than the number of observed matches. This parallels other econometric settings in which the effective sample size differs from the raw number of observations: in local linear regression, it is proportional to the sample size times the bandwidth, while with clustered data it is determined by the number of clusters.

In labor-market applications, Assumption~\ref{ass:cycles} can hold even when the matching network consists of many disconnected local subgraphs, provided that the number of cycles grows with the numbers of workers and firms. In the empirical setting of \citet{kline2024firm}, for example, the data contain roughly 750{,}000 workers, 70{,}000 firms, and 5{,}000 cycles. This large number of cycles suggests that the asymptotic approximation may be informative about finite-sample behavior.

The next step is to characterize the asymptotic distribution of $\hat\beta_{L,s}$, which provides the basis for inference and for testing no complementarities.

\begin{remark} \label{remark:consistency}
    {\normalfont (No Need to Estimate Productivities)}
    A key feature of $\hat{\beta}_{L,s}$ is that it is consistent without requiring estimation of the productivity collections $\boldsymbol{\alpha}$ and $\boldsymbol{\psi}$. Consistent estimation of these objects would generally require sufficiently many informative observations for each worker and firm. This is in contrast to more standard approaches, such as the iterative least squares method of \citet{bai2009panel}, which estimate $\beta_0$ jointly with unit-specific effects, so the asymptotic behavior of the estimator of $\beta_0$ depends on the properties of the estimators of those effects. By isolating $\beta_0$ directly, the cycle-based estimator avoids this dependence.
\end{remark}

\subsubsection{Asymptotic Normality} \label{sec:inference_beta0}

The estimator $\hat{\beta}_{L,s}$ is a ratio of averages, so its asymptotic distribution can be derived using the Lyapunov central limit theorem. Define
\begin{gather*}
    u_\ell
    \coloneqq
    \epsilon_{1,\ell} + \beta_0 \epsilon_{2,\ell},
    \qquad
    u_{\ell,s}
    \coloneqq
    s_\ell u_\ell,
\end{gather*}
where $\epsilon_{1,\ell}$ and $\epsilon_{2,\ell}$ are defined in Equations~\ref{eq:epsilon1} and~\ref{eq:epsilon2}.

Because $s_\ell^2 = 1$, the variance of $u_{\ell,s}$, conditional on the orientation sequence when it is random, does not depend on the orientation. Define
\begin{gather*}
    \sigma_{u,L}^2
    \coloneqq
    \frac{1}{L} \sum_{\ell=1}^L \Var(u_\ell).
\end{gather*}

With this notation, the asymptotic distribution of $\hat{\beta}_{L,s}$ follows.

\begin{theorem} \label{thm:asympt_normality}
    {\normalfont (Asymptotic Normality)}
    Under Assumption~\ref{ass:consistency}, as $I \to \infty$ and $J \to \infty$,
    \begin{gather*}
        \frac{
            \sqrt{L} \left(\hat{\beta}_{L,s} - \beta_0\right)
        }{
            \sigma_{u,L} / |\kappa_{L,s}|
        }
        \xrightarrow{d}
        \mathcal{N}(0,1).
    \end{gather*}
\end{theorem}

The estimator $\hat{\beta}_{L,s}$ converges at rate $\sqrt{L}$, where $L$, the number of edge-disjoint four-cycles used by the estimator, is the effective sample size. In a complete bipartite graph with even $I$ and $J$, the edges can be partitioned into $IJ/4$ edge-disjoint four-cycles, yielding a rate proportional to $\sqrt{IJ}$. More generally, the asymptotic argument depends on $L \to \infty$ rather than on both dimensions diverging. Thus, the estimator remains consistent if either $I$ or $J$ is fixed while the other grows, provided the graph contains a growing collection of four-cycles.

Two scaling terms appear in the asymptotic distribution. The first, $\sigma_{u,L}$, is the square root of the average variance of the composite errors. It reflects the variability of the underlying outcome errors $\eta_{ij}$ and is invariant to the orientation sequence.

The second, $|\kappa_{L,s}|$, measures the effective identifying signal. It can be small either because the worker and firm productivity gaps within cycles are limited or because the chosen orientations cause positive and negative productivity-gap contributions to cancel. Its presence shows explicitly how orientation affects precision: as the orientation rule aligns productivity heterogeneity more consistently across cycles, $|\kappa_{L,s}|$ increases and the asymptotic variance decreases.

\paragraph{Feasible Estimation of the Scaling Factor}

To make the asymptotic normality result operational for inference, the scaling factor $\sigma_{u,L} / |\kappa_{L,s}|$ must be estimated. The quantity $\hat{\kappa}_{L,s} \coloneqq -\frac{1}{L} \sum_{\ell=1}^L s_\ell \hat{\Delta}_{2,\ell}$ consistently estimates $\kappa_{L,s}$, as established in the proof of Theorem~\ref{thm:asympt_normality}. A feasible estimator of $\sigma_{u,L}^2$ is
\begin{gather} \label{eq:sigma_hat}
    \hat{\sigma}_{u,L,s}^2
    \coloneqq
    \frac{1}{L} \sum_{\ell=1}^L
    \left(\hat{\Delta}_{1,\ell} + \hat{\beta}_{L,s} \hat{\Delta}_{2,\ell}\right)^2,
\end{gather}
whose consistency is established in the following proposition.

\begin{prop} \label{prop:variance_estimator}
    {\normalfont (Consistent Variance Estimator)}
    Under Assumption~\ref{ass:consistency}, as $I \to \infty$ and $J \to \infty$,
    \begin{gather*}
        \hat{\sigma}_{u,L,s}^2 - \sigma_{u,L}^2
        \xrightarrow{p}
        0.
    \end{gather*}
\end{prop}

When using the rank-based orientation rule, I denote these estimators evaluated at $\boldsymbol{s}_{L,z}$ by $\hat{\kappa}_{L,z}$ and $\hat{\sigma}_{u,L,z}^2$.

Theorem~\ref{thm:asympt_normality} and these feasible estimators of the scaling terms yield asymptotically valid confidence intervals for $\beta_0$ whenever the orientation sequence satisfies Assumption~\ref{ass:labeling}. I next provide conditions under which the rank-based orientation rule induces such a sequence.

\subsubsection{Validity of Rank-Based Orientation}

I now provide sufficient conditions under which the rank-based orientation sequence satisfies Assumption~\ref{ass:labeling}. To clarify the role of the ranking variables, first consider the oracle case in which they coincide with the latent productivities: $\boldsymbol{z}^{\alpha} = \boldsymbol{\alpha}$ and $\boldsymbol{z}^{\psi} = \boldsymbol{\psi}$. The induced orientation then agrees with the sign of the productivity-gap product in every cycle, so $s_{\ell,z} A_\ell = |A_\ell|$. It follows that $\kappa_{L,z} \coloneqq \frac{1}{L} \sum_{\ell=1}^L s_{\ell,z} A_\ell = \frac{1}{L} \sum_{\ell=1}^L |A_\ell| = \mu_L > C_\mu$, and Assumption~\ref{ass:labeling} holds.

The rank-based orientation rule does not require the ranking variables to recover the latent ranking in every cycle. Some cycles may be assigned a sign opposite to that of $A_\ell$, provided that correctly oriented productivity-gap variation dominates incorrectly oriented variation in the aggregate. Equivalently, $\kappa_{L,z} = \frac{1}{L} \sum_{\ell=1}^L s_{\ell,z} A_\ell$ must remain bounded away from zero. The conditions below ensure this by requiring each ranking variable to contain useful information about the corresponding latent ranking, the information from the two sides not to systematically offset across cycles, and orientation errors not to be concentrated among cycles with the largest productivity gaps.

To formalize these requirements, define the worker-side and firm-side ranking-agreement signs as
\begin{align*}
    r_{\ell,z}^{\alpha}
    &\coloneqq
    \operatorname{sign}\left[
        \left(z^\alpha_{i_{\ell 1}} - z^\alpha_{i_{\ell 2}}\right)
        \left(\alpha_{i_{\ell 1}} - \alpha_{i_{\ell 2}}\right)
    \right], \\
    r_{\ell,z}^{\psi}
    &\coloneqq
    \operatorname{sign}\left[
        \left(z^\psi_{j_{\ell 1}} - z^\psi_{j_{\ell 2}}\right)
        \left(\psi_{j_{\ell 1}} - \psi_{j_{\ell 2}}\right)
    \right].
\end{align*}
The sign $r_{\ell,z}^{\alpha}$ equals one when the worker ranking variable and the worker productivity difference have the same sign and equals minus one when they have opposite signs. Thus, $r_{\ell,z}^{\alpha}$ records whether the ranking variable agrees with the latent worker ordering in cycle $\ell$. The sign $r_{\ell,z}^{\psi}$ has the analogous interpretation for firms.

Define the overall ranking-agreement sign as
\begin{gather*}
    r_{\ell,z} \coloneqq r_{\ell,z}^{\alpha} r_{\ell,z}^{\psi},
\end{gather*}
and note that, by construction, $s_{\ell,z} A_\ell = r_{\ell,z} |A_\ell|$. Thus, $r_{\ell,z}$ determines the sign with which cycle $\ell$ contributes to the identifying variation, while $|A_\ell|$ determines the magnitude of that contribution.

Define the average ranking-agreement signs
\begin{gather*}
    \bar{r}_L^\alpha
    \coloneqq
    \frac{1}{L} \sum_{\ell=1}^L r_{\ell,z}^{\alpha},
    \qquad
    \bar{r}_L^\psi
    \coloneqq
    \frac{1}{L} \sum_{\ell=1}^L r_{\ell,z}^{\psi},
    \qquad
    \bar{r}_L
    \coloneqq
    \frac{1}{L} \sum_{\ell=1}^L r_{\ell,z},
\end{gather*}
and impose the following conditions on the sequences $\{r_{\ell,z}^{\alpha}\}$ and $\{r_{\ell,z}^{\psi}\}$.

\begin{ass} \label{ass:instruments}
    \begin{enumerate}[label=4.\arabic*]
        \item[]

        \item \label{ass:relevance}
        {\normalfont (Ranking Relevance)}
        The ranking variables agree sufficiently often with the corresponding latent productivity orderings: there exist constants $c_\alpha,c_\psi > 0$ such that, for all sufficiently large $L$,
        \begin{gather*}
            \bar{r}_L^\alpha \geq c_\alpha,
            \qquad
            \bar{r}_L^\psi \geq c_\psi.
        \end{gather*}

        \item \label{ass:nonegative}
        {\normalfont (No Negative Association)}
        Worker-side and firm-side ranking agreements are not negatively associated across cycles: for all sufficiently large $L$,
        \begin{gather*}
            \frac{1}{L} \sum_{\ell=1}^L
            \left(r_{\ell,z}^{\alpha} - \bar{r}_L^\alpha\right)
            \left(r_{\ell,z}^{\psi} - \bar{r}_L^\psi\right)
            \geq 0.
        \end{gather*}

        \item \label{ass:nolargegaps}
        {\normalfont (No Large-Gap Penalty)}
        Overall ranking agreement is not negatively associated with the magnitude of the productivity-gap product: for all sufficiently large $L$,
        \begin{gather*}
            \frac{1}{L} \sum_{\ell=1}^L
            \left(|A_\ell| - \mu_L\right)
            \left(r_{\ell,z} - \bar{r}_L\right)
            \geq 0.
        \end{gather*}
    \end{enumerate}
\end{ass}

Assumption~\ref{ass:relevance} requires each ranking variable to contain useful information about the corresponding latent productivity ordering. On the worker side, agreement occurs when the difference in the ranking variable has the same sign as the difference in productivity; disagreement occurs when the signs are opposite. The firm-side condition has the same interpretation. Perfect rankings are not required: each ranking variable must agree with the corresponding latent ordering more often than it disagrees, with this advantage bounded away from zero as $L$ grows.

Assumption~\ref{ass:nonegative} requires the ranking information on the worker and firm sides not to systematically offset across cycles. The overall cycle orientation is correct when the two ranking variables either both agree with or both reverse the corresponding latent orderings, because two reversals preserve the sign of the productivity-gap product. It is incorrect when one ranking variable agrees and the other reverses the latent ordering. Thus, individual ranking relevance does not guarantee a valid combined orientation if agreement on one side systematically coincides with disagreement on the other. The assumption rules out this offsetting and ensures that, before accounting for differences in cycle strength, correctly oriented cycles dominate incorrectly oriented ones.

Assumption~\ref{ass:nolargegaps} requires the overall orientation to remain informative among cycles with the strongest identifying variation. A cycle contributes more to identification when the absolute product of its worker and firm productivity gaps is larger. Positive average orientation agreement would therefore be insufficient if incorrectly oriented cycles systematically had larger productivity gaps than correctly oriented cycles. The assumption rules out this pattern. It allows orientation errors, including among highly informative cycles, but requires them not to become systematically more prevalent as $|A_\ell|$ increases. Consequently, weighting cycles by the strength of their identifying variation does not overturn the positive orientation agreement implied by Assumptions~\ref{ass:relevance} and~\ref{ass:nonegative}.

The following proposition establishes that these conditions are sufficient for the rank-based orientation sequence to satisfy Assumption~\ref{ass:labeling}.

\begin{prop} \label{prop:instrument_validity}
    {\normalfont (Validity of Rank-Based Orientation)}
    Under Assumptions~\ref{ass:error_eta}, \ref{ass:cycles}, \ref{ass:seq_mu}, and~\ref{ass:instruments}, the rank-based orientation sequence $\boldsymbol{s}_{L,z}$ satisfies
    \begin{gather*}
        \kappa_{L,z}
        \coloneqq
        \frac{1}{L} \sum_{\ell=1}^L s_{\ell,z} A_\ell
        \geq
        \mu_L c_\alpha c_\psi
        >
        0
    \end{gather*}
    for all sufficiently large $L$. Hence, $\boldsymbol{s}_{L,z}$ satisfies Assumption~\ref{ass:labeling}.
\end{prop}

Proposition~\ref{prop:instrument_validity} connects the feasible estimator $\hat{\beta}_{L,z}$ to the generic theory. Under the stated conditions, $\hat{\beta}_{L,z}$ is strongly consistent and asymptotically normal and admits a feasible variance estimator. It can therefore be used not only to estimate the interaction parameter but also to construct a formal test of modularity. The next section develops this test and studies its properties.

\subsubsection{Testing No Complementarities} \label{sec:modularity_test}

The TWFE model is nested within the Tukey model as the special case $\beta_0 = 0$. The asymptotic distribution of the rank-based estimator $\hat{\beta}_{L,z}$ can therefore be used to test the null hypothesis $H_0 : \beta_0 = 0$, which corresponds to modularity of the interaction function. Rejecting $H_0$ therefore rejects the TWFE specification and provides evidence against modularity. Failure to reject does not establish modularity outside the Tukey class, since richer non-modular interaction functions may not be detected by a test based solely on $\beta_0$. To the best of my knowledge, no formal test of modularity has previously been developed in the BI framework, even though whether complementarities are present is a recurring question in empirical applications.

Testing $H_0 : \beta_0 = 0$ is the focus of \citet{tukey1949one}, although that procedure does not estimate $\beta_0$. Under homoskedastic, normally distributed errors $\eta_{ij}$ and a complete matching network $G_{IJ}$, Tukey's procedure first estimates the additive effects $\boldsymbol{\alpha}$ and $\boldsymbol{\psi}$ under the null using a two-way fixed effects regression and then tests whether the interaction term $\alpha_i \psi_j$ significantly improves model fit.

Related diagnostic approaches have been used in the TWFE literature, including \citet{card2013workplace} and \citet{fenizia2022managers}. In sparse matching networks, however, consistent estimation of the individual productivities may not be feasible, which complicates residual-based procedures. By contrast, under the conditions developed above, the rank-based estimator $\hat{\beta}_{L,z}$ remains consistent with heteroskedastic errors and sparse matching networks.

\paragraph{Test Statistic.}
Consider the studentized statistic
\begin{gather*}
    \hat{T}_{L,z}
    \coloneqq
    \frac{\sqrt{L} \hat{\kappa}_{L,z} \hat{\beta}_{L,z}}{\hat{\sigma}_{u,L,z}} 
    =
    \frac{
        \sum_{\ell=1}^L s_{\ell,z} \hat{\Delta}_{1,\ell}
    }{
        \sqrt{
            \sum_{\ell=1}^L
            \left(\hat{\Delta}_{1,\ell} + \hat{\beta}_{L,z} \hat{\Delta}_{2,\ell}\right)^2
        }
    }.
\end{gather*}
For a nominal size $\gamma$, define the test
\begin{gather*}
    \phi_{L,z}(\hat{T}_{L,z},\gamma)
    =
    \ind\left\{
        |\hat{T}_{L,z}| \geq z_{1-\gamma/2}
    \right\},
\end{gather*}
where $z_{1-\gamma/2}$ is the $(1-\gamma/2)$ quantile of the standard normal distribution. The asymptotic validity and consistency of the test follow from Theorem~\ref{thm:asympt_normality}, as summarized in the following corollary.

\begin{cor} \label{cor:test_modularity}
    {\normalfont (Test of Modularity within the Tukey Model)}
    Under Assumptions~\ref{ass:error_eta}, \ref{ass:cycles}, \ref{ass:seq_mu}, and~\ref{ass:instruments}, as $I \to \infty$ and $J \to \infty$, the test $\phi_{L,z}(\hat{T}_{L,z},\gamma)$ is asymptotically valid under $H_0 : \beta_0 = 0$:
    \begin{gather*}
        \lim_{L \to \infty}
        \E\left[
            \phi_{L,z}(\hat{T}_{L,z},\gamma)
        \right]
        =
        \gamma.
    \end{gather*}
    Under the alternative $H_1 : \beta_0 \neq 0$, the test is consistent:
    \begin{gather*}
        \lim_{L \to \infty}
        \E\left[
            \phi_{L,z}(\hat{T}_{L,z},\gamma)
        \right]
        =
        1.
    \end{gather*}
\end{cor}

The test controls asymptotic size under modularity, but Theorem~\ref{thm:asympt_normality} establishes consistency only when the interaction function follows the Tukey model. Outside this class, a non-modular interaction function may generate a zero value for the cycle-based parameter targeted by the test, leaving the test with no asymptotic power against that alternative. Thus, $\phi_{L,z}$ tests a necessary but not sufficient condition for modularity. This is analogous to testing independence using a correlation coefficient: nonzero correlation implies dependence, but zero correlation does not imply independence. Accordingly, rejection provides evidence against modularity, whereas failure to reject does not establish that the interaction function is modular.

\subsection{Discussion: Estimating Agent Productivities}
\label{sec:estimation_inference_prod}

Theorem~\ref{thm:tukey_identification} shows that, once $\beta_0$ is known, connectedness of the matching network is sufficient to identify the productivity collections $\boldsymbol{\alpha}$ and $\boldsymbol{\psi}$. Identification, however, does not ensure precise estimation of these high-dimensional parameters when the network is sparse. If each agent is observed in only a small number of matches, the information available about an individual productivity remains limited even as the overall network grows. This difficulty is not specific to the Tukey model. In the TWFE model, consistent estimation of individual fixed effects generally requires graph conditions substantially stronger than those needed for identification and often not satisfied in sparse empirical networks \citep{jochmans2019fixed,kline2020leave}.

Supplementary Material~\ref{sec:sm_prod} outlines an alternating least-squares procedure that uses $\hat{\beta}_{L,z}$ as an input. Establishing the asymptotic properties of the resulting productivity estimators requires additional conditions on the matching network and an analysis of how first-stage estimation error in $\hat{\beta}_{L,z}$ propagates to the second stage. I leave this analysis, including consistency for the individual productivities and relevant functions of them, to future work.

\section{Empirical Illustration} \label{sec:empirical_illustration}

In this section, I revisit the application in \citet{limodio2021bureaucrat}\footnote{The data used for this empirical illustration are publicly available on the author's website.} to illustrate how the interaction parameter in the Tukey model can be estimated in practice and what it reveals about the complementarity structure of a two-sided interaction.

\citet{limodio2021bureaucrat} studies the interaction between managers and tasks in the public sector, focusing on the assignment of World Bank bureaucrats (hereafter, managers) to development projects in low- and middle-income countries. Each manager is responsible for designing a project and supervising its implementation. Project success is measured using ratings produced by the World Bank's Independent Evaluation Group, which assess the extent to which the project's main objectives were achieved. The original analysis models project success using a TWFE specification with a manager-specific ability component $\alpha_i$ and a country-specific component $\psi_j$. The administrative data record manager-country assignments over time, together with project characteristics and evaluations, allowing the matching network to be constructed and the outcome associated with each realized match to be observed.

The combination of assignment data and standardized performance evaluations provides a useful setting for studying the allocation of bureaucrats across tasks. The main result in \citet{limodio2021bureaucrat} is negative sorting: high-performing managers are disproportionately assigned to low-performing countries. Several mechanisms may contribute to this pattern, including the Bank's objective of assigning stronger managers to weaker countries, career incentives and promotion dynamics, demand for specialized skills in more difficult environments, and reallocations following adverse shocks such as natural disasters.

For this illustration, I use a simplified empirical specification without controls. The main specification in \citet{limodio2021bureaucrat} also includes year and sector fixed effects. The estimates below should therefore be interpreted within the simplified specification and are not intended as a direct replication of the original results.

The data contain 3{,}385 projects corresponding to 1{,}876 distinct manager-country pairs. When the same pair is observed in multiple projects, I average the project outcomes to obtain a single observation $y_{ij}$ for that edge. The resulting matching network contains 697 managers, 127 countries, and 1{,}876 edges. Within this network, I select 228 edge-disjoint four-cycles involving 369 managers and 114 countries. The cycles used by the estimator therefore contain approximately half of the network's edges and managers and about 90\% of its countries. These cycles provide the identifying variation used to estimate the Tukey interaction parameter. To aggregate this variation, the cycle statistics must first be assigned a coherent orientation.

I estimate $\beta_0$ using the rank-based estimator $\hat{\beta}_{L,z}$. For managers, I use average project size, measured by the average loan amount of the projects they oversee. For countries, I use the Public Investment Management Index (PIMI) developed by \citet{dabla2012investing}. These ranking variables determine the orientation $s_{\ell,z}$ assigned to each cycle. Their primary justification concerns Assumption~\ref{ass:relevance}, which requires each variable to contain information about the corresponding latent productivity ordering. \citet{limodio2021bureaucrat} documents that more capable managers tend to be assigned larger loans, suggesting that average project size ranks managers in the same direction as $\alpha_i$. The PIMI measures the quality of the institutions governing the appraisal, selection, implementation, and evaluation of public investment projects, making it a natural ranking variable for the country productivity component $\psi_j$. Neither variable must recover the latent ranking perfectly; each must agree with the corresponding productivity ordering more often than it disagrees.

The remaining assumptions concern the pattern and consequences of ranking errors. Assumption~\ref{ass:nonegative} requires errors in the two rankings not to systematically offset one another. The cycle orientation is correct when both ranking variables agree with the corresponding latent rankings or when both reverse them, but it is incorrect when only one ranking is reversed. Assumption~\ref{ass:nolargegaps} requires orientation errors not to be disproportionately concentrated among cycles with the largest productivity-gap products, which contain the most identifying variation. Although these restrictions cannot be verified directly because the productivities are latent, their violation would require unusual patterns in the performance of the ranking variables. Violating Assumption~\ref{ass:nolargegaps}, for example, would require the ranking variables to fail systematically precisely when the underlying productivity differences are largest, even though economically informative proxies would generally be expected to distinguish agents more reliably when those differences are pronounced. Similarly, violating Assumption~\ref{ass:nonegative} would require the accuracy of the manager-side ranking to be systematically negatively associated with the accuracy of the country-side ranking within the selected cycles. Because average loan size and the PIMI have clear economic relationships with the corresponding latent productivity components and are constructed from distinct information, they provide plausible ranking variables for inducing an orientation sequence that satisfies Assumption~\ref{ass:labeling}.

Figure~\ref{fig:estimates} reports the rank-based estimate $\hat{\beta}_{L,z} = -0.196$, together with its nominal 90\% confidence interval, $(-0.293,-0.099)$. The interval excludes zero, and the $p$-value for the null hypothesis $\beta_0 = 0$ is 0.001, providing evidence against the null of no complementarities. Within the Tukey specification, the negative estimate indicates negative complementarities between managers and countries: the relative contribution of a higher-productivity manager is greater in a lower-productivity country. The magnitude is also economically meaningful. Supplementary Material~\ref{sec:interpretation} maps the estimate into a representative decomposition of the average outcome and shows that, under a symmetric benchmark in which managers and countries contribute equally to the additive component, the multiplicative component accounts for approximately 22\% of the gross magnitude of the decomposition.

\begin{figure}[t]
    \centering
    \includegraphics[width=0.65\textwidth]{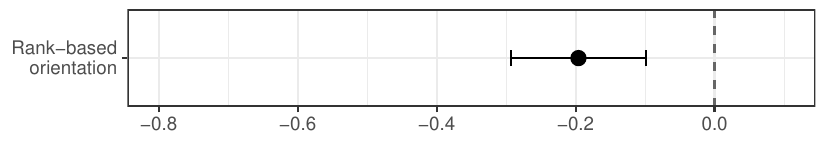}
    \caption{\small Estimate of the interaction parameter under the rank-based orientation rule. The interval is a nominal 90 percent asymptotic confidence interval for $\beta_0$.}
    \label{fig:estimates}
\end{figure}

To illustrate the role of orientation, I recompute the estimator using 50{,}000 independently generated sequences of random signs $s_\ell \in \{-1,1\}$. This exercise provides a descriptive benchmark for orientations that contain no information about the latent productivity rankings. Under random orientations, the signed cycle contributions tend to cancel, so the estimated orientation signal $\hat{\kappa}_{L,s}$ is typically close to zero. Because $\hat{\kappa}_{L,s}$ determines the denominator of the estimator, this cancellation generates highly unstable estimates of $\beta_0$. Figure~\ref{fig:empirical_ill} reports the random-orientation distributions of $\hat{\kappa}_{L,s}$ and $\hat{\beta}_{L,s}$, with the values obtained from the rank-based orientation highlighted in red.

As expected, the random-orientation distribution of $\hat{\kappa}_{L,s}$ is centered near zero, while the corresponding estimates of $\beta_0$ are widely dispersed because many random sequences produce a denominator close to zero. The orientation based on average loan size and the PIMI yields an absolute signal $|\hat{\kappa}_{L,z}|$ larger than approximately 80\% of those generated by random signs. Thus, the rank-based orientation aligns the observed cycle statistics more strongly than a typical uninformed orientation. This diagnostic does not establish the validity of the ranking variables, but it suggests that they provide useful orientation information and limit cancellation of the identifying variation across cycles.

\begin{figure}[t]
    \centering
    \includegraphics[width=0.70\textwidth]{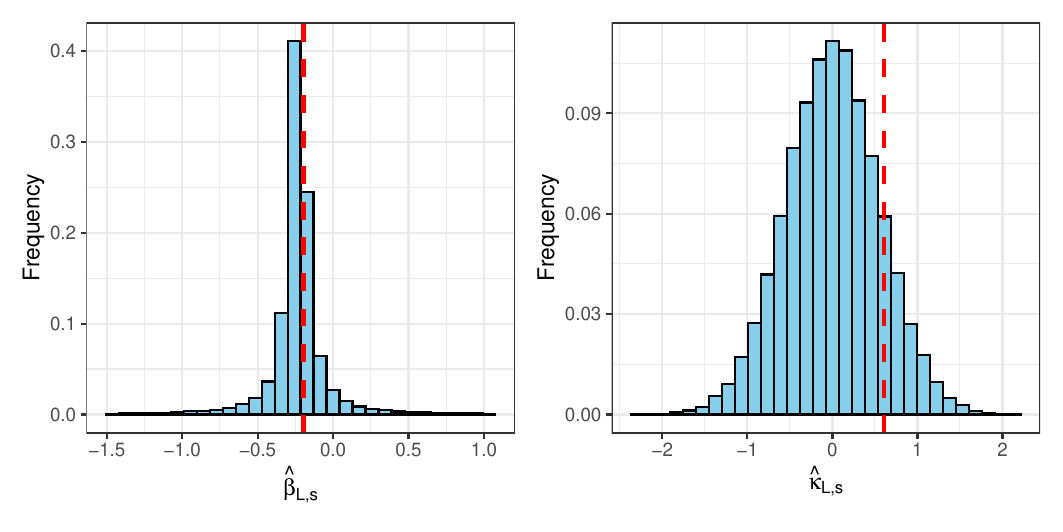}
    \caption{\small Empirical distributions of $\hat{\beta}_{L,s}$ and $\hat{\kappa}_{L,s}$ across random orientation sequences. For the distribution of $\hat{\beta}_{L,s}$, I exclude the 2.8\% of draws outside $(-1.5,1)$ to make the histogram informative. Across all draws, the minimum and maximum estimates are $-413.589$ and $1543.647$.}
    \label{fig:empirical_ill}
\end{figure}

The estimated negative complementarities provide a possible explanation for the negative sorting documented by \citet{limodio2021bureaucrat}. If higher-productivity managers contribute relatively more in lower-productivity countries, assigning stronger managers to more difficult environments may improve project performance. \citet{limodio2021bureaucrat} proposes this mechanism, but it cannot be assessed within the modular TWFE specification. Establishing whether the observed assignment is efficient or increases aggregate performance would require estimates of manager and country productivities under the Tukey model and an analysis of counterfactual assignments.

Overall, the empirical illustration shows that the Tukey model can be implemented in a setting where TWFE is commonly used and can uncover an economically meaningful feature of the interaction that the additive model rules out by construction. It also highlights the practical importance of orientation: the cycle-based estimator is informative only when the selected ranking variables align the identifying variation sufficiently to prevent cancellation across cycles.

\section{Conclusion} \label{sec:conclusion}

This paper studied how complementarities can be identified, estimated, and tested in two-sided interaction models when only a sparse subset of potential matches is observed. The analysis separated the interaction function, which determines how latent agent characteristics jointly generate outcomes, from the matching network, which determines what can be learned about that function from the observed matches. The identification results reveal a fundamental trade-off between these two components: richer interaction functions require more informative network structures. Because the matching network is observed, this trade-off provides a direct way to assess whether the available data can support a proposed model of complementarities.

I centered the analysis on the Tukey model, which extends TWFE by summarizing complementarities with a single interaction parameter. Identification requires only a modest strengthening of the conditions for TWFE: in addition to connectedness, the matching network must contain an informative four-cycle. I developed a cycle-based estimator of the interaction parameter that does not require estimating the latent productivities and remains consistent and asymptotically normal in sparse networks with bounded node degrees. Its limiting distribution provides the basis for a formal test of no complementarities within the Tukey model. The empirical illustration showed that the model can be implemented in a setting where TWFE is commonly used and can reveal an economically meaningful feature of the interaction that the additive model rules out by construction.

The estimation and inference analysis focused on the Tukey interaction parameter. Researchers may also be interested in estimating the productivity collections $\boldsymbol{\alpha}$ and $\boldsymbol{\psi}$ or functions of them, such as sorting measures and variance decompositions. I outline an alternating least-squares procedure for this purpose. Analyzing its properties requires additional conditions on the matching network and an account of first-stage estimation error. The BI framework, by clarifying the distinct roles of the interaction function and the matching network, provides a foundation for studying these and other related questions.

\appendix

\section*{Appendix Roadmap}

Appendix~\ref{appendix:approximation} studies the TWFE approximation to the Tukey model in the noiseless case. Appendix~\ref{appendix:outcome_based_labeling} explains why assigning cycle orientations using realized outcomes is invalid. Appendix~\ref{appendix:lemmas} collects auxiliary lemmas, and Appendix~\ref{appendix:proofs} contains proofs of the results stated in the main text.

\section{TWFE as an Approximation to the Tukey Model} \label{appendix:approximation}

The TWFE model is often interpreted as an approximation, analogous to using the best linear projection to summarize the relationship between two random variables. This Appendix studies the behavior of this approximation when the worker-firm interaction follows the Tukey model.

Let $n_i = \sum_j D_{ij}$ and $m_j = \sum_i D_{ij}$ denote the degrees of worker $i$ and firm $j$, respectively, and let $n = \sum_i n_i$ denote the total number of observed matches. Using this notation, I define the TWFE projections.

\begin{df} \label{def:twfe_projection}
    {\normalfont (TWFE Projection)}
    Under Assumption~\ref{ass:connectedness} and the normalization $\alpha_1 = 0$, the TWFE projections $(\boldsymbol{\alpha}^{twfe},\boldsymbol{\psi}^{twfe})$ are defined by
    \begin{gather*}
        (\boldsymbol{\alpha}^{twfe},\boldsymbol{\psi}^{twfe})
        =
        (C'C)^{-1} C' \boldsymbol{\theta}_{\mathcal{O}},
    \end{gather*}
    where $C \in \R^{n \times (I+J-1)}$ is the design matrix linking the free worker productivities $(\alpha_2,\dots,\alpha_I)$ and the firm productivities $(\psi_1,\dots,\psi_J)$ to the observed matches.
\end{df}

These projections are the coefficients from a TWFE regression in the noiseless model. Each row of $C$ identifies the worker and firm in one observed match, omitting the worker indicator when $i = 1$ because of the normalization. The matrix $C'C$ has block form
\begin{gather*}
    C'C
    =
    \begin{pmatrix}
        \operatorname{diag}(n_2,\dots,n_I) & D \\
        D' & \operatorname{diag}(m_1,\dots,m_J)
    \end{pmatrix},
\end{gather*}
where $D$ is the worker-firm incidence block with the row corresponding to worker 1 removed. Thus, $C'C$ is the reduced signless Laplacian of $G_{IJ}$ and reflects the connectivity structure of the matching network.

Under the TWFE model, $(\boldsymbol{\alpha}^{twfe},\boldsymbol{\psi}^{twfe})$ coincide with the latent productivities $(\boldsymbol{\alpha},\boldsymbol{\psi})$. Under the Tukey model, their bias depends on $C'C$ and therefore on the full matching network, but admits a closed-form characterization.

For each worker $i$ and firm $j$, define the average partner productivities
$\bar{\psi}_i = n_i^{-1}\sum_j D_{ij}\psi_j$ and
$\bar{\alpha}_j = m_j^{-1}\sum_i D_{ij}\alpha_i$.
The next proposition characterizes the TWFE projections under the Tukey model.

\begin{prop} \label{prop:misspecification}
    {\normalfont (Misspecified TWFE)}
    Under the~\ref{eq:tukey}, the normalization $\alpha_1 = 0$, and Assumption~\ref{ass:connectedness}, for $i > 1$ and all $j$,
    \begin{align*}
        \alpha_i^{twfe}
        &=
        \alpha_i
        +
        \beta_0
        \left(
            \sum_{i'=2}^{I}
            \Lambda_{i-1,i'-1} n_{i'} \alpha_{i'} \bar{\psi}_{i'}
            +
            \sum_{j'=1}^{J}
            \Lambda_{i-1,I-1+j'} m_{j'} \psi_{j'} \bar{\alpha}_{j'}
        \right), \\
        \psi_j^{twfe}
        &=
        \psi_j
        +
        \beta_0
        \left(
            \sum_{i'=2}^{I}
            \Lambda_{I-1+j,i'-1} n_{i'} \alpha_{i'} \bar{\psi}_{i'}
            +
            \sum_{j'=1}^{J}
            \Lambda_{I-1+j,I-1+j'} m_{j'} \psi_{j'} \bar{\alpha}_{j'}
        \right),
    \end{align*}
    where $\Lambda = (C'C)^{-1}$, the inverse of the reduced signless Laplacian of $G_{IJ}$.
\end{prop}

\begin{proof}
    Let $A \in \mathbb{R}^{IJ \times I}$ and $B \in \mathbb{R}^{IJ \times J}$ be the worker and firm indicator matrices, let $Y \in \mathbb{R}^{IJ}$ stack the $\theta_{ij}$, and let $S \in \mathbb{R}^{n \times IJ}$ select the observed matches. After imposing $\alpha_1 = 0$, define
    \begin{gather*}
        C \coloneqq [SA_{-1}\;\;SB],
        \qquad
        c \coloneqq
        [\alpha_2,\dots,\alpha_I,\psi_1,\dots,\psi_J]',
    \end{gather*}
    where $A_{-1}$ omits the first column of $A$. Then $\boldsymbol{\theta}_{\mathcal{O}} = SY$ and
    \begin{gather*}
        c^{twfe}
        \coloneqq
        [\alpha_2^{twfe},\dots,\alpha_I^{twfe},
        \psi_1^{twfe},\dots,\psi_J^{twfe}]'
        =
        (C'C)^{-1}C'\boldsymbol{\theta}_{\mathcal{O}}.
    \end{gather*}

    Let $a \coloneqq SA\boldsymbol{\alpha}$, $p \coloneqq SB\boldsymbol{\psi}$, and $h \coloneqq a \odot p$. Under the~\ref{eq:tukey},
    \begin{gather*}
        \boldsymbol{\theta}_{\mathcal{O}}
        =
        Cc+\beta_0 h,
        \qquad
        c^{twfe}
        =
        c+\beta_0(C'C)^{-1}C'h.
    \end{gather*}
    The first $I-1$ entries of $C'h$ are
    \begin{gather*}
        \sum_j D_{ij}\alpha_i\psi_j
        =
        n_i\alpha_i\bar{\psi}_i,
        \qquad i=2,\dots,I,
    \end{gather*}
    while its remaining $J$ entries are
    \begin{gather*}
        \sum_i D_{ij}\alpha_i\psi_j
        =
        m_j\psi_j\bar{\alpha}_j,
        \qquad j=1,\dots,J.
    \end{gather*}
    Substituting these expressions and $\Lambda=(C'C)^{-1}$ yields the result.
\end{proof}

The matrix $\Lambda=(C'C)^{-1}$, the inverse of the reduced signless Laplacian, encodes the global connectivity of the matching network. Its entries determine how the omitted interaction components associated with each worker and firm enter every coefficient of the TWFE projection. This interpretation is preferable to viewing $\Lambda_{u,v}$ as a simple measure of proximity, since its magnitude and sign need not depend monotonically on path length or node degree.

The bias in Proposition~\ref{prop:misspecification} is linear in $\beta_0$ and vanishes when $\beta_0=0$. The terms $n_{i'}\alpha_{i'}\bar{\psi}_{i'}$ and $m_{j'}\psi_{j'}\bar{\alpha}_{j'}$ are node-level sums of the multiplicative component omitted by TWFE, while the weights $\Lambda_{u,v}$ distribute these terms across the projected worker and firm effects. Consequently, each coefficient's bias depends on the entire matching network rather than only on local degrees or partners. Closed-form expressions for $\Lambda$ are available for some graph families \citep{hessert2021moore}, but in general it must be computed from the observed network.

\subsection{Bias in Sorting Measures}

Because the bias at each node depends on the full pattern of observed matches, workers with identical productivities may have different TWFE projections. The distortion can even reverse rankings: $\alpha_i > \alpha_{i'}$ may coexist with $\alpha_i^{twfe} < \alpha_{i'}^{twfe}$. This contrasts with the balanced-panel case, in which the TWFE projections preserve ordinal information up to a global sign.

These distortions also affect sorting measures based on the projected effects. A common measure is the correlation between $\alpha_i^{twfe}$ and $\psi_j^{twfe}$ across observed matches, with positive values interpreted as positive sorting. Proposition~\ref{prop:misspecification} can be used to construct cases in which this correlation differs sharply from its counterpart based on the latent productivities. In the example below, the latent productivities have a correlation of approximately $0.3$ across matches, whereas the correlation based on the TWFE projections is close to zero when $\beta_0 = 3$. Thus, despite positive sorting and a supermodular interaction function, the TWFE approximation suggests the absence of sorting.

\begin{eg} \label{example:correlation}
    {\normalfont (Sorting Pattern)}
    Consider the matching network with $I=J=3$ and edges
    $\mathcal{O}_{IJ}
    =
    \{(i_1,j_1), (i_1,j_2), (i_1,j_3), (i_2,j_2), (i_2,j_3),(i_3,j_3)\}$,
    shown in Figure~\ref{fig:bias}. Let the worker productivities be
    $(\alpha_1,\alpha_2,\alpha_3)=(4,5,2)$ and the firm productivities be
    $(\psi_1,\psi_2,\psi_3)=(10,8,1)$.

    \begin{figure}[ht]
        \centering
        \begin{minipage}{0.45\textwidth}
            \centering
            \begin{tikzpicture}[
                scale=0.8,
                >=stealth,
                font=\small,
                pointI/.style={circle,fill=purple,draw=black,inner sep=3pt},
                pointJ/.style={circle,fill=green,draw=black,inner sep=3pt}
            ]
                \foreach \y/\lab in {1/1,0/2,-1/3}{
                    \node[pointI] (I\lab) at (0,\y) {};
                    \node[pointJ] (J\lab) at (4,\y) {};
                }
                \node[left=4pt of I1] {$i_1$};
                \node[left=4pt of I2] {$i_2$};
                \node[left=4pt of I3] {$i_3$};
                \node[right=4pt of J1] {$j_1$};
                \node[right=4pt of J2] {$j_2$};
                \node[right=4pt of J3] {$j_3$};
                \draw (I1) -- (J1);
                \draw (I1) -- (J2);
                \draw (I1) -- (J3);
                \draw (I2) -- (J2);
                \draw (I2) -- (J3);
                \draw (I3) -- (J3);
            \end{tikzpicture}
        \end{minipage}
        \hfill
        \begin{minipage}{0.45\textwidth}
            \centering
            \includegraphics[width=\linewidth]{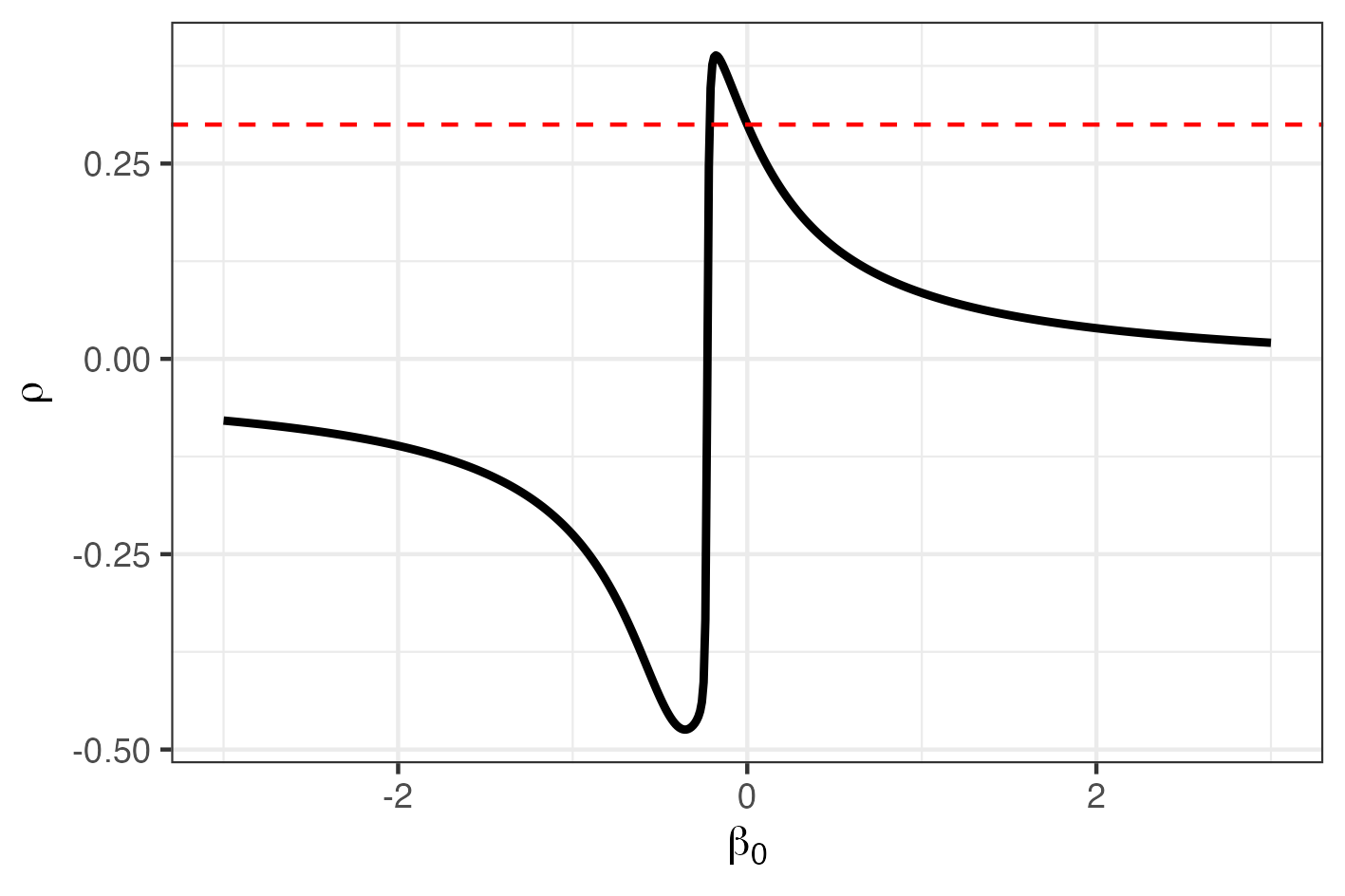}
        \end{minipage}
        \caption{\small Left: matching network in Example~\ref{example:correlation}. Right: correlation $\rho$ between the worker and firm TWFE projections across observed matches as a function of $\beta_0$.}
        \label{fig:bias}
    \end{figure}

    After imposing the observationally equivalent normalization $\alpha_1=0$, Proposition~\ref{prop:misspecification} gives the TWFE projections. Figure~\ref{fig:bias} reports their correlation across observed matches as a function of $\beta_0$. The correlation between the latent productivities is approximately $0.3$ and is recovered when $\beta_0=0$, but the correlation between the projections can be close to zero or negative. In particular, when $\beta_0=3$, the interaction function is supermodular and the projected correlation is only $0.02$, as the misspecification bias nearly offsets the underlying positive sorting.
\end{eg}

Bias in the covariance also distorts variance decompositions, making their components difficult to interpret. More generally, because TWFE projections may misrank agents, two-stage procedures that use them as dependent variables in subsequent regressions can also be misleading.

Example~\ref{example:correlation} hence demonstrates that approximating the Tukey model with the TWFE specification, even in the noiseless case, produces parameters of limited informational value once complementarities are non-negligible. This underscores the importance of methods that allow $f$ to feature complementarities. The Tukey model provides the minimal extension of TWFE required to capture them, offering a simple yet flexible structure for analysis.

\section{Invalidity of Outcome-Based Orientation}
\label{appendix:outcome_based_labeling}

The main text uses external ranking variables to orient the cycles. A natural alternative is to construct the orientation directly from observed outcomes. For each cycle $\ell$, define
\begin{gather*}
    d_{\ell,y}^{\alpha}
    \coloneqq
    y_{i_{\ell 1}j_{\ell 1}}
    +
    y_{i_{\ell 1}j_{\ell 2}}
    -
    y_{i_{\ell 2}j_{\ell 1}}
    -
    y_{i_{\ell 2}j_{\ell 2}}, \qquad
    d_{\ell,y}^{\psi}
    \coloneqq
    y_{i_{\ell 1}j_{\ell 1}}
    +
    y_{i_{\ell 2}j_{\ell 1}}
    -
    y_{i_{\ell 1}j_{\ell 2}}
    -
    y_{i_{\ell 2}j_{\ell 2}}.
\end{gather*}
and let
\begin{gather*}
    s_{\ell,y}
    \coloneqq
    \operatorname{sign}
    \left(
        d_{\ell,y}^{\alpha}d_{\ell,y}^{\psi}
    \right).
\end{gather*}
The resulting estimator is
\begin{gather*}
    \hat{\beta}_{L,y}
    \coloneqq
    -
    \frac{
        \sum_{\ell=1}^L s_{\ell,y}\hat{\Delta}_{1,\ell}
    }{
        \sum_{\ell=1}^L s_{\ell,y}\hat{\Delta}_{2,\ell}
    }.
\end{gather*}
This rule violates Assumption~\ref{ass:labeling} because $s_{\ell,y}$ depends on the same realization errors that enter the cycle statistics. The following example shows that the resulting inconsistency can persist asymptotically.

\begin{eg}
{\normalfont (Outcome-Based Orientation)}
Suppose that $\beta_0=0$ and, for every selected cycle,
\begin{gather*}
    \alpha_{i_{\ell 1}}
    =
    \psi_{j_{\ell 1}}
    =
    2,
    \qquad
    \alpha_{i_{\ell 2}}
    =
    \psi_{j_{\ell 2}}
    =
    1,
    \qquad
    \eta_{i_{\ell 2}j_{\ell 1}}
    =
    \eta_{i_{\ell 2}j_{\ell 2}}
    =
    0.
\end{gather*}
Assume that the remaining errors are mutually independent and independent across cycles. If
\begin{gather*}
    \eta_{i_{\ell 1}j_{\ell 1}}
    =
    \begin{cases}
        5,  & \text{with probability } 1/2, \\
        -5, & \text{with probability } 1/2,
    \end{cases}
    \qquad
    \eta_{i_{\ell 1}j_{\ell 2}}
    =
    \begin{cases}
        4,   & \text{with probability } 5/6, \\
        -20, & \text{with probability } 1/6,
    \end{cases}
\end{gather*}
then both errors have mean zero, but
\begin{gather*}
    \hat{\beta}_{L,y}
    \xrightarrow{\mathrm{a.s.}}
    \frac{5}{9}.
\end{gather*}
Alternatively, if
\begin{gather*}
    \eta_{i_{\ell 1}j_{\ell 1}}
    =
    \begin{cases}
        1,  & \text{with probability } 3/4, \\
        -3, & \text{with probability } 1/4,
    \end{cases}
    \qquad
    \eta_{i_{\ell 1}j_{\ell 2}}
    =
    \begin{cases}
        2,  & \text{with probability } 1/2, \\
        -2, & \text{with probability } 1/2,
    \end{cases}
\end{gather*}
then
\begin{gather*}
    \hat{\beta}_{L,y}
    \xrightarrow{\mathrm{a.s.}}
    -\frac{3}{5}.
\end{gather*}
\end{eg}

Thus, even under the null of no complementarities and with independent mean-zero errors, outcome-based orientation can produce a nonzero probability limit of either sign. A valid orientation rule must therefore be independent of the realization errors entering the cycle statistics.

\section{Auxiliary Lemmas} \label{appendix:lemmas}

\begin{lemma} \label{lemma:error_epsilon}
    {\normalfont (Regularity of Cycle-Based Error Terms)}
    Under Assumptions~\ref{ass:error_eta} and~\ref{ass:labeling}, the oriented errors $s_\ell\epsilon_{1,\ell}$ and $s_\ell\epsilon_{2,\ell}$ have mean zero and uniformly bounded $(2+\delta)$ moments. Moreover, for $k\in\{1,2\}$,
    \begin{gather*}
        \Var(s_\ell\epsilon_{k,\ell})
        =
        \Var(\epsilon_{k,\ell}).
    \end{gather*}
\end{lemma}

\begin{proof}
    The errors $\epsilon_{1,\ell}$ and $\epsilon_{2,\ell}$ are defined in Equations~\ref{eq:epsilon1} and~\ref{eq:epsilon2}. Compactness of the parameter spaces implies that the deterministic components $\theta_{ij}$ are uniformly bounded. Assumption~\ref{ass:error_eta} and independence of the four outcome errors within each cycle then imply that $\epsilon_{1,\ell}$ and $\epsilon_{2,\ell}$ have mean zero and uniformly bounded $(2+\delta)$ moments.

    By Assumption~\ref{ass:labeling}, $s_\ell$ is either deterministic or independent of the outcome errors. Because $s_\ell\in\{-1,1\}$, for $k\in\{1,2\}$,
    \begin{gather*}
        \E[s_\ell\epsilon_{k,\ell}]
        =
        0,
        \qquad
        |s_\ell\epsilon_{k,\ell}|
        =
        |\epsilon_{k,\ell}|.
    \end{gather*}
    The moment bounds follow immediately, and
    \begin{gather*}
        \Var(s_\ell\epsilon_{k,\ell})
        =
        \E[s_\ell^2\epsilon_{k,\ell}^2]
        =
        \E[\epsilon_{k,\ell}^2]
        =
        \Var(\epsilon_{k,\ell}).
    \end{gather*}
\end{proof}

\begin{lemma} \label{lemma:error_u}
    {\normalfont (Properties of the Composite Error)}
    Under Assumptions~\ref{ass:error_eta} and~\ref{ass:labeling}, the oriented composite error $u_{\ell,s}=s_\ell u_\ell$ has mean zero and uniformly bounded $(2+\delta)$ moments. Moreover, $\Var(u_{\ell,s})=\Var(u_\ell)$, and the average variance $\sigma_{u,L}^2=L^{-1}\sum_{\ell=1}^L\Var(u_\ell)$ is bounded away from zero.
\end{lemma}

\begin{proof}
    Recall that $u_\ell=\epsilon_{1,\ell}+\beta_0\epsilon_{2,\ell}$. Lemma~\ref{lemma:error_epsilon}, together with compactness of the parameter space for $\beta_0$, implies that $u_\ell$ has mean zero and uniformly bounded $(2+\delta)$ moments. By Assumption~\ref{ass:labeling}, $s_\ell$ is either deterministic or independent of the outcome errors. Since $s_\ell^2=1$, $\E[u_{\ell,s}]=0$, $\E[|u_{\ell,s}|^{2+\delta}]=\E[|u_\ell|^{2+\delta}]$, and $\Var(u_{\ell,s})=\Var(u_\ell)$. The lower bound on $\sigma_{u,L}^2$ follows from the nondegeneracy condition in Assumption~\ref{ass:error_eta}.
\end{proof}

\section{Proofs of Main Results} \label{appendix:proofs}

\subsection{Proof of Theorem~\ref{thm:beta_identification}} \label{proof:beta_identification}

Consider a cycle of length $2K$, with $K\geq2$, and denote its outcomes in traversal order by $\theta^{(1)},\dots,\theta^{(2K)}$. Let the odd-indexed edges be $(i_k,j_k)$ and the even-indexed edges be $(i_{k+1},j_k)$, where $i_{K+1}=i_1$. Define $u_i\coloneqq1+\beta_0\alpha_i$ and $v_j\coloneqq1+\beta_0\psi_j$. Under the~\ref{eq:tukey},
\begin{gather*}
    1+\beta_0\theta_{ij}
    =
    (1+\beta_0\alpha_i)(1+\beta_0\psi_j)
    =
    u_iv_j.
\end{gather*}
Multiplying over the odd and even edges yields
\begin{gather*}
    \prod_{k=1}^{K}
    \bigl(1+\beta_0\theta^{(2k-1)}\bigr)
    =
    \prod_{k=1}^{K}u_{i_k}v_{j_k}, \qquad
    \text{and} \qquad
    \prod_{k=1}^{K}
    \bigl(1+\beta_0\theta^{(2k)}\bigr)
    =
    \prod_{k=1}^{K}u_{i_{k+1}}v_{j_k}.
\end{gather*}
The products on the right contain the same worker and firm factors. Hence,
\begin{gather*}
    \prod_{k=1}^{K}
    \bigl(1+\beta_0\theta^{(2k-1)}\bigr)
    =
    \prod_{k=1}^{K}
    \bigl(1+\beta_0\theta^{(2k)}\bigr).
\end{gather*}

A term of degree $r$ in either product is obtained by selecting the outcome term from exactly $r$ of its $K$ factors. Define
\begin{align*}
    S_r^{\mathrm{odd}}
    &\coloneqq
    \sum_{1\leq k_1<\cdots<k_r\leq K}
    \prod_{q=1}^{r}\theta^{(2k_q-1)},
    \qquad
    S_r^{\mathrm{even}}
    \coloneqq
    \sum_{1\leq k_1<\cdots<k_r\leq K}
    \prod_{q=1}^{r}\theta^{(2k_q)},
\end{align*}
with $S_0^{\mathrm{odd}}=S_0^{\mathrm{even}}=1$. By the distributive law,
\begin{gather*}
    \prod_{k=1}^{K}
    \bigl(1+\beta_0\theta^{(2k-1)}\bigr)
    =
    \sum_{r=0}^{K}S_r^{\mathrm{odd}}\beta_0^r, \qquad
    \text{and} \qquad
    \prod_{k=1}^{K}
    \bigl(1+\beta_0\theta^{(2k)}\bigr)
    =
    \sum_{r=0}^{K}S_r^{\mathrm{even}}\beta_0^r.
\end{gather*}
Let $\Delta_r\coloneqq S_r^{\mathrm{odd}}-S_r^{\mathrm{even}}$. Canceling the common constant term gives
\begin{align*}
    0
    =
    \sum_{r=1}^{K}\Delta_r\beta_0^r
    =
    \beta_0
    \left(
        \Delta_1+\Delta_2\beta_0+\cdots+\Delta_K\beta_0^{K-1}
    \right).
\end{align*}
The first factor gives a mechanical root. The interaction parameter must therefore satisfy
\begin{gather*}
    Q(\beta_0)
    \coloneqq
    \Delta_1+\Delta_2\beta_0+\cdots+\Delta_K\beta_0^{K-1}
    =
    0.
\end{gather*}

Any nonzero admissible root $b$ of $Q$ generates productivities consistent with the cycle outcomes. Fix $u_{i_1}(b)=1$ and recover the remaining transformed productivities recursively from
\begin{gather*}
    v_{j_k}(b)
    =
    \frac{1+b\theta^{(2k-1)}}{u_{i_k}(b)},
    \qquad
    u_{i_{k+1}}(b)
    =
    \frac{1+b\theta^{(2k)}}{v_{j_k}(b)}.
\end{gather*}
The cycle restriction ensures that the recursion closes at $i_1$. The corresponding productivities are $\alpha_i(b)=(u_i(b)-1)/b$ and $\psi_j(b)=(v_j(b)-1)/b$. Productivities on trees attached to the cycle can be recovered in the same way.

For a four-cycle, $K=2$ and
\begin{gather*}
    \Delta_1
    =
    \theta_{i_1j_1}+\theta_{i_2j_2}
    -\theta_{i_1j_2}-\theta_{i_2j_1}, \qquad
    \text{and} \qquad
    \Delta_2
    =
    \theta_{i_1j_1}\theta_{i_2j_2}
    -\theta_{i_1j_2}\theta_{i_2j_1}.
\end{gather*}
Under the~\ref{eq:tukey},
\begin{gather*}
    \Delta_1
    =
    \beta_0
    (\alpha_{i_1}-\alpha_{i_2})
    (\psi_{j_1}-\psi_{j_2}), \qquad
    \text{and} \qquad
    \Delta_2
    =
    -
    (\alpha_{i_1}-\alpha_{i_2})
    (\psi_{j_1}-\psi_{j_2}).
\end{gather*}
Assumption~\ref{ass:informative_cycle} implies $\Delta_2\neq0$, so
\begin{gather*}
    \beta_0
    =
    -
    \frac{
        \theta_{i_1j_1}+\theta_{i_2j_2}
        -\theta_{i_1j_2}-\theta_{i_2j_1}
    }{
        \theta_{i_1j_1}\theta_{i_2j_2}
        -\theta_{i_1j_2}\theta_{i_2j_1}
    }.
\end{gather*}
Thus, an informative four-cycle identifies $\beta_0$.

It remains to establish necessity when $G_{IJ}$ contains at most one cycle. If $G_{IJ}$ is a forest, fix a candidate value $b$ and one productivity in each component. The remaining productivities can be recovered recursively along the unique paths of the component. Since there is no closed path, this procedure imposes no restriction on $b$, and $\beta_0$ is not identified.

If the unique cycle is a noninformative four-cycle, then $\Delta_2=0$ and, under the Tukey model, $\Delta_1=-\beta_0\Delta_2=0$. Hence, the cycle imposes no restriction on $\beta_0$. If instead the unique cycle has length $2K$ with $K>2$, the restriction $Q(b)=0$ is a polynomial of degree $K-1$ that can have multiple admissible roots. Because point identification requires uniqueness for every admissible parameter value, a unique longer cycle does not point-identify $\beta_0$ over the parameter space. This argument implicitly assumes that the compact parameter spaces $A$, $\Psi$, and $B$ are sufficiently rich that multiple algebraic roots and their recursively generated productivities are admissible. Under highly restricted parameter spaces, a longer cycle could trivially achieve point identification by ruling out all but one root.

Therefore, when $G_{IJ}$ contains at most one cycle, an informative four-cycle is necessary and sufficient to identify $\beta_0$.

\subsection{Proof of Theorem~\ref{thm:tukey_identification}}

The theorem excludes parameter values for which $1+\beta_0\alpha_i=0$ for some worker $i$ or $1+\beta_0\psi_j=0$ for some firm $j$. At such values, all transformed outcomes incident to that node are zero, preventing the propagation argument below.

Suppose first that $G_{IJ}$ is connected. When $\beta_0=0$, the model reduces to TWFE, and connectedness identifies the productivity collections under the normalization $\alpha_{i_0}=0$. Now let $\beta_0\neq0$ and define
\begin{gather*}
    \alpha_i' \coloneqq 1+\beta_0\alpha_i,
    \qquad
    \psi_j' \coloneqq 1+\beta_0\psi_j,
    \qquad
    \theta_{ij}' \coloneqq 1+\beta_0\theta_{ij}.
\end{gather*}
Under the~\ref{eq:tukey},
\begin{gather*}
    \theta_{ij}'
    =
    \alpha_i'\psi_j'.
\end{gather*}
Because $\alpha_{i_0}=0$, the normalization implies $\alpha_{i_0}'=1$. Hence, for every firm $j$ matched with $i_0$, $\psi_j'=\theta_{i_0j}'$ is identified. For every worker $i$ matched with one of these firms, $\alpha_i'=\theta_{ij}'/\psi_j'$ is then identified. Repeating this argument along paths from $i_0$ identifies all transformed worker and firm productivities because the graph is connected and all transformed productivities are nonzero. The original productivities follow from
\begin{gather*}
    \alpha_i
    =
    \frac{\alpha_i'-1}{\beta_0},
    \qquad
    \psi_j
    =
    \frac{\psi_j'-1}{\beta_0}.
\end{gather*}
Thus, connectedness is sufficient.

Conversely, suppose that $G_{IJ}$ is disconnected. When $\beta_0=0$, adding a constant to all worker productivities in one component and subtracting it from all firm productivities in that component leaves its outcomes unchanged. For $\beta_0\neq0$, consider any component not containing the normalized worker $i_0$. For any nonzero constant $c$, replacing
\begin{gather*}
    \alpha_i'
    \quad\text{with}\quad
    c\alpha_i',
    \qquad
    \psi_j'
    \quad\text{with}\quad
    c^{-1}\psi_j'
\end{gather*}
for every worker and firm in that component leaves all products $\alpha_i'\psi_j'$ unchanged. It therefore leaves all observed outcomes unchanged while generating different productivity collections. Hence, disconnectedness prevents joint identification, and connectedness is necessary.

\subsection{Proof of Theorem~\ref{thm:consistency}}

Using $\Delta_{1,\ell}=\beta_0A_\ell$, $\Delta_{2,\ell}=-A_\ell$, and $u_{\ell,s}=s_\ell(\epsilon_{1,\ell}+\beta_0\epsilon_{2,\ell})$,
\begin{align*}
    \hat{\beta}_{L,s}-\beta_0
    &=
    \frac{
        L^{-1}\sum_{\ell=1}^L u_{\ell,s}
    }{
        \kappa_{L,s}
        -
        L^{-1}\sum_{\ell=1}^L s_\ell\epsilon_{2,\ell}
    }.
\end{align*}

When the orientation sequence is random, condition on its realization. Because the four-cycles are edge-disjoint, the cycle errors are independent across $\ell$. Lemmas~\ref{lemma:error_epsilon} and~\ref{lemma:error_u} imply that $s_\ell\epsilon_{2,\ell}$ and $u_{\ell,s}$ have mean zero and uniformly bounded $(2+\delta)$ moments. The strong law of large numbers for independent triangular arrays therefore gives
\begin{gather*}
    \frac{1}{L}\sum_{\ell=1}^L u_{\ell,s}
    \xrightarrow{\mathrm{a.s.}}
    0,
    \qquad
    \frac{1}{L}\sum_{\ell=1}^L s_\ell\epsilon_{2,\ell}
    \xrightarrow{\mathrm{a.s.}}
    0.
\end{gather*}

By Assumption~\ref{ass:labeling}, $|\kappa_{L,s}|>C_\kappa$ almost surely for all sufficiently large $L$. Hence, the denominator is bounded away from zero almost surely, and
\begin{gather*}
    \hat{\beta}_{L,s}
    \xrightarrow{\mathrm{a.s.}}
    \beta_0.
\end{gather*}

\subsection{Proof of Theorem~\ref{thm:asympt_normality}}

By the definition of $\hat{\beta}_{L,s}$ and the identity
$\Delta_{1,\ell}+\beta_0\Delta_{2,\ell}=0$,
\begin{align*}
    0
    &=
    \frac{1}{L}
    \sum_{\ell=1}^L
    s_\ell
    \left(
        \hat{\Delta}_{1,\ell}
        +
        \hat{\beta}_{L,s}\hat{\Delta}_{2,\ell}
    \right) \\
    &=
    \frac{1}{L}
    \sum_{\ell=1}^L
    u_{\ell,s}
    +
    \left(
        \hat{\beta}_{L,s}-\beta_0
    \right)
    \frac{1}{L}
    \sum_{\ell=1}^L
    s_\ell\hat{\Delta}_{2,\ell}.
\end{align*}
Therefore,
\begin{gather*}
    \sqrt{L}
    \left(
        \hat{\beta}_{L,s}-\beta_0
    \right)
    =
    \frac{
        L^{-1/2}
        \sum_{\ell=1}^L
        u_{\ell,s}
    }{
        \hat{\kappa}_{L,s}
    }.
\end{gather*}

When the orientation sequence is random, condition on its realization. Because the four-cycles are edge-disjoint, the variables $u_{\ell,s}$ are independent across $\ell$. By Lemma~\ref{lemma:error_u}, they have mean zero, uniformly bounded $(2+\delta)$ moments, and average variance $\sigma_{u,L}^2$ bounded away from zero. The Lyapunov central limit theorem for triangular arrays then gives
\begin{gather*}
    \frac{
        L^{-1/2}
        \sum_{\ell=1}^L
        u_{\ell,s}
    }{
        \sigma_{u,L}
    }
    \xrightarrow{d}
    \mathcal{N}(0,1).
\end{gather*}

Moreover, by Lemma~\ref{lemma:error_epsilon},
\begin{gather*}
    \hat{\kappa}_{L,s}
    =
    \kappa_{L,s}
    -
    \frac{1}{L}
    \sum_{\ell=1}^L
    s_\ell\epsilon_{2,\ell}
    =
    \kappa_{L,s}
    +
    o_{\mathrm{a.s.}}(1).
\end{gather*}
Since $|\kappa_{L,s}|$ is bounded away from zero by Assumption~\ref{ass:labeling},
\begin{gather*}
    \frac{|\kappa_{L,s}|}{\hat{\kappa}_{L,s}}
    =
    \operatorname{sign}(\kappa_{L,s})
    +
    o_{\mathrm{a.s.}}(1).
\end{gather*}
Hence,
\begin{align*}
    \frac{
        \sqrt{L}
        \left(
            \hat{\beta}_{L,s}-\beta_0
        \right)
    }{
        \sigma_{u,L}/|\kappa_{L,s}|
    }
    &=
    \frac{
        L^{-1/2}
        \sum_{\ell=1}^L
        u_{\ell,s}
    }{
        \sigma_{u,L}
    }
    \frac{|\kappa_{L,s}|}{\hat{\kappa}_{L,s}}
    \xrightarrow{d}
    \mathcal{N}(0,1),
\end{align*}
because multiplication by the fixed sign
$\operatorname{sign}(\kappa_{L,s})\in\{-1,1\}$ does not change the standard normal distribution.

\subsection{Proof of Proposition~\ref{prop:variance_estimator}}

Define $\hat{u}_{\ell,s}\coloneqq s_\ell(\hat{\Delta}_{1,\ell}+\hat{\beta}_{L,s}\hat{\Delta}_{2,\ell})$. Since $s_\ell^2=1$,
\begin{gather*}
    \hat{\sigma}_{u,L,s}^2
    =
    \frac{1}{L}
    \sum_{\ell=1}^L
    \hat{u}_{\ell,s}^2.
\end{gather*}
Using $\Delta_{1,\ell}+\beta_0\Delta_{2,\ell}=0$,
\begin{gather*}
    \hat{u}_{\ell,s}
    =
    u_{\ell,s}
    +
    (\hat{\beta}_{L,s}-\beta_0)
    s_\ell\hat{\Delta}_{2,\ell}.
\end{gather*}
Therefore,
\begin{align*}
    \hat{\sigma}_{u,L,s}^2-\sigma_{u,L}^2
    &=
    \frac{1}{L}
    \sum_{\ell=1}^L
    \left(
        u_{\ell,s}^2-\E[u_{\ell,s}^2]
    \right) \\
    &\quad
    +
    2(\hat{\beta}_{L,s}-\beta_0)
    \frac{1}{L}
    \sum_{\ell=1}^L
    u_\ell\hat{\Delta}_{2,\ell}
    +
    (\hat{\beta}_{L,s}-\beta_0)^2
    \frac{1}{L}
    \sum_{\ell=1}^L
    \hat{\Delta}_{2,\ell}^2,
\end{align*}
where $u_{\ell,s}s_\ell=u_\ell$.

By Lemma~\ref{lemma:error_u}, the variables
$u_{\ell,s}^2-\E[u_{\ell,s}^2]$ are independent, mean zero, and have uniformly bounded $(1+\delta/2)$ moments. The weak law of large numbers for independent triangular arrays therefore gives
\begin{gather*}
    \frac{1}{L}
    \sum_{\ell=1}^L
    \left(
        u_{\ell,s}^2-\E[u_{\ell,s}^2]
    \right)
    \xrightarrow{p}
    0.
\end{gather*}
The uniform moment bounds also imply
\begin{gather*}
    \frac{1}{L}\sum_{\ell=1}^L u_\ell^2
    =
    O_p(1),
    \qquad
    \frac{1}{L}\sum_{\ell=1}^L \hat{\Delta}_{2,\ell}^2
    =
    O_p(1).
\end{gather*}
By the Cauchy--Schwarz inequality,
\begin{gather*}
    \frac{1}{L}
    \sum_{\ell=1}^L
    u_\ell\hat{\Delta}_{2,\ell}
    =
    O_p(1).
\end{gather*}
Finally, Theorem~\ref{thm:consistency} gives
$\hat{\beta}_{L,s}-\beta_0=o_p(1)$. Hence, the last two terms are $o_p(1)$, and
\begin{gather*}
    \hat{\sigma}_{u,L,s}^2-\sigma_{u,L}^2
    \xrightarrow{p}
    0.
\end{gather*}

\subsection{Proof of Proposition~\ref{prop:instrument_validity}}

Because the ranking variables and the reference ordering are treated as nonrandom, the rank-based orientation sequence $\boldsymbol{s}_{L,z}$ is deterministic and therefore automatically satisfies the exogeneity requirement in Assumption~\ref{ass:labeling}.

By construction, $s_{\ell,z}A_\ell=r_{\ell,z}|A_\ell|$, so
\begin{gather*}
    \kappa_{L,z}
    =
    \frac{1}{L}
    \sum_{\ell=1}^L
    |A_\ell|r_{\ell,z}.
\end{gather*}
Moreover,
\begin{gather*}
    \bar{r}_L
    =
    \frac{1}{L}
    \sum_{\ell=1}^L
    r_{\ell,z}^{\alpha}r_{\ell,z}^{\psi}
    =
    \frac{1}{L}
    \sum_{\ell=1}^L
    \left(
        r_{\ell,z}^{\alpha}-\bar{r}_L^\alpha
    \right)
    \left(
        r_{\ell,z}^{\psi}-\bar{r}_L^\psi
    \right)
    +
    \bar{r}_L^\alpha\bar{r}_L^\psi.
\end{gather*}
Assumptions~\ref{ass:relevance} and~\ref{ass:nonegative} therefore imply
\begin{gather*}
    \bar{r}_L
    \geq
    c_\alpha c_\psi.
\end{gather*}

Let $\mu_L=L^{-1}\sum_{\ell=1}^L|A_\ell|$. Then
\begin{gather*}
    \kappa_{L,z}
    =
    \frac{1}{L}
    \sum_{\ell=1}^L
    \left(
        |A_\ell|-\mu_L
    \right)
    \left(
        r_{\ell,z}-\bar{r}_L
    \right)
    +
    \mu_L\bar{r}_L
    \geq
    \mu_L\bar{r}_L
    \geq
    \mu_Lc_\alpha c_\psi,
\end{gather*}
where the first inequality follows from Assumption~\ref{ass:nolargegaps}. By Assumption~\ref{ass:seq_mu}, for all sufficiently large $L$,
\begin{gather*}
    \kappa_{L,z}
    >
    C_\mu c_\alpha c_\psi
    >
    0.
\end{gather*}
Hence, $\boldsymbol{s}_{L,z}$ satisfies Assumption~\ref{ass:labeling}.

\putbib
\end{bibunit}
\input{supplementary_material.tex}
\end{document}

%% file: supplementary_material.tex
\graphicspath{{figures_and_tables/figures/}}

% Supplementary Material — compiled as part of main.tex
\begingroup
\begin{bibunit}
\setcounter{section}{0}
\renewcommand{\thesection}{S\arabic{section}}
\renewcommand{\thesubsection}{\thesection.\arabic{subsection}}
\setcounter{subsection}{0}
% Unique hyperref anchors: still in \appendix mode, so default names would
% reuse appendix.A, appendix.B, ... and collide with the main-text appendix.
\renewcommand{\theHsection}{supplement.S\arabic{section}}
\renewcommand{\theHsubsection}{supplement.S\arabic{section}.\arabic{subsection}}
\renewcommand{\theHsubsubsection}{supplement.S\arabic{section}.\arabic{subsection}.\arabic{subsubsection}}
\numberwithin{equation}{section}
\numberwithin{theorem}{section}
\numberwithin{prop}{section}
\numberwithin{cor}{section}
\numberwithin{ass}{section}
\numberwithin{df}{section}
\numberwithin{remark}{section}
\numberwithin{eg}{section}

\clearpage
\setcounter{page}{1}
\renewcommand{\thepage}{S\arabic{page}}

\section*{Supplementary Material Contents}

The Supplementary Material contains extensions and supporting analyses that complement the main text. Section~\ref{sec:sm_microfoundation} provides a labor-market interpretation of the Tukey model. Sections~\ref{sec:sm_extensions} and~\ref{sec:sm_ident_supp} study richer interaction functions and the stronger graph conditions required for their identification. Sections~\ref{sec:sm_er}--\ref{sec:sm_mc} present additional results on graph structure, productivity estimation, interpretation, and simulation performance. Section~\ref{sec:sm_proofs} contains proofs of the results stated in the Supplementary Material.

\addcontentsline{toc}{section}{Supplementary Material Contents}
\begin{enumerate}[label=\textbf{S\arabic*.}, leftmargin=2.8em]
    \item A Labor-Market Microfoundation for the Tukey Wage Equation
    \item Richer Interaction Functions
    \item Identification in the HS and Isotonic Models
    \item Four-Cycles in the Erd\H{o}s--R\'enyi Model
    \item Productivity Estimation
    \item Interpreting the Magnitude of $\beta_0$
    \item Monte Carlo Simulations
    \item Supplementary Proofs
\end{enumerate}

\section{A Labor-Market Microfoundation for the Tukey Wage Equation}
\label{sec:sm_microfoundation}

This section presents a simple labor-market environment in which the Tukey model arises as a wage equation. The objective is not to provide a complete structural model of labor-market equilibrium, but to identify conditions on production and wage setting under which the systematic component of the wage paid by firm $j$ to worker $i$ takes the form
\begin{gather*}
    \theta_{ij}
    =
    \alpha_i+\psi_j+\beta_0\alpha_i\psi_j.
\end{gather*}
The derivation adapts the sequential-auction logic of \citet{PostelVinayRobin2002} and the wage-setting discussion in \citet{DiAddarioKlineSaggioSolvsten2023}. The main differences are that utility is linear in wage levels and production is nonseparable in worker and firm capabilities.

Let $a_i\in[0,1]$ denote the productive capability of worker $i$, and let $p_j\in[0,1]$ denote the organizational or technological capability of firm $j$. Output is generated by two types of tasks. A share $\omega\in[0,1]$ is complementary and is completed only when both the worker and the firm are effective, with success probability $a_ip_j$. The remaining share $1-\omega$ is substitutable and is completed when either the worker or the firm can solve the task, with success probability $1-(1-a_i)(1-p_j)=a_i+p_j-a_ip_j$.

Let $s>0$ denote the value of the task component, and allow for direct worker and firm productivity components $r_a a_i$ and $r_p p_j$. Net of any common baseline, match output is
\begin{align}
    R(a_i,p_j)
    &=
    r_aa_i+r_pp_j
    +
    s\left[
        \omega a_ip_j
        +
        (1-\omega)(a_i+p_j-a_ip_j)
    \right] \label{eq:Rap} \\
    &=
    \left[r_a+s(1-\omega)\right]a_i
    +
    \left[r_p+s(1-\omega)\right]p_j
    +
    s(2\omega-1)a_ip_j. \notag
\end{align}
Define
\begin{gather*}
    m_\alpha
    \coloneqq
    r_a+s(1-\omega),
    \qquad
    m_\psi
    \coloneqq
    r_p+s(1-\omega),
    \qquad
    m_{\alpha\psi}
    \coloneqq
    s(2\omega-1),
\end{gather*}
and assume $m_\alpha>0$ and $m_\psi>0$. Let
\begin{gather*}
    \alpha_i
    \coloneqq
    m_\alpha a_i,
    \qquad
    \psi_j
    \coloneqq
    m_\psi p_j.
\end{gather*}
Substituting $a_i=\alpha_i/m_\alpha$ and $p_j=\psi_j/m_\psi$ into Equation~\ref{eq:Rap} gives
\begin{align}
    R(a_i,p_j)
    &=
    \alpha_i+\psi_j
    +
    m_{\alpha\psi}
    \left(
        \frac{\alpha_i}{m_\alpha}
    \right)
    \left(
        \frac{\psi_j}{m_\psi}
    \right) \notag \\
    &=
    \alpha_i+\psi_j+\beta_0\alpha_i\psi_j, \label{eq:sm_task_tukey}
\end{align}
where
\begin{gather*}
    \beta_0
    \coloneqq
    \frac{m_{\alpha\psi}}{m_\alpha m_\psi}
    =
    \frac{
        s(2\omega-1)
    }{
        \left[r_a+s(1-\omega)\right]
        \left[r_p+s(1-\omega)\right]
    }.
\end{gather*}

Equation~\ref{eq:sm_task_tukey} has the Tukey form. Its interaction parameter has a direct interpretation: since $s>0$, $m_\alpha>0$, and $m_\psi>0$,
\begin{gather*}
    \operatorname{sign}(\beta_0)
    =
    \operatorname{sign}(2\omega-1).
\end{gather*}
Thus, $\beta_0>0$ when complementary tasks dominate, $\beta_0<0$ when substitutable tasks dominate, and $\beta_0=0$ when the two forces exactly offset. In the last case, match output is additive.

It remains to connect match output to hiring wages. Workers have linear flow utility, $U(w)=w$. The wage-setting argument follows the sequential-auction logic of \citet{PostelVinayRobin2002} and its implementation in \citet{DiAddarioKlineSaggioSolvsten2023}. In this framework, the hiring wage depends on the worker's outside option, which reflects the value of remaining with the origin firm rather than accepting the destination firm's offer. Rather than repeating the full dynamic derivation, I specialize the argument to the present setting with linear utility and nonseparable match output.

Suppose worker $i$ is attached to origin firm $q$ and receives an offer from destination firm $j$. Let $O(\alpha_i,q)$ denote the worker's outside value, expressed as a flow-wage equivalent. This value includes the continuation value associated with remaining at the origin firm and receiving future offers. Let $\lambda\in[0,1]$ denote the worker's bargaining power with the destination firm. Under this framework and transferable utility, the hiring wage is
\begin{align}
    w_{ijq}(\lambda)
    &=
    O(\alpha_i,q)
    +
    \lambda
    \left[
        R(a_i,p_j)-O(\alpha_i,q)
    \right] \notag \\
    &=
    \lambda R(a_i,p_j)
    +
    (1-\lambda)O(\alpha_i,q). \label{eq:sm_bargaining}
\end{align}
A general poaching model therefore does not by itself yield the Tukey wage equation. When $\lambda<1$, the wage depends on the origin firm through $O(\alpha_i,q)$, as in standard sequential-auction models.

In the limiting case $\lambda=1$, the worker has full bargaining power and Equation~\ref{eq:sm_bargaining} gives $w_{ijq}(1)=R(a_i,p_j)$. The origin firm then drops out, and Equation~\ref{eq:sm_task_tukey} implies
\begin{gather*}
    w_{ij}
    =
    \alpha_i+\psi_j+\beta_0\alpha_i\psi_j.
\end{gather*}
Thus, under full worker bargaining power, the hiring wage equals the match output generated at the destination firm.

Finally, let the observed hiring wage be $y_{ij}=w_{ij}+\eta_{ij}$, where $\eta_{ij}$ is a match-specific disturbance satisfying $\E[\eta_{ij}]=0$. The systematic component of the observed wage is then
\begin{gather*}
    \theta_{ij}
    =
    \E[y_{ij}]
    =
    \alpha_i+\psi_j+\beta_0\alpha_i\psi_j,
\end{gather*}
which is the Tukey model studied in the main text, specialized to a labor-market setting.

\section{Richer Interaction Functions}\label{sec:sm_extensions}

The Tukey model introduces complementarities in a parsimonious and interpretable way. Richer interaction patterns can be accommodated by relaxing its restrictions. The heterogeneous-slope (HS) model allows the Tukey interaction parameter to vary across firms, capturing heterogeneity in complementarities. The isotonic model provides a fully nonparametric alternative in which the interaction function is required only to be strictly increasing in each argument, while its cross-partial derivative remains unrestricted.

\subsection{Heterogeneous-Slope Model} \label{sec:hs_model}

Rather than modeling firm productivity as a scalar, suppose instead that each firm is characterized by a productivity vector $\psi_j = (b_j, a_j)$. This leads to the specification
\begin{gather} \label{eq:hs}
    \theta_{ij} = a_j + b_j \alpha_i, \tag{HS model}
\end{gather}
where $a_j$ is a firm-specific intercept and $b_j$ is a firm-specific slope.

\citet{bonhomme2019distributional} use the same functional form within a grouped specification, in which workers and firms belonging to the same group share latent characteristics. In the BI framework, the HS model corresponds to the limiting case in which each worker and firm forms a separate group.

The HS model nests the Tukey model. If there exists a common parameter $\beta_0$ such that $b_j=1+\beta_0a_j$ for every firm, then
\begin{gather*}
    \theta_{ij}
    =
    a_j+(1+\beta_0a_j)\alpha_i
    =
    \alpha_i+a_j+\beta_0\alpha_i a_j,
\end{gather*}
which is the Tukey specification with firm productivity $\psi_j=a_j$. The HS model relaxes the restriction linking a firm's additive productivity to the marginal contribution of worker productivity. This additional flexibility allows researchers to distinguish a firm's intrinsic productivity $a_j$, which enters the model in levels, from its capacity to extract value from workers and exploit complementarities, governed by $b_j$. Unlike the Tukey model, which implicitly ties these two roles together, the HS model permits them to differ, thereby enabling empirical investigation of whether more productive firms are also those in which workers' marginal productivity is higher.

\subsection{Isotonic Model} \label{sec:isotonic_model}

The HS model imposes a parametric form on the interaction function. As a more flexible benchmark, consider a specification that leaves the pattern of complementarities unrestricted. As in nonparametric regression, some regularity is needed for the observed data to be informative about an otherwise unrestricted function. In the bipartite interaction setting, I impose monotonicity.

Productivities remain scalar, and the interaction function satisfies
\begin{gather} \label{eq:isotonic}
    \theta_{ij}
    =
    f_m(\alpha_i,\psi_j),
    \tag{Isotonic model}
\end{gather}
where $f_m\colon A\times\Psi\to\R$ is strictly increasing in each argument. This restriction connects the model to seriation problems \citep{flammarion2019optimal}, matrix completion for bivariate isotonic matrices under unknown permutations \citep{mao2020towards}, and nonparametric latent-space models of network formation \citep{gao2020nonparametric}.

Monotonicity is a common shape restriction in economics; see \citet{chetverikov2018econometrics} for a survey. Here, it requires the ordering of any two workers to be preserved across the firms with which they are both observed, and analogously for firms. In particular, if $\alpha_i>\alpha_{i'}$, then
\begin{gather*}
    f_m(\alpha_i,\psi_j)
    >
    f_m(\alpha_{i'},\psi_j)
\end{gather*}
for every firm $j$, with the analogous property holding across firms for a fixed worker. The TWFE model satisfies this restriction.

Beyond monotonicity, the isotonic model places no parametric restriction on $f_m$. When $f_m$ is twice differentiable, its first partial derivatives are nonnegative, while the cross-partial derivative governing complementarities remains unrestricted. Complementarity patterns may therefore vary throughout the productivity space: the interaction function can be supermodular for some worker-firm combinations and submodular for others. By contrast, in the HS model, the marginal productivity of a worker may vary across firms but is constant across workers within each firm.

\subsection{Interaction Flexibility and Network Requirements} \label{sec:tradeoff}

The TWFE, Tukey, HS, and isotonic models trace out a spectrum of interaction functions within the BI framework. At one extreme, the modular TWFE model rules out complementarities. At the other, the isotonic model imposes only strict monotonicity and allows complementarity patterns to vary across the productivity space. The Tukey and HS models lie between these extremes, introducing complementarities through increasingly flexible parametric structures.

The matching network $G_{IJ}$ provides a second dimension of variation by recording the observed worker-firm pairs. At one extreme, the network is complete, with $D_{ij}=1$ for every pair. At the other, it may be sparse or disconnected, with each worker and firm observed in only a small number of matches.

The interaction function and the matching network jointly determine what can be learned from the data. Richer interaction functions introduce more unknown structure, while the network determines which comparisons are available to recover it. This creates a trade-off: more flexible specifications require more informative network structures, whereas sparse networks can support identification only under stronger restrictions on the interaction function.

The two components differ in an important respect. The interaction function is unobserved and must be specified by the researcher, whereas the matching network is observed and its relevant properties can be checked directly. For the TWFE model, \citet{abowd1999high} and \citet{jochmans2019fixed} provide graph conditions for identification and, under additional assumptions, estimation of the productivity parameters. The results developed here extend the analysis to the Tukey, HS, and isotonic models with incomplete and potentially sparse matching networks, characterizing how the required graph structure changes with the flexibility of the interaction function.

In applications, researchers can first select an interaction specification appropriate to the economic setting and then verify whether the observed matching network satisfies the corresponding graph conditions. The parameters may fail to be point-identified, may be point-identified without supporting consistent estimation under the relevant asymptotic sequence, or may be both identified and consistently estimable. Distinguishing among these cases clarifies which questions the available data can credibly answer. Naturally, the richer and more connected the graph, the wider the scope of questions that can be addressed.

\section{Identification in the HS and Isotonic Models}\label{sec:sm_ident_supp}

\subsection{Identification in the Heterogeneous-Slope Model} \label{sec:hs_identification}

The heterogeneous-slope (HS) model allows the Tukey interaction parameter to vary across firms. Under a complete matching network, \citet{bonhomme2019distributional} establish point identification of the worker productivities $\alpha_i$ and the firm parameters $a_j$ and $b_j$. Completeness, however, is stronger than necessary. I next provide a weaker condition on the matching network that is necessary and sufficient for identification with incomplete matching.

To state the graph condition, let
\begin{gather*}
    N(i)
    \coloneqq
    \left\{
        j\in[J]:
        (i,j)\in\mathcal{O}_{IJ}
    \right\}
\end{gather*}
denote the set of firms with which worker $i$ is observed, and let
\begin{gather*}
    \mathcal{I}_M
    \coloneqq
    \left\{
        i\in[I]:
        |N(i)|\geq2
    \right\}
\end{gather*}
denote the set of mobile workers. The firm-mobility hypergraph has vertex set $[J]$ and one hyperedge $N(i)$ for each $i\in\mathcal{I}_M$. The hyperedges are labeled by workers, so two workers observed at the same set of firms remain distinct hyperedges.

For any partition $\mathcal{P}$ of the firm set, let $q_i(\mathcal{P})$ denote the number of elements of $\mathcal{P}$ that intersect $N(i)$. A worker whose employment history intersects $q_i(\mathcal{P})$ groups of firms provides $q_i(\mathcal{P})-1$ comparisons across these groups.

\begin{ass} \label{ass:hs_connectivity}
    {\normalfont (Location-and-Scale Connectivity)}
    For every partition $\mathcal{P}$ of $[J]$,
    \begin{gather*}
        \sum_{i\in\mathcal{I}_M}
        \left[
            q_i(\mathcal{P})-1
        \right]
        \geq
        2
        \left(
            |\mathcal{P}|-1
        \right).
    \end{gather*}
\end{ass}

The graph in Figure~\ref{fig:bipartite_graph_snow_panel_a} violates Assumption~\ref{ass:hs_connectivity}. Consider the partition that places $j_1$ and $j_2$ in one group and $j_3$ in the other. Only worker $i_4$ is observed in both groups, so the left-hand side of the inequality is one, while the right-hand side is two.

Adding the edge between $i_5$ and $j_1$, as in Figure~\ref{fig:bipartite_graph_snow_panel_b}, makes the condition hold. Every partition into two groups is crossed by at least two mobile workers. For the partition into the three singleton firms, workers $i_1$, $i_2$, $i_4$, and $i_5$ each provide one comparison, so both sides of the inequality are equal to four.

\begin{figure}[ht]
    \centering
    \captionsetup[subfigure]{justification=centering}

    \begin{subfigure}[t]{0.48\textwidth}
        \centering
        \begin{tikzpicture}[
            scale=0.8,
            >=stealth,
            font=\small,
            pointI/.style={circle,fill=purple,draw=black,inner sep=3pt},
            pointJ/.style={circle,fill=green,draw=black,inner sep=3pt}
        ]
            \foreach \y/\lab in {1.5/1,0.5/2,-0.5/3,-1.5/4,-2.5/5} {
                \node[pointI] (I\lab) at (0,\y) {};
            }
            \foreach \y/\lab in {0.5/1,-0.5/2,-2.5/3} {
                \node[pointJ] (J\lab) at (4,\y) {};
            }

            \node[left=4pt of I1] {$i_1$};
            \node[left=4pt of I2] {$i_2$};
            \node[left=4pt of I3] {$i_3$};
            \node[left=4pt of I4] {$i_4$};
            \node[left=4pt of I5] {$i_5$};

            \node[right=4pt of J1] {$j_1$};
            \node[right=4pt of J2] {$j_2$};
            \node[right=4pt of J3] {$j_3$};

            \draw (I1) -- (J1);
            \draw (I1) -- (J2);
            \draw (I2) -- (J1);
            \draw (I2) -- (J2);
            \draw (I3) -- (J2);
            \draw (I4) -- (J2);
            \draw (I5) -- (J3);
            \draw (I4) -- (J3);
        \end{tikzpicture}
        \caption{}
        \label{fig:bipartite_graph_snow_panel_a}
    \end{subfigure}
    \hfill
    \begin{subfigure}[t]{0.48\textwidth}
        \centering
        \begin{tikzpicture}[
            scale=0.8,
            >=stealth,
            font=\small,
            pointI/.style={circle,fill=purple,draw=black,inner sep=3pt},
            pointJ/.style={circle,fill=green,draw=black,inner sep=3pt}
        ]
            \foreach \y/\lab in {1.5/1,0.5/2,-0.5/3,-1.5/4,-2.5/5} {
                \node[pointI] (I\lab) at (0,\y) {};
            }
            \foreach \y/\lab in {0.5/1,-0.5/2,-2.5/3} {
                \node[pointJ] (J\lab) at (4,\y) {};
            }

            \node[left=4pt of I1] {$i_1$};
            \node[left=4pt of I2] {$i_2$};
            \node[left=4pt of I3] {$i_3$};
            \node[left=4pt of I4] {$i_4$};
            \node[left=4pt of I5] {$i_5$};

            \node[right=4pt of J1] {$j_1$};
            \node[right=4pt of J2] {$j_2$};
            \node[right=4pt of J3] {$j_3$};

            \draw (I1) -- (J1);
            \draw (I1) -- (J2);
            \draw (I2) -- (J1);
            \draw (I2) -- (J2);
            \draw (I3) -- (J2);
            \draw (I4) -- (J2);
            \draw (I5) -- (J3);
            \draw (I4) -- (J3);
            \draw[red,very thick] (I5) -- (J1);
        \end{tikzpicture}
        \caption{}
        \label{fig:bipartite_graph_snow_panel_b}
    \end{subfigure}

    \caption{\small Two matching networks. Panel (a) does not satisfy Assumption~\ref{ass:hs_connectivity}. Panel (b) satisfies the assumption after the red edge is added.}
    \label{fig:bipartite_graph_snow}
\end{figure}

As in the Tukey model, the graph condition must be combined with productivity heterogeneity. Assumption~\ref{ass:informative_cycle} requires heterogeneity on both sides of a four-cycle; otherwise, the cycle contrasts contain no information about $\beta_0$. In the HS model, workers connecting different firm schedules must likewise have sufficiently heterogeneous productivities. If the relevant workers have identical or otherwise exceptional productivity values, mobility comparisons that are distinct in the graph may impose the same restriction on the firm schedules.

\begin{ass} \label{ass:hs_heterogeneity}
    {\normalfont (HS Heterogeneity)}
    All firm slopes are nonzero. The productivities of the mobile workers are in general position relative to the firm-mobility hypergraph.
\end{ass}

General position means that the mobile-worker productivities generate the maximum rank permitted by the firm-mobility hypergraph. Assumption~\ref{ass:hs_heterogeneity} excludes exceptional productivity values satisfying polynomial equalities that make otherwise distinct mobility comparisons linearly dependent. For example, when two firms share two workers, the assumption requires the two workers to have different productivities. More generally, the excluded values form a proper algebraic subset of the productivity space. If $A$ contains an interval, this subset has Lebesgue measure zero in $A^{|\mathcal{I}_M|}$, so the assumption is satisfied for almost every collection of mobile-worker productivities.

Intuitively, because the HS model generalizes the Tukey model by allowing the interaction parameter to vary across firms, Assumptions~\ref{ass:hs_connectivity} and~\ref{ass:hs_heterogeneity} strengthen Assumptions~\ref{ass:connectedness} and~\ref{ass:informative_cycle}. In particular, they imply that every firm belongs to a cycle involving heterogeneous worker productivities and that these cycles are connected through shared nodes. These features allow identifying information to propagate across the matching network. The next theorem shows that the full conditions are necessary and sufficient for identification of the HS model parameters.

\begin{theorem} \label{thm:hs_identification}
    {\normalfont (Identification in the HS Model)}
    Under the~\ref{eq:hs}, the normalization $a_{j_0}=b_{j_0}=1$, and Assumption~\ref{ass:hs_heterogeneity}, Assumption~\ref{ass:hs_connectivity} is necessary and sufficient for identification of $\boldsymbol{\alpha}$, $\boldsymbol{a}$, and $\boldsymbol{b}$.
\end{theorem}

Assumption~\ref{ass:hs_connectivity} requires a matching network substantially richer than those typically observed in applications. In particular, it implies that every firm belongs to a cycle. In the labor-market data studied by \citet{kline2024firm}, nearly half of firms do not belong to any cycle, so their firm parameters cannot be identified under the HS model. This share provides only a lower bound on the extent of nonidentification: cycle membership is not sufficient, because mobility comparisons must also be sufficiently numerous and distributed across the firm market. The HS model may therefore leave the parameters of a large share of firms unidentified even under otherwise favorable conditions.

Theorem~\ref{thm:hs_identification} thus shows that the additional flexibility provided by firm-specific complementarities comes with demanding graph requirements that are unlikely to hold in many empirical settings with two-sided interactions. Dimension reduction restrictions, such as the grouped structure in \citet{bonhomme2019distributional}, can make the model applicable under weaker matching patterns.

\subsection{Identification in the Isotonic Model} \label{sec:isotonic_identification}

One might ask whether the restrictive graph condition found for identification with the HS model is driven not by its richer interaction function, but by the use of multidimensional firm productivity, which introduces additional parameters.

The isotonic model provides an alternative: it allows fully flexible complementarities while retaining scalar productivities. Here, the interaction function is left entirely nonparametric. As the next proposition shows, this flexibility comes with a sharp limitation: $f_m$, $\boldsymbol{\alpha}$, and $\boldsymbol{\psi}$ can be recovered only up to strictly monotonic reparameterizations.

\begin{prop} \label{prop:isotonic_monotone_equivalence}
{\normalfont (Lack of Cardinal Identification in the Isotonic Model)}
    Model $(f_m, \boldsymbol{\alpha}, \boldsymbol{\psi}, G_{IJ})$ can be identified only up to any strictly increasing reparameterization of $\boldsymbol{\alpha}$ and $\boldsymbol{\psi}$.
\end{prop}

Proposition~\ref{prop:isotonic_monotone_equivalence} implies that only the ordinal information (the ranking) of worker and firm productivities can be identified. Cardinal differences are not preserved under strictly monotonic transformations, so $f_m$, $\boldsymbol{\alpha}$, and $\boldsymbol{\psi}$ cannot be separately identified. 

While this prevents recovery of productivity levels, the ranks of $\{ \alpha_i \}$ and $\{ \psi_j \}$, and thus the rankings of the vectors $\boldsymbol{\alpha}$ and $\boldsymbol{\psi}$, are still identifiable. Rank-based methods can therefore be employed to study sorting patterns and the determinants of productivity ranks without imposing cardinal structure.

To state an identification condition for the isotonic model, I first introduce the notion of within-side diameters.

\begin{df}
{\normalfont (Within-side Diameter)}
    The within-side diameters of a graph $G_{IJ}$ are the largest shortest-path distances between any two nodes on the same side of the bipartition: letting $d(u,v)$ be the number of edges in the shortest path between nodes $u$ and $v$,
    \begin{gather*}
        \operatorname{diam}_{I}(G_{IJ}) = \max_{i,i'\in[I]} d(i,i'), \qquad \operatorname{diam}_{J}(G_{IJ}) = \max_{j,j'\in[J]} d(j,j').
    \end{gather*}
\end{df}

In the graph of Figure~\ref{fig:bipartite_graph_snow_panel_a}, the within-side diameter for $I$ is $4$, as the shortest path from $i_1$ to $i_5$ spans four edges. For $J$, the diameter is $4$, corresponding to the path between $j_1$ and $j_3$. In Figure~\ref{fig:bipartite_graph_snow_panel_b}, the diameter for $J$ falls to $2$ thanks to the shorter path $j_1 \to i_5 \to j_3 $, while the diameter for $I$ remains $4$.

The following condition for identification in the isotonic model directly involves the within-side diameters.

\begin{ass} \label{ass:diameter}
{\normalfont (Diameter 2)}
    The matching network $G_{IJ}$ has within-side diameters equal to two.
\end{ass}

Assumption~\ref{ass:diameter} requires that every pair of workers shares at least one common firm, and every pair of firms shares at least one common worker. This ensures that all nodes on the same side of the bipartition are linked through a single intermediary on the opposite side. The graph in Figure~\ref{fig:bipartite_graph_diameter_panel_a} fails this property, as workers $i_3$ and $i_5$ have no firm in common. Adding the edge $(i_3,j_3)$, as in Figure~\ref{fig:bipartite_graph_diameter_panel_b}, creates the necessary connections and yields within-side diameters of two.

\begin{figure}[ht]
  \centering
  \captionsetup[subfigure]{justification=centering}

  % ---------- (a) ----------
  \begin{subfigure}[t]{0.48\textwidth}
    \centering
    \begin{tikzpicture}[scale=0.8,
      >=stealth,
      font=\small,
      pointI/.style={circle,fill=purple,draw=black,inner sep=3pt},
      pointJ/.style={circle,fill=green,draw=black,inner sep=3pt}
    ]
      % 1. points
      \foreach \y/\lab in {1.5/1,0.5/2,-0.5/3,-1.5/4,-2.5/5}{
          \node[pointI] (I\lab) at (0,\y) {};
      }
      \foreach \y/\lab in {0.5/1,-0.5/2,-2.5/3}{
          \node[pointJ] (J\lab) at (4,\y) {};
      }

      % 2. labels
      \node[left=4pt of I1] {$i_1$};
      \node[left=4pt of I2] {$i_2$};
      \node[left=4pt of I3] {$i_3$};
      \node[left=4pt of I4] {$i_4$};
      \node[left=4pt of I5] {$i_5$};

      \node[right=4pt of J1] {$j_1$};
      \node[right=4pt of J2] {$j_2$};
      \node[right=4pt of J3] {$j_3$};

      % 3. edges
      \draw (I1) -- (J1);
      \draw (I1) -- (J2);
      \draw (I2) -- (J1);
      \draw (I2) -- (J2);
      \draw (I3) -- (J2);
      \draw (I4) -- (J2);
      \draw (I5) -- (J3);
      \draw (I4) -- (J3);
      \draw (I5) -- (J1);
    \end{tikzpicture}
    \caption{}
    \label{fig:bipartite_graph_diameter_panel_a}
  \end{subfigure}
  \hfill
  % ---------- (b) ----------
  \begin{subfigure}[t]{0.48\textwidth}
    \centering
    \begin{tikzpicture}[scale=0.8,
      >=stealth,
      font=\small,
      pointI/.style={circle,fill=purple,draw=black,inner sep=3pt},
      pointJ/.style={circle,fill=green,draw=black,inner sep=3pt}
    ]
      % 1. points
      \foreach \y/\lab in {1.5/1,0.5/2,-0.5/3,-1.5/4,-2.5/5}{
          \node[pointI] (I\lab) at (0,\y) {};
      }
      \foreach \y/\lab in {0.5/1,-0.5/2,-2.5/3}{
          \node[pointJ] (J\lab) at (4,\y) {};
      }

      % 2. labels
      \node[left=4pt of I1] {$i_1$};
      \node[left=4pt of I2] {$i_2$};
      \node[left=4pt of I3] {$i_3$};
      \node[left=4pt of I4] {$i_4$};
      \node[left=4pt of I5] {$i_5$};

      \node[right=4pt of J1] {$j_1$};
      \node[right=4pt of J2] {$j_2$};
      \node[right=4pt of J3] {$j_3$};

      % 3. edges
      \draw (I1) -- (J1);
      \draw (I1) -- (J2);
      \draw (I2) -- (J1);
      \draw (I2) -- (J2);
      \draw (I3) -- (J2);
      \draw (I4) -- (J2);
      \draw (I5) -- (J3);
      \draw (I4) -- (J3);
      \draw (I5) -- (J1);
      \draw[red, very thick] (I3) -- (J3);
    \end{tikzpicture}
    \caption{}
    \label{fig:bipartite_graph_diameter_panel_b}
  \end{subfigure}

  \caption{\small Two examples of matching networks: (a) a network that does not satisfy Assumption~\ref{ass:diameter}, and (b) a network that does, with within-side diameters equal to 2.}
  \label{fig:bipartite_diameter_snow}
\end{figure}

The next result shows that a within-side diameter of two is exactly the connectivity needed to recover the rankings of $\boldsymbol{\alpha}$ and $\boldsymbol{\psi}$.

\begin{theorem} \label{thm:isotonic_identification}
{\normalfont (Identification of Rankings in the Isotonic Model)}
    Under the \ref{eq:isotonic}, Assumption~\ref{ass:diameter} is necessary and sufficient for identification of the rankings of $\boldsymbol{\alpha}$ and $\boldsymbol{\psi}$.
\end{theorem}

From an empirical standpoint, Assumption~\ref{ass:diameter} is strong: it requires that every pair of workers has at least one firm in common and that every pair of firms shares at least one worker. In labor-market data, this would mean that any two workers have worked for the same firm, an unlikely occurrence outside of small or highly interconnected markets. As with the HS model, the isotonic model therefore has limited empirical applicability when the focus is on point identification.

That said, point identification is not always essential for extracting useful information from the data. Even when Assumption~\ref{ass:diameter} fails, the isotonic model can still deliver partial identification of the productivity rankings. Although I do not study partial identification in this paper, it remains a promising direction to explore. Following the approach of \cite{crippa2025pairwise}, for example, one could derive informative sets for the rankings of $\boldsymbol{\alpha}$ and $\boldsymbol{\psi}$.

\section{Four-Cycles in the Erd\H{o}s--R\'enyi Model}
\label{sec:sm_er}

In this section, I study conditions under which the bipartite Erd\H{o}s--R\'enyi random graph contains, with probability approaching one, a diverging number of four-cycles. In the BI framework, the matching network $G_{IJ}$ is treated as fixed, so its number of cycles is deterministic. The objective here is not to model network formation, but to clarify the graph density associated with the presence of many four-cycles and compare it with the density required for connectivity. Similar insights may arise under other random graph models, but the Erd\H{o}s--R\'enyi model provides a simple benchmark in which the network structure is governed by a single parameter.

In the bipartite Erd\H{o}s--R\'enyi graph $G_{IJ}(p_{IJ})$, each pair $(i,j)\in[I]\times[J]$ is independently linked with probability $p_{IJ}\in(0,1)$. The link probability is constant across potential edges but may vary with $I$ and $J$. Let $C_4(I,J)$ denote the random number of four-cycles in $G_{IJ}(p_{IJ})$. I consider an asymptotic sequence satisfying
\begin{gather*}
    0
    <
    \underline{c}
    \leq
    \frac{I}{J}
    \leq
    \overline{c}
    <
    \infty
\end{gather*}
as $I,J\to\infty$, so that the numbers of agents on the two sides grow at the same rate. The following proposition provides a sufficient condition for the total number of four-cycles to diverge.

\begin{prop} \label{prop:ER}
    {\normalfont (Four-Cycles in the Erd\H{o}s--R\'enyi Model)}
    If
    \begin{gather*}
        \sqrt{IJ}p_{IJ}
        \to
        \infty,
    \end{gather*}
    then
    \begin{gather*}
        C_4(I,J)
        \xrightarrow{p}
        \infty.
    \end{gather*}
\end{prop}

Under balanced growth, Proposition~\ref{prop:ER} requires a link probability larger than order $1/\sqrt{IJ}$. By contrast, connectivity requires a link probability of order $\log(\sqrt{IJ})/\sqrt{IJ}$; see Section~8.2 of \citet{blum2020foundations}. Because $\log(\sqrt{IJ})\to\infty$, the density required for connectivity implies the condition in Proposition~\ref{prop:ER}. The converse does not hold: the total number of four-cycles can diverge even when $\sqrt{IJ}p_{IJ}$ grows more slowly than $\log(\sqrt{IJ})$. Thus, connectivity requires a strictly denser graph than the existence of a diverging total number of four-cycles.

Proposition~\ref{prop:ER} therefore shows that the graph density associated with connectivity generates abundant four-cycle variation. By itself, however, this result establishes only that the total number of four-cycles diverges, whereas Assumption~\ref{ass:cycles} is stated in terms of a diverging collection of edge-disjoint four-cycles, which provides the cleanest formulation of the asymptotic argument.

\section{Productivity Estimation}
\label{sec:sm_prod}

The main text focuses on estimation and inference for the interaction parameter $\beta_0$. Researchers may also be interested in the individual productivity parameters in $\boldsymbol{\alpha}$ and $\boldsymbol{\psi}$, as well as functions of these collections, such as sorting measures and variance components. Two cases arise. When $\beta_0=0$, the model reduces to the standard TWFE specification, and existing estimators can be applied. I briefly review the available results and the additional graph conditions required for inference. When $\beta_0\neq0$, I describe a computational procedure for obtaining productivity estimates after estimating $\beta_0$. I do not establish consistency or derive a limiting distribution for these estimates. Doing so would require additional conditions on the sequence of matching networks and an analysis of how first-stage estimation error in $\hat{\beta}_{L,z}$ propagates to the productivity estimates.

\subsection{Productivity Estimation under TWFE}

When $\beta_0=0$, the Tukey model reduces to
\begin{gather*}
    y_{ij}
    =
    \alpha_i+\psi_j+\eta_{ij}.
\end{gather*}
Under Assumption~\ref{ass:connectedness} and the normalization $\alpha_{i_0}=0$, the productivity collections can be estimated using the TWFE estimator
\begin{gather*}
    \left(
        \hat{\boldsymbol{\alpha}}^{twfe},
        \hat{\boldsymbol{\psi}}^{twfe}
    \right)
    =
    (C'C)^{-1}C'\boldsymbol{y}_{\mathcal{O}},
\end{gather*}
where the design matrix $C$ is defined as in Definition~\ref{def:twfe_projection}.

\citet{jochmans2019fixed} study the asymptotic properties of the TWFE estimator. For inference on an individual productivity parameter, they derive an asymptotic distribution under conditions that require the degree of the corresponding node to diverge. More generally, inference on functionals of $\boldsymbol{\alpha}$ and $\boldsymbol{\psi}$ requires information to accumulate across many nodes. These conditions can be difficult to satisfy in labor-market applications, where workers are often observed at only a small number of firms and many firms employ only a limited number of observed workers.

\citet{kline2020leave} study inference on quadratic functions of the productivity collections, including variances and covariances. Their results require additional restrictions on the sequence of matching networks that prevent the graph from becoming sparsely connected or fragmenting into subgraphs joined by only a small number of edges. As Table~IV of their paper illustrates, these restrictions can be demanding in labor-market data.

These results highlight the difficulty of estimating individual productivities and functions of them in sparse matching networks. Information about an individual productivity is concentrated in the relatively few outcomes involving that agent. This differs from the interaction parameter $\beta_0$, which is common to all observed matches and can therefore be estimated by aggregating identifying variation across the network.

\subsection{Productivity Estimation under the Tukey Model}

Suppose that Assumption~\ref{ass:connectedness} holds, an estimator $\hat{\beta}_{L,z}$ of $\beta_0$ is available, and $\hat{\beta}_{L,z}\neq0$. Define the transformed outcomes
\begin{gather*}
    \tilde{y}_{ij}
    \coloneqq
    1+\hat{\beta}_{L,z}y_{ij}.
\end{gather*}
For any candidate productivity collections, define the corresponding transformed factors as
\begin{gather*}
    u_i
    \coloneqq
    1+\hat{\beta}_{L,z}\alpha_i,
    \qquad
    v_j
    \coloneqq
    1+\hat{\beta}_{L,z}\psi_j.
\end{gather*}

At the true parameter value,
\begin{align*}
    1+\beta_0y_{ij}
    &=
    1+\beta_0
    \left(
        \alpha_i+\psi_j+\beta_0\alpha_i\psi_j+\eta_{ij}
    \right) \\
    &=
    (1+\beta_0\alpha_i)(1+\beta_0\psi_j)
    +
    \beta_0\eta_{ij}.
\end{align*}
Thus, the transformed outcome matrix has a rank-one deterministic component. This representation motivates estimating the transformed productivity factors by
\begin{gather} \label{eq:transformed_prod_estimator}
    \left(
        \hat{\boldsymbol{u}},
        \hat{\boldsymbol{v}}
    \right)
    \in
    \arg\min_{\boldsymbol{u},\boldsymbol{v}}
    \sum_{(i,j)\in\mathcal{O}_{IJ}}
    \left(
        \tilde{y}_{ij}-u_iv_j
    \right)^2
    \qquad
    \text{s.t.}
    \qquad
    u_{i_0}=1.
\end{gather}
The normalization $u_{i_0}=1$ corresponds to the normalization $\alpha_{i_0}=0$ used in the identification analysis.

Given a solution to Equation~\ref{eq:transformed_prod_estimator}, the productivity estimates are recovered as
\begin{gather} \label{eq:recover_productivities}
    \hat{\alpha}_i
    =
    \frac{\hat{u}_i-1}{\hat{\beta}_{L,z}},
    \qquad
    \hat{\psi}_j
    =
    \frac{\hat{v}_j-1}{\hat{\beta}_{L,z}}.
\end{gather}

The criterion in Equation~\ref{eq:transformed_prod_estimator} is equivalent to an unrestricted plug-in nonlinear least-squares criterion in the original parameterization. Indeed, for
\begin{gather*}
    u_i
    =
    1+\hat{\beta}_{L,z}\alpha_i,
    \qquad
    v_j
    =
    1+\hat{\beta}_{L,z}\psi_j,
\end{gather*}
it follows that
\begin{align*}
    \tilde{y}_{ij}-u_iv_j
    &=
    1+\hat{\beta}_{L,z}y_{ij}
    -
    (1+\hat{\beta}_{L,z}\alpha_i)
    (1+\hat{\beta}_{L,z}\psi_j) \\
    &=
    \hat{\beta}_{L,z}
    \left(
        y_{ij}
        -
        \alpha_i
        -
        \psi_j
        -
        \hat{\beta}_{L,z}\alpha_i\psi_j
    \right).
\end{align*}
Consequently,
\begin{align*}
    \sum_{(i,j)\in\mathcal{O}_{IJ}}
    \left(
        \tilde{y}_{ij}-u_iv_j
    \right)^2
    &=
    \hat{\beta}_{L,z}^{\,2}
    \sum_{(i,j)\in\mathcal{O}_{IJ}}
    \left(
        y_{ij}
        -
        \alpha_i
        -
        \psi_j
        -
        \hat{\beta}_{L,z}\alpha_i\psi_j
    \right)^2.
\end{align*}
Because $\hat{\beta}_{L,z}^{\,2}$ does not depend on the productivity parameters and the transformation is one-to-one when $\hat{\beta}_{L,z}\neq0$, the two criteria have the same minimizers.

Although the criterion in Equation~\ref{eq:transformed_prod_estimator} is not jointly convex in $(\boldsymbol{u},\boldsymbol{v})$, it is quadratic in either vector when the other is held fixed. This suggests an alternating least-squares procedure. Starting from an initial value $\boldsymbol{v}^{(0)}$, update each worker factor according to
\begin{gather} \label{eq:als_a}
    u_i^{(t+1)}
    \coloneqq
    \frac{
        \displaystyle
        \sum_{j:D_{ij}=1}
        v_j^{(t)}\tilde{y}_{ij}
    }{
        \displaystyle
        \sum_{j:D_{ij}=1}
        \left(v_j^{(t)}\right)^2
    },
    \qquad
    i\neq i_0,
\end{gather}
while imposing
\begin{gather*}
    u_{i_0}^{(t+1)}
    =
    1.
\end{gather*}
Given $\boldsymbol{u}^{(t+1)}$, update each firm factor according to
\begin{gather} \label{eq:als_p}
    v_j^{(t+1)}
    \coloneqq
    \frac{
        \displaystyle
        \sum_{i:D_{ij}=1}
        u_i^{(t+1)}\tilde{y}_{ij}
    }{
        \displaystyle
        \sum_{i:D_{ij}=1}
        \left(u_i^{(t+1)}\right)^2
    },
    \qquad
    j=1,\dots,J.
\end{gather}

The updates in Equations~\ref{eq:als_a} and~\ref{eq:als_p} are well defined whenever the corresponding denominators are nonzero. Each update minimizes the criterion with respect to one vector of transformed productivity factors while holding the other fixed and therefore weakly decreases the objective function.

The iterations can be stopped when the change in the objective function or in the transformed productivity factors falls below a chosen tolerance. Because the criterion is nonconvex, the solution reached by the algorithm may depend on the initialization. In practice, the procedure can be run from multiple initial values, retaining the solution that attains the smallest objective value.

A formal analysis of the asymptotic properties of $\hat{\boldsymbol{u}}$ and $\hat{\boldsymbol{v}}$ is beyond the scope of this paper. Such an analysis would characterize the network conditions under which the resulting estimates consistently recover the productivity collections, or relevant functions of them, while accounting for first-stage estimation error in $\hat{\beta}_{L,z}$.

\section{Interpreting the Magnitude of \texorpdfstring{$\beta_0$}{beta0}}
\label{sec:interpretation}

The sign of $\beta_0$ in the Tukey model has a direct interpretation. Because $\beta_0$ is the constant cross-partial derivative of the interaction function, the function is supermodular when $\beta_0>0$, submodular when $\beta_0<0$, and modular when $\beta_0=0$. Its magnitude is less immediate to interpret, but is important for assessing economic significance. Even when the test rejects the null of no complementarities within the Tukey model, one may still ask whether the multiplicative component is quantitatively important relative to the additive component.

This section proposes a simple way to place the magnitude of $\beta_0$ on an interpretable scale. The exercise is purely descriptive: it abstracts from estimation uncertainty and treats $\beta_0$ as given. In applications, the same calculations can be implemented by replacing $\beta_0$ with an estimate.

A useful analogy is the interpretation of a slope coefficient in a linear regression. Consider
\begin{gather*}
    y_i
    =
    a+bx_i+\epsilon_i.
\end{gather*}
The magnitude of $b$ depends on the scales of both the outcome and the regressor. A common scale-free summary is the standardized effect
\begin{gather*}
    \frac{
        b\operatorname{sd}(x_i)
    }{
        \operatorname{sd}(y_i)
    },
\end{gather*}
which measures the change in the outcome, in standard-deviation units, associated with a one-standard-deviation change in the regressor.

An analogous measure in the Tukey model standardizes the multiplicative term. Let
\begin{gather*}
    s_\alpha
    \coloneqq
    \operatorname{sd}(\alpha_i),
    \qquad
    s_\psi
    \coloneqq
    \operatorname{sd}(\psi_j),
    \qquad
    s_y
    \coloneqq
    \operatorname{sd}(y_{ij}).
\end{gather*}
A natural standardized measure is
\begin{gather*}
    \frac{
        \beta_0s_\alpha s_\psi
    }{
        s_y
    }.
\end{gather*}
This measure has a simple interpretation, but requires estimating $s_\alpha s_\psi$, a function of the latent productivity collections. This is substantially more demanding than estimating $\beta_0$. Even in the simpler TWFE model, \citet{kline2020leave} show that consistent estimation of variance components and related functions of fixed effects requires strong connectivity conditions on the matching network that are often not satisfied in applications. Developing an analogous estimation theory for $s_\alpha s_\psi$ in the Tukey model remains an interesting question. For this reason, I do not pursue this standardized measure here.

A second way to interpret the magnitude of a coefficient in the linear model is to decompose a fitted outcome at a benchmark value of the regressor. For example, evaluating the model at the mean of the regressor gives
\begin{gather*}
    \bar{y}
    =
    a+b\bar{x}.
\end{gather*}
The absolute contribution share of the regressor is then
\begin{gather*}
    \frac{
        |b\bar{x}|
    }{
        |a|+|b\bar{x}|
    }.
\end{gather*}
This quantity is unit-free and measures the relative gross magnitude of the term involving $x_i$. The use of absolute values is useful when the two components have opposite signs, because it measures their relative importance before cancellations occur.

I use the same logic to construct a benchmark decomposition in the Tukey model. The idea is to fix an outcome level and measure the relative magnitude of the multiplicative component. I take the target outcome to be the sample mean across observed matches,
\begin{gather*}
    \bar{y}
    \coloneqq
    \frac{1}{|\mathcal{O}_{IJ}|}
    \sum_{(i,j)\in\mathcal{O}_{IJ}}
    y_{ij}.
\end{gather*}
The exercise should not be interpreted as estimating the average structural decomposition of outcomes. Rather, it is a calibration that maps $(\bar{y},\beta_0)$ into a unit-free measure of the relative magnitude of the multiplicative term.

Consider a representative pair whose productivities satisfy
\begin{gather} \label{eq:gamma_benchmark}
    \alpha_i
    =
    \gamma x,
    \qquad
    \psi_j
    =
    (1-\gamma)x,
\end{gather}
where $\gamma\in[0,1]$. The parameter $\gamma$ determines how the representative additive component is divided between the two sides of the interaction. When $\gamma=1/2$, the two sides contribute symmetrically. When $\gamma$ is close to zero or one, most of the additive component is attributed to one side.

Setting the representative outcome equal to $\bar{y}$ gives
\begin{gather} \label{eq:equation}
    \bar{y}
    =
    x^\ast(\gamma)
    +
    \beta_0\gamma(1-\gamma)
    \left[x^\ast(\gamma)\right]^2.
\end{gather}
Define the set of admissible benchmark allocations as
\begin{gather*}
    \Gamma(\bar{y},\beta_0)
    \coloneqq
    \left\{
        \gamma\in[0,1]:
        1+4\beta_0\gamma(1-\gamma)\bar{y}
        \geq
        0
    \right\}.
\end{gather*}
For $\gamma\in\Gamma(\bar{y},\beta_0)$ such that $\beta_0\gamma(1-\gamma)\neq0$, the solution to Equation~\ref{eq:equation} that is continuous as $\beta_0\gamma(1-\gamma)\to0$ is
\begin{gather} \label{eq:solution}
    x^\ast(\gamma)
    =
    \frac{
        -1+
        \sqrt{
            1+4\beta_0\gamma(1-\gamma)\bar{y}
        }
    }{
        2\beta_0\gamma(1-\gamma)
    }.
\end{gather}
When $\beta_0\gamma(1-\gamma)=0$, I define $x^\ast(\gamma)=\bar{y}$ by continuity.

The corresponding additive and multiplicative components are
\begin{gather*}
    A(\gamma)
    \coloneqq
    x^\ast(\gamma),
    \qquad
    M(\gamma)
    \coloneqq
    \beta_0\gamma(1-\gamma)
    \left[x^\ast(\gamma)\right]^2.
\end{gather*}
The absolute contribution share of the multiplicative component is
\begin{align*}
    S_M(\gamma)
    &\coloneqq
    \frac{
        |M(\gamma)|
    }{
        |A(\gamma)|+|M(\gamma)|
    }
    =
    \frac{
        \left|
            \sqrt{
                1+4\beta_0\gamma(1-\gamma)\bar{y}
            }
            -1
        \right|
    }{
        2+
        \left|
            \sqrt{
                1+4\beta_0\gamma(1-\gamma)\bar{y}
            }
            -1
        \right|
    }.
\end{align*}
This quantity measures the gross contribution of the multiplicative term in the representative decomposition. The use of absolute values is deliberate. When $\beta_0<0$, the additive and multiplicative components may have opposite signs and partially offset each other. A signed share may then understate their relative magnitudes. The absolute contribution share instead measures the importance of the two components before such cancellations occur.

The share $S_M(\gamma)$ varies with $\gamma$ because fixing the representative outcome $\bar{y}$ does not uniquely determine how the additive component is divided between the two sides. One can therefore report $S_M(\gamma)$ over the admissible set $\Gamma(\bar{y},\beta_0)$ or focus on a benchmark value. A natural benchmark is $\gamma=1/2$, when admissible, which corresponds to the symmetric case in which the two sides contribute equally to the additive component.

In the empirical application, $\bar{y}=4.106$ and $\hat{\beta}_{L,z}=-0.196$. For these values, $\Gamma(\bar{y},\hat{\beta}_{L,z})=[0,1]$, and Figure~\ref{fig:beta_int} reports $S_M(\gamma)$ over the full range of $\gamma$. The main text reports the symmetric benchmark $\gamma=1/2$, for which
\begin{gather*}
    S_M(1/2)
    =
    0.218.
\end{gather*}
Thus, under the symmetric representative decomposition, the multiplicative component accounts for approximately $22$ percent of the gross magnitude of the decomposition.

\begin{figure}[htbp]
    \centering
    \includegraphics[width=0.7\textwidth]{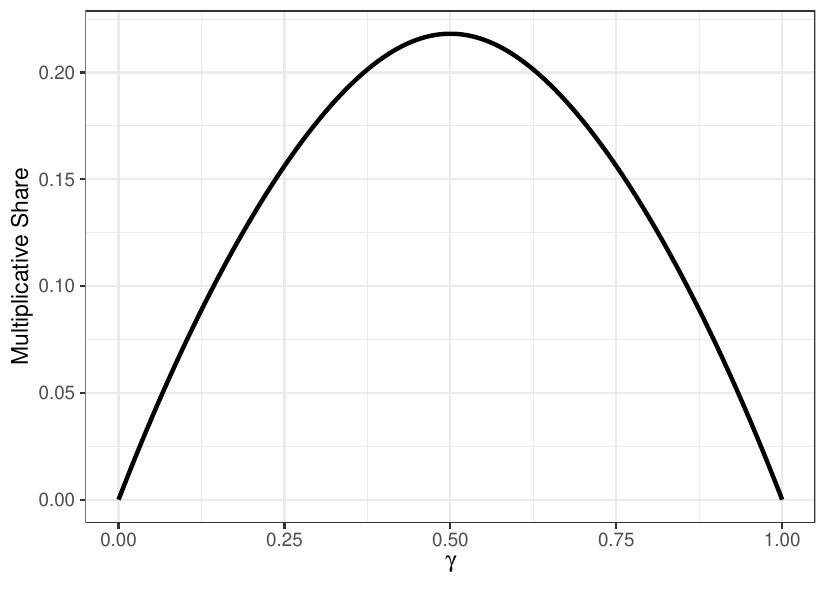}
    \caption{\small Absolute contribution share of the multiplicative component as a function of $\gamma$ in the empirical application.}
    \label{fig:beta_int}
\end{figure}

\section{Monte Carlo Simulations}
\label{sec:sm_mc}

In this section, I investigate the finite-sample behavior of the rank-based estimator $\hat{\beta}_{L,z}$ through Monte Carlo simulations. Because the estimator uses only variation within four-cycles, I abstract from the remainder of the matching network: I fix the number of cycles $L$ and simulate outcomes only for the matches belonging to these cycles. This design isolates how performance depends on the number of cycles, the quality of the orientation rule, and the magnitude of productivity variation relative to outcome noise.

Because the latent productivities are known in the simulations, I control directly whether the orientation rule agrees with their rankings rather than explicitly generating ranking variables. For each cycle, let $r_{\ell,z}^{\alpha}$ and $r_{\ell,z}^{\psi}$ denote the worker-side and firm-side ranking-agreement signs defined in Section~\ref{sec:estimation_inference_beta0}. Their product,
\begin{gather*}
    r_{\ell,z}
    =
    r_{\ell,z}^{\alpha}r_{\ell,z}^{\psi},
\end{gather*}
determines whether the overall orientation agrees with the sign of the productivity-gap product. Given $r_{\ell,z}$, I set
\begin{gather*}
    s_{\ell,z}
    =
    r_{\ell,z}\operatorname{sign}(A_\ell).
\end{gather*}
Because productivities are drawn from a continuous distribution, $A_\ell\neq0$ almost surely. By construction,
\begin{gather*}
    s_{\ell,z}A_\ell
    =
    r_{\ell,z}|A_\ell|.
\end{gather*}

The simulation design proceeds in two steps.

\emph{Step 1. Generation of productivities and orientations.}
For each cycle $\ell=1,\dots,L$, I independently draw
\begin{gather*}
    \left(
        \alpha_{i_{\ell 1}},
        \alpha_{i_{\ell 2}},
        \psi_{j_{\ell 1}},
        \psi_{j_{\ell 2}}
    \right)'
    \sim
    \mathcal{N}(\mu,\Sigma),
\end{gather*}
where
\begin{gather*}
    \mu
    =
    (1,3,1,3)',
    \qquad
    \Sigma
    =
    \begin{bmatrix}
        1   & 0   & 0.5 & 0.5 \\
        0   & 1   & 0.5 & 0.5 \\
        0.5 & 0.5 & 1   & 0 \\
        0.5 & 0.5 & 0   & 1
    \end{bmatrix}.
\end{gather*}

I impose perfect agreement on the firm side by setting
\begin{gather*}
    r_{\ell,z}^{\psi}
    =
    1,
\end{gather*}
and independently draw the worker-side agreement sign as
\begin{gather*}
    r_{\ell,z}^{\alpha}
    =
    \begin{cases}
        1, & \text{with probability } p, \\
        -1, & \text{with probability } 1-p.
    \end{cases}
\end{gather*}
Thus, $p$ controls the quality of the worker-side ranking. An analogous exercise could vary the quality of the firm-side ranking or both rankings simultaneously.

\emph{Step 2. Generation of outcomes.}
For each Monte Carlo replication, I independently draw the four error terms in every cycle from
\begin{gather*}
    \eta_{ij}
    \sim
    \mathcal{N}(0,\sigma_\eta^2).
\end{gather*}
Together with the fixed productivities, orientations, and value of $\beta_0$, these errors generate the $4L$ observed outcomes.

I vary the number of cycles, the common standard deviation of the outcome errors, and the orientation quality:
\begin{gather*}
    L
    \in
    \{100,500,1000,5000\},
    \qquad
    \sigma_\eta
    \in
    \{0.5,1,2\},
    \qquad
    p
    \in
    \{1,0.85,0.65,0.5\}.
\end{gather*}

When $p>0.5$, correctly oriented cycles are more likely than incorrectly oriented cycles, generating a positive orientation signal on average. The case $p=0.5$ represents an uninformative orientation rule. In the corresponding asymptotic design, positive and negative contributions cancel, so $\kappa_{L,z}$ is not bounded away from zero and Assumption~\ref{ass:labeling} fails.

For each combination of $L$, $\sigma_\eta$, and $p$, Step~1 is implemented once, while Step~2 is repeated 10{,}000 times. Thus, the productivities and orientation sequence are held fixed across Monte Carlo replications, while only the outcome errors are redrawn. This design mirrors the conditional asymptotic analysis, in which $\boldsymbol{\alpha}$, $\boldsymbol{\psi}$, $G_{IJ}$, $\boldsymbol{z}^{\alpha}$, and $\boldsymbol{z}^{\psi}$ are treated as fixed.

Table~\ref{tab:mse} reports the mean squared error of $\hat{\beta}_{L,z}$ under $\beta_0=0$. Table~\ref{tab:width} reports the average width of the corresponding nominal 90 percent confidence interval. Table~\ref{tab:rejection_rates} reports rejection rates for the two-sided test of $H_0:\beta_0=0$ at the nominal significance level $\gamma=0.10$, both under the null and under the alternative $\beta_0=1$.

\begin{table}[htbp]
    \centering
    \input{figures_and_tables/tables/tab_mse.tex}
    \caption{\small Mean squared error across values of $\sigma_\eta$, $L$, and $p$ under $\beta_0=0$.}
    \label{tab:mse}
\end{table}

\begin{table}[htbp]
    \centering
    \input{figures_and_tables/tables/tab_width.tex}
    \caption{\small Average width of nominal 90 percent confidence intervals across values of $\sigma_\eta$, $L$, and $p$ under $\beta_0=0$.}
    \label{tab:width}
\end{table}

\begin{table}[htbp]
    \centering
    \input{figures_and_tables/tables/tab_rejection_rates.tex}
    \caption{\small Rejection rates across values of $\sigma_\eta$, $L$, and $p$ for the two-sided test of $H_0:\beta_0=0$ at the nominal significance level $\gamma=0.10$. The left panel reports rejection probabilities under the null $\beta_0=0$, and the right panel reports rejection probabilities under the alternative $\beta_0=1$.}
    \label{tab:rejection_rates}
\end{table}

When $p>0.5$, the mean squared error and confidence interval width generally increase with the outcome error variance and decrease with the number of cycles. The power of the test increases with $L$ and decreases as outcome noise increases. These patterns are consistent with the asymptotic results.

The most important pattern concerns orientation quality. As $p$ approaches $0.5$, the effective signal $|\kappa_{L,z}|$ typically becomes smaller. The ratio estimator consequently becomes more variable, and its confidence intervals become wider. This deterioration is especially pronounced when the number of cycles is small or outcome noise is large.

When $p=0.5$, Assumption~\ref{ass:labeling} fails, and the asymptotic distribution derived in the main text does not apply. The denominator of the estimator is frequently close to zero, producing unstable estimates, large mean squared errors, and highly variable confidence interval widths. The rejection probabilities in this case should therefore not be interpreted as the size or power of an asymptotically valid test.

Overall, the simulations reinforce the importance of the effective orientation signal. A growing number of cycles is not sufficient by itself: the orientation rule must also align the productivity-gap products sufficiently often to prevent their aggregate contribution from canceling. When it does, the estimator becomes more precise as $L$ increases. When it does not, the ratio estimator remains unstable.

\section{Supplementary Proofs}\label{sec:sm_proofs}

\subsection{Proof of Theorem~\ref{thm:hs_identification}}

The proof proceeds in four steps. I first explain the two normalizations. I then show that, after normalization, identification is governed by the rank of a matrix collecting the restrictions generated by worker mobility. Next, I characterize the generic rank of this matrix using Assumption~\ref{ass:hs_connectivity}. I finally combine this characterization with Assumption~\ref{ass:hs_heterogeneity} to establish sufficiency and construct an observationally equivalent parameter collection when the graph condition fails.

I begin with the normalizations. For any $q_1,q_2\in\R$ with $q_1\neq0$, define
\begin{gather*}
    \alpha_i'
    =
    q_1\alpha_i+q_2,
    \qquad
    b_j'
    =
    \frac{b_j}{q_1},
    \qquad
    a_j'
    =
    a_j-\frac{b_jq_2}{q_1}.
\end{gather*}
Then
\begin{align*}
    a_j'+b_j'\alpha_i'
    &=
    a_j
    -
    \frac{b_jq_2}{q_1}
    +
    \frac{b_j}{q_1}
    \left(
        q_1\alpha_i+q_2
    \right) \\
    &=
    a_j+b_j\alpha_i.
\end{align*}
The HS model is therefore invariant to a common location and scale transformation of worker productivity, together with corresponding transformations of the firm parameters. Two normalizations are needed to eliminate these degrees of freedom. Because $b_{j_0}\neq0$, choosing
\begin{gather*}
    q_1
    =
    b_{j_0},
    \qquad
    q_2
    =
    a_{j_0}-1
\end{gather*}
gives $a_{j_0}'=b_{j_0}'=1$.

I next show how identification reduces to a rank condition. Consider two parameter collections
\begin{gather*}
    \left(
        \boldsymbol{\alpha},
        \boldsymbol{a},
        \boldsymbol{b}
    \right)
    \qquad
    \text{and}
    \qquad
    \left(
        \widetilde{\boldsymbol{\alpha}},
        \widetilde{\boldsymbol{a}},
        \widetilde{\boldsymbol{b}}
    \right)
\end{gather*}
that satisfy the HS model, obey the same normalization, and generate the same systematic outcomes. Because all firm slopes are nonzero, define
\begin{gather*}
    x_j
    \coloneqq
    \frac{b_j}{\widetilde{b}_j},
    \qquad
    y_j
    \coloneqq
    \frac{
        a_j-\widetilde{a}_j
    }{
        \widetilde{b}_j
    }.
\end{gather*}
For every observed match,
\begin{gather*}
    a_j+b_j\alpha_i
    =
    \widetilde{a}_j
    +
    \widetilde{b}_j\widetilde{\alpha}_i,
\end{gather*}
which implies
\begin{gather} \label{eq:hs_affine_maps}
    \widetilde{\alpha}_i
    =
    x_j\alpha_i+y_j.
\end{gather}
Thus, any observationally equivalent parameterization can be represented by a firm-specific affine transformation of worker productivity.

If worker $i$ is observed at firms $j$ and $k$, the value of $\widetilde{\alpha}_i$ recovered from Equation~\ref{eq:hs_affine_maps} must be the same at both firms. Therefore,
\begin{gather} \label{eq:hs_worker_comparison}
    \alpha_i
    \left(
        x_j-x_k
    \right)
    +
    y_j-y_k
    =
    0.
\end{gather}
The normalization implies
\begin{gather*}
    x_{j_0}
    =
    1,
    \qquad
    y_{j_0}
    =
    0.
\end{gather*}
Define
\begin{gather*}
    u_j
    \coloneqq
    x_j-1,
    \qquad
    v_j
    \coloneqq
    y_j.
\end{gather*}

For each mobile worker $i\in\mathcal{I}_M$, choose an arbitrary reference firm $r(i)\in N(i)$. Construct the matrix $M_{\mathcal{H}}(\boldsymbol{\alpha})$ with two columns for every firm other than $j_0$. For each $i\in\mathcal{I}_M$ and each $j\in N(i)\setminus\{r(i)\}$, add a row containing
\begin{gather*}
    (\alpha_i,1)
\end{gather*}
in the columns associated with firm $j$,
\begin{gather*}
    -(\alpha_i,1)
\end{gather*}
in the columns associated with firm $r(i)$, and zeros elsewhere. If either firm is $j_0$, omit its two columns. Different choices of $r(i)$ produce row-equivalent matrices and therefore do not affect rank.

Let $\boldsymbol{\delta}$ collect $(u_j,v_j)$ across all $j\neq j_0$. Stacking Equation~\ref{eq:hs_worker_comparison} across the mobile workers gives
\begin{gather} \label{eq:hs_rank_system}
    M_{\mathcal{H}}(\boldsymbol{\alpha})
    \boldsymbol{\delta}
    =
    \boldsymbol{0}.
\end{gather}

Equation~\ref{eq:hs_rank_system} shows why the rank of $M_{\mathcal{H}}(\boldsymbol{\alpha})$ governs identification. If
\begin{gather} \label{eq:hs_actual_full_rank}
    \operatorname{rank}
    \left(
        M_{\mathcal{H}}(\boldsymbol{\alpha})
    \right)
    =
    2J-2,
\end{gather}
then Equation~\ref{eq:hs_rank_system} implies $\boldsymbol{\delta}=\boldsymbol{0}$. Hence,
\begin{gather*}
    x_j
    =
    1,
    \qquad
    y_j
    =
    0
\end{gather*}
for every firm, and therefore
\begin{gather*}
    \widetilde{a}_j
    =
    a_j,
    \qquad
    \widetilde{b}_j
    =
    b_j.
\end{gather*}
Equation~\ref{eq:hs_affine_maps} then gives
\begin{gather*}
    \widetilde{\alpha}_i
    =
    \alpha_i
\end{gather*}
for every mobile worker. For each nonmobile worker, equality of the corresponding firm parameters and systematic outcome identifies $\alpha_i$ because the firm's slope is nonzero. Thus, full column rank implies identification.

Conversely, rank deficiency generates nonidentification. Suppose that $M_{\mathcal{H}}(\boldsymbol{\alpha})$ has a nonzero null vector $\boldsymbol{\delta}$, with associated components $(u_j,v_j)$ and $u_{j_0}=v_{j_0}=0$. For $t\in\R$, define
\begin{gather*}
    x_j(t)
    \coloneqq
    1+tu_j,
    \qquad
    y_j(t)
    \coloneqq
    tv_j.
\end{gather*}
For sufficiently small $t$, $x_j(t)\neq0$ for every firm. Because $\boldsymbol{\delta}$ belongs to the kernel, for every mobile worker $i$,
\begin{gather*}
    x_j(t)\alpha_i+y_j(t)
\end{gather*}
takes the same value at every firm $j\in N(i)$. Define this common value as $\alpha_i(t)$. For a nonmobile worker, define $\alpha_i(t)$ using the worker's unique firm.

Now set
\begin{gather*}
    b_j(t)
    \coloneqq
    \frac{b_j}{x_j(t)},
    \qquad
    a_j(t)
    \coloneqq
    a_j-b_j(t)y_j(t).
\end{gather*}
For every observed match,
\begin{align*}
    a_j(t)+b_j(t)\alpha_i(t)
    &=
    a_j
    -
    \frac{b_j}{x_j(t)}y_j(t)
    +
    \frac{b_j}{x_j(t)}
    \left[
        x_j(t)\alpha_i+y_j(t)
    \right] \\
    &=
    a_j+b_j\alpha_i.
\end{align*}
The normalization is preserved because
\begin{gather*}
    x_{j_0}(t)
    =
    1,
    \qquad
    y_{j_0}(t)
    =
    0.
\end{gather*}
For sufficiently small $t$, the transformed slopes remain nonzero and the transformed parameter collection remains admissible. If the initial mobile-worker productivities satisfy Assumption~\ref{ass:hs_heterogeneity}, this condition also continues to hold for all sufficiently small $t$, because a maximal-rank minor that is nonzero at $t=0$ remains nonzero in a neighborhood of zero. Since $\boldsymbol{\delta}\neq\boldsymbol{0}$, every sufficiently small $t\neq0$ produces a distinct parameter collection with the same systematic outcomes. Rank deficiency therefore generates nonidentification.

The preceding argument establishes the algebraic link between rank and identification. I now separate the roles of the matching structure and the productivity values. Replace each mobile-worker productivity $\alpha_i$ in $M_{\mathcal{H}}(\boldsymbol{\alpha})$ with an indeterminate $X_i$, and denote the resulting symbolic matrix by $M_{\mathcal{H}}(\boldsymbol{X})$. Its rank over $\mathbb{Q}(\boldsymbol{X})$ is the generic rank permitted by the firm-mobility hypergraph.

I next show that
\begin{gather} \label{eq:hs_generic_rank}
    \operatorname{rank}
    \left(
        M_{\mathcal{H}}(\boldsymbol{X})
    \right)
    =
    2J-2
\end{gather}
if and only if Assumption~\ref{ass:hs_connectivity} holds.

Suppose first that Assumption~\ref{ass:hs_connectivity} holds. This assumption is the $2$-weak partition-connectivity condition of \citet{alrabiah2025random}. With firms as vertices, mobile-worker histories as hyperedges, and $k=2$, the matrix $M_{\mathcal{H}}(\boldsymbol{X})$ coincides, up to row signs, row and column order, the order of the two columns within each firm block, and the choice of reference firm within each hyperedge, with their reduced intersection matrix with Vandermonde rows
\begin{gather*}
    (1,X_i).
\end{gather*}
Theorem~2.11 of \citet{alrabiah2025random} implies that this reduced intersection matrix has full column rank over the corresponding rational-function field. The determinant of each maximal minor is a polynomial with integer coefficients, so a minor that is nonzero over that field is also not identically zero over $\mathbb{Q}$. Equation~\ref{eq:hs_generic_rank} therefore holds over $\mathbb{Q}(\boldsymbol{X})$.

Conversely, suppose Assumption~\ref{ass:hs_connectivity} fails. Then there exists a partition $\mathcal{P}$ of the firms such that
\begin{gather*}
    \sum_{i\in\mathcal{I}_M}
    \left[
        q_i(\mathcal{P})-1
    \right]
    <
    2
    \left(
        |\mathcal{P}|-1
    \right).
\end{gather*}
Consider vectors whose two entries are constant across all firms belonging to the same element of $\mathcal{P}$. After setting the two entries associated with the element containing $j_0$ equal to zero, this subspace has dimension
\begin{gather*}
    2
    \left(
        |\mathcal{P}|-1
    \right).
\end{gather*}
For worker $i$, requiring the transformations to agree across the $q_i(\mathcal{P})$ groups intersecting $N(i)$ imposes at most
\begin{gather*}
    q_i(\mathcal{P})-1
\end{gather*}
independent restrictions on this subspace. The total number of restrictions is therefore smaller than its dimension. Hence, the kernel contains a nonzero vector, and
\begin{gather*}
    \operatorname{rank}
    \left(
        M_{\mathcal{H}}(\boldsymbol{X})
    \right)
    <
    2J-2.
\end{gather*}
This establishes the equivalence between Assumption~\ref{ass:hs_connectivity} and generic full column rank.

I can now complete the proof. Suppose first that Assumption~\ref{ass:hs_connectivity} holds. The generic-rank characterization gives
\begin{gather*}
    \operatorname{rank}
    \left(
        M_{\mathcal{H}}(\boldsymbol{X})
    \right)
    =
    2J-2.
\end{gather*}
Assumption~\ref{ass:hs_heterogeneity} states that the actual mobile-worker productivities attain this generic rank. Therefore, Equation~\ref{eq:hs_actual_full_rank} holds, and the rank-identification argument above implies that $\boldsymbol{\alpha}$, $\boldsymbol{a}$, and $\boldsymbol{b}$ are identified. This proves sufficiency.

Finally, suppose Assumption~\ref{ass:hs_connectivity} fails. The partition argument above provides a nonzero null vector not only for the symbolic matrix, but for $M_{\mathcal{H}}(\boldsymbol{\alpha})$ at every collection of mobile-worker productivities. Choose an admissible normalized parameter collection satisfying Assumption~\ref{ass:hs_heterogeneity}. The rank-deficiency construction above then produces a distinct normalized parameter collection, also satisfying Assumption~\ref{ass:hs_heterogeneity}, that generates the same systematic outcomes. The parameters are therefore not identified. This proves necessity.

\subsection{Proof of Proposition~\ref{prop:isotonic_monotone_equivalence}}

Let $(f_m,\boldsymbol{\alpha},\boldsymbol{\psi},G_{IJ})$ be any representation satisfying
\begin{gather*}
    \theta_{ij}
    =
    f_m(\alpha_i,\psi_j).
\end{gather*}
Consider any admissible strictly increasing bijections $g_\alpha\colon A\to A'$ and $g_\psi\colon\Psi\to\Psi'$, and define
\begin{gather*}
    \alpha_i'
    =
    g_\alpha(\alpha_i),
    \qquad
    \psi_j'
    =
    g_\psi(\psi_j),
    \qquad
    f_m'(\alpha,\psi)
    =
    f_m\left(g_\alpha^{-1}(\alpha),g_\psi^{-1}(\psi)\right).
\end{gather*}
Because $g_\alpha^{-1}$ and $g_\psi^{-1}$ are strictly increasing, $f_m'$ is strictly increasing in each argument and therefore belongs to the isotonic model. Moreover,
\begin{align*}
    f_m'(\alpha_i',\psi_j')
    &=
    f_m\left(
        g_\alpha^{-1}\left(g_\alpha(\alpha_i)\right),
        g_\psi^{-1}\left(g_\psi(\psi_j)\right)
    \right) \\
    &=
    f_m(\alpha_i,\psi_j)
    =
    \theta_{ij}.
\end{align*}
Hence, $(f_m,\boldsymbol{\alpha},\boldsymbol{\psi},G_{IJ})$ and $(f_m',\boldsymbol{\alpha}',\boldsymbol{\psi}',G_{IJ})$ generate the same collection of observed deterministic outcomes. Therefore, the isotonic model is invariant to separate strictly increasing reparameterizations of the worker and firm productivities.

\subsection{Proof of Theorem~\ref{thm:isotonic_identification}}

I prove the result for the ranking of $\boldsymbol{\alpha}$; the argument for $\boldsymbol{\psi}$ is analogous.

I first establish sufficiency. Fix any two workers $i$ and $i'$. By Assumption~\ref{ass:diameter}, there exists a firm $j$ such that $(i,j),(i',j)\in\mathcal{O}_{IJ}$. Because $f_m$ is strictly increasing in its first argument, the sign of $\theta_{ij}-\theta_{i'j}$ reveals whether $\alpha_i$ is greater than, equal to, or smaller than $\alpha_{i'}$. Repeating this comparison for every pair $(i,i')$ identifies the ranking of the workers.

I next establish necessity. Suppose that Assumption~\ref{ass:diameter} fails on the worker side. Then there exist two distinct workers $i$ and $i'$ that do not share any firm. Consider an admissible representation in which $i$ and $i'$ are adjacent in the worker ranking, with $\alpha_i>\alpha_{i'}$, and all other workers have lower productivity. Order the firms so that every firm matched with $i$ has higher productivity than every firm matched with $i'$, and choose the observed outcomes so that
\begin{gather*}
    \theta_{ij}
    >
    \theta_{i'j'}
    \qquad
    \text{for every }
    (i,j),(i',j')\in\mathcal{O}_{IJ}.
\end{gather*}
These inequalities can be chosen with strict slack so that the observed entries admit a completion to a matrix that is strictly increasing along both rows and columns.

Because $i$ and $i'$ have no common firm, no column contains observed outcomes for both workers. Their positions in the worker ranking can therefore be reversed while leaving every observed entry unchanged. In each column matched with one of the two workers, the entry corresponding to the other worker is unobserved and can be chosen, together with the remaining missing entries, to preserve strict monotonicity along rows and columns. The resulting completion defines a second admissible isotonic representation that generates the same observed outcomes but satisfies $\alpha_i'<\alpha_{i'}'$.

Thus, the relative order of $i$ and $i'$ is not identified whenever two workers have no common firm. The analogous argument applies when two firms have no common worker. Assumption~\ref{ass:diameter} is therefore necessary and sufficient for identification of the rankings of $\boldsymbol{\alpha}$ and $\boldsymbol{\psi}$.

\subsection{Proof of Proposition~\ref{prop:ER}}

Let $\mathcal{I}_{\ell}$ be the indicator that the quartet $\ell=(i,i',j,j')$, with $i,i'\in[I]$ and $j,j'\in[J]$, forms a four-cycle, and let $\mathcal{L}$ denote the collection of all potential quartets. The number of four-cycles is
\begin{gather*}
    C_4(I,J)
    =
    \sum_{\ell\in\mathcal{L}}\mathcal{I}_{\ell}.
\end{gather*}
There are $\binom{I}{2}\binom{J}{2}$ potential four-cycles, each of which occurs with probability $p_{IJ}^4$. Therefore,
\begin{gather*}
    \E[C_4(I,J)]
    =
    \binom{I}{2}\binom{J}{2}p_{IJ}^4
    =
    \frac{I(I-1)J(J-1)}{4}p_{IJ}^4.
\end{gather*}

Consider the variance:
\begin{align*}
    \operatorname{Var}(C_4(I,J))
    &=
    \sum_{\ell\in\mathcal{L}}
    \operatorname{Var}(\mathcal{I}_{\ell})
    +
    \sum_{\ell\neq\ell'}
    \operatorname{Cov}(\mathcal{I}_{\ell},\mathcal{I}_{\ell'}).
\end{align*}
Because $\mathcal{I}_{\ell}^2=\mathcal{I}_{\ell}$,
\begin{align*}
    \sum_{\ell\in\mathcal{L}}
    \operatorname{Var}(\mathcal{I}_{\ell})
    &=
    \binom{I}{2}\binom{J}{2}
    p_{IJ}^4(1-p_{IJ}^4) \\
    &=
    \mathcal{O}\left(I^2J^2p_{IJ}^4\right).
\end{align*}

The covariance between two distinct cycle indicators is zero unless the cycles share at least one edge. Two distinct four-cycles can share either one or two edges.

If they share one edge, their union contains seven distinct edges. The number of such pairs is of order $\mathcal{O}(I^3J^3)$, and
\begin{gather*}
    \operatorname{Cov}(\mathcal{I}_{\ell},\mathcal{I}_{\ell'})
    =
    p_{IJ}^7-p_{IJ}^8
    =
    \mathcal{O}(p_{IJ}^7).
\end{gather*}
Their total contribution to the variance is therefore
\begin{gather*}
    \mathcal{O}\left(I^3J^3p_{IJ}^7\right).
\end{gather*}

If they share two edges, their union contains six distinct edges. Such cycles share either one worker and two firms or one firm and two workers, so the number of pairs is of order
\begin{gather*}
    \mathcal{O}\left(I^3J^2+I^2J^3\right).
\end{gather*}
Moreover,
\begin{gather*}
    \operatorname{Cov}(\mathcal{I}_{\ell},\mathcal{I}_{\ell'})
    =
    p_{IJ}^6-p_{IJ}^8
    =
    \mathcal{O}(p_{IJ}^6).
\end{gather*}
The corresponding contribution is therefore
\begin{gather*}
    \mathcal{O}\left(
        \left(I^3J^2+I^2J^3\right)p_{IJ}^6
    \right).
\end{gather*}

Combining these terms gives
\begin{align*}
    \operatorname{Var}(C_4(I,J))
    &=
    \mathcal{O}\left(I^2J^2p_{IJ}^4\right)
    +
    \mathcal{O}\left(I^3J^3p_{IJ}^7\right) \\
    &\qquad
    +
    \mathcal{O}\left(
        \left(I^3J^2+I^2J^3\right)p_{IJ}^6
    \right).
\end{align*}
Because $I/J$ is bounded away from zero and infinity,
\begin{align*}
    \frac{
        \operatorname{Var}(C_4(I,J))
    }{
        \{\E[C_4(I,J)]\}^2
    }
    &=
    \mathcal{O}\left(
        \frac{1}{I^2J^2p_{IJ}^4}
    \right)
    +
    \mathcal{O}\left(
        \frac{1}{IJp_{IJ}}
    \right) \\
    &\qquad
    +
    \mathcal{O}\left(
        \frac{I+J}{I^2J^2p_{IJ}^2}
    \right)
    \to
    0
\end{align*}
whenever $\sqrt{IJ}p_{IJ}\to\infty$.

For any $\epsilon>0$, Chebyshev's inequality therefore yields
\begin{align*}
    \Pr\left\{
        \left|
            \frac{C_4(I,J)}{\E[C_4(I,J)]}-1
        \right|
        \geq\epsilon
    \right\}
    &\leq
    \frac{
        \operatorname{Var}(C_4(I,J))
    }{
        \epsilon^2\{\E[C_4(I,J)]\}^2
    }
    \to
    0.
\end{align*}
Hence,
\begin{gather*}
    \frac{C_4(I,J)}{\E[C_4(I,J)]}
    \xrightarrow{p}
    1.
\end{gather*}
Finally, $\sqrt{IJ}p_{IJ}\to\infty$ implies
\begin{gather*}
    \E[C_4(I,J)]
    \to
    \infty,
\end{gather*}
and thus
\begin{gather*}
    C_4(I,J)
    \xrightarrow{p}
    \infty.
\end{gather*}

\putbib
\end{bibunit}
\endgroup

%% file: figures_and_tables/tables/tab_mse.tex
\begin{tabular}{ccrrrrr}
\hline \hline
\multicolumn{2}{c}{} & & \multicolumn{4}{c}{Mean squared errors} \\
\multicolumn{2}{c}{} & & \multicolumn{4}{c}{$p$} \\
\cline{4-7}
\textbf{$\sigma_{ij}$} & $L$ & & 1 & 0.85 & 0.65 & 0.5 \\
\hline
$0.5$ & $100$ & & $0.0004$ & $0.001$ & $0.068$ & $104.760$ \\
$0.5$ & $500$ & & $0.0001$ & $0.0003$ & $0.001$ & $209.599$ \\
$0.5$ & $1000$ & & $0.0001$ & $0.0001$ & $0.001$ & $202.407$ \\
$0.5$ & $5000$ & & $0$ & $0$ & $0.0001$ & $16.723$ \\
\hline
$1$ & $100$ & & $0.002$ & $0.470$ & $9.173$ & $225.307$ \\
$1$ & $500$ & & $0.0005$ & $0.001$ & $30.133$ & $15.068$ \\
$1$ & $1000$ & & $0.0002$ & $0.0005$ & $0.005$ & $371.678$ \\
$1$ & $5000$ & & $0$ & $0.0001$ & $0.0005$ & $390.308$ \\
\hline
$2$ & $100$ & & $0.348$ & $53.365$ & $586.414$ & $371.702$ \\
$2$ & $500$ & & $0.003$ & $0.018$ & $229.784$ & $101.623$ \\
$2$ & $1000$ & & $0.001$ & $0.002$ & $5.009$ & $165.683$ \\
$2$ & $5000$ & & $0.0002$ & $0.0003$ & $0.003$ & $633.327$ \\
\hline \hline
\end{tabular}

%% file: figures_and_tables/tables/tab_width.tex
\begin{tabular}{ccrrrrr}
\hline \hline
\multicolumn{2}{c}{} & & \multicolumn{4}{c}{Average confidence interval width} \\
\multicolumn{2}{c}{} & & \multicolumn{4}{c}{$p$} \\
\cline{4-7}
\textbf{$\sigma_{ij}$} & $L$ & & 1 & 0.85 & 0.65 & 0.5 \\
\hline
$0.5$ & $100$ & & $0.070$ & $0.110$ & $0.890$ & $2330.210$ \\
$0.5$ & $500$ & & $0.030$ & $0.060$ & $0.100$ & $8171.620$ \\
$0.5$ & $1000$ & & $0.020$ & $0.030$ & $0.080$ & $5022.450$ \\
$0.5$ & $5000$ & & $0.010$ & $0.020$ & $0.040$ & $342.600$ \\
\hline
$1$ & $100$ & & $0.150$ & $2.740$ & $84.360$ & $9941.750$ \\
$1$ & $500$ & & $0.070$ & $0.110$ & $553.890$ & $203.280$ \\
$1$ & $1000$ & & $0.050$ & $0.080$ & $0.200$ & $5879.220$ \\
$1$ & $5000$ & & $0.020$ & $0.030$ & $0.080$ & $16986.770$ \\
\hline
$2$ & $100$ & & $2.980$ & $730.570$ & $15752.510$ & $7345.480$ \\
$2$ & $500$ & & $0.160$ & $0.310$ & $2214.970$ & $2782.270$ \\
$2$ & $1000$ & & $0.100$ & $0.150$ & $34.490$ & $3163.600$ \\
$2$ & $5000$ & & $0.040$ & $0.060$ & $0.160$ & $11696.520$ \\
\hline \hline
\end{tabular}

%% file: figures_and_tables/tables/tab_rejection_rates.tex
\begin{tabular}{ccrrrrrrrrrr}
\hline \hline
\multicolumn{3}{c}{} & \multicolumn{4}{c}{Rejection rate under $H_0$} &  & \multicolumn{4}{c}{Rejection rate under $H_1$} \\
\multicolumn{3}{c}{} & \multicolumn{4}{c}{$p$}  &   & \multicolumn{4}{c}{$p$} \\
\cline{4-7} \cline{9-12}
\textbf{$\sigma_{ij}$} & $L$ & & 1 & 0.85 & 0.65 & 0.5 & & 1 & 0.85 & 0.65 & 0.5 \\
\hline
$0.5$ & $100$ & & $0.098$ & $0.077$ & $0.084$ & $0.046$ & & $0.983$ & $0.895$ & $0.020$ & $0$ \\
$0.5$ & $500$ & & $0.102$ & $0.097$ & $0.080$ & $0.028$ & & $1$ & $1$ & $0.909$ & $0$ \\
$0.5$ & $1000$ & & $0.102$ & $0.098$ & $0.076$ & $0.040$ & & $1$ & $1$ & $0.995$ & $0$ \\
$0.5$ & $5000$ & & $0.105$ & $0.100$ & $0.102$ & $0.069$ & & $1$ & $1$ & $1$ & $0$ \\
\hline
$1$ & $100$ & & $0.093$ & $0.109$ & $0.113$ & $0.112$ & & $0.572$ & $0.409$ & $0.071$ & $0$ \\
$1$ & $500$ & & $0.091$ & $0.085$ & $0.118$ & $0.139$ & & $0.997$ & $0.916$ & $0.267$ & $0$ \\
$1$ & $1000$ & & $0.094$ & $0.090$ & $0.093$ & $0.132$ & & $1$ & $0.997$ & $0.650$ & $0$ \\
$1$ & $5000$ & & $0.099$ & $0.100$ & $0.089$ & $0.122$ & & $1$ & $1$ & $0.994$ & $0$ \\
\hline
$2$ & $100$ & & $0.116$ & $0.133$ & $0.189$ & $0.192$ & & $0.173$ & $0.086$ & $0.014$ & $0.005$ \\
$2$ & $500$ & & $0.094$ & $0.102$ & $0.143$ & $0.188$ & & $0.695$ & $0.482$ & $0.078$ & $0.002$ \\
$2$ & $1000$ & & $0.088$ & $0.086$ & $0.116$ & $0.168$ & & $0.930$ & $0.681$ & $0.179$ & $0.002$ \\
$2$ & $5000$ & & $0.093$ & $0.097$ & $0.092$ & $0.187$ & & $1$ & $1$ & $0.608$ & $0.024$ \\
\hline \hline
\end{tabular}